\documentclass[12pt]{article}
\usepackage{epsfig} 
\usepackage{amsmath,amsthm, amssymb, latexsym, bm}
\usepackage{lipsum}
\usepackage{color,soul}
\usepackage{xspace}
\usepackage{algorithm}
\usepackage{algpseudocode}
\usepackage{graphicx}
\usepackage{natbib}
\usepackage{multirow}
\usepackage{subcaption}
\usepackage{setspace}        
\usepackage[hang]{footmisc}
\usepackage[absolute,overlay]{textpos}
\usepackage{xcolor}
\usepackage{caption}
\usepackage{enumitem}
\usepackage[final]{microtype}

\newtheoremstyle{shortspace}
 {2pt}
  {2pt}
  {\itshape} 
  {} 
  {\bfseries} 
  {.} 
  {5pt plus 1pt minus 1pt} 
  {} 

\theoremstyle{shortspace}
\allowdisplaybreaks
\algnewcommand\algorithmicforeach{\textbf{for each}}
\algdef{S}[FOR]{ForEach}[1]{\algorithmicforeach\ #1\ \algorithmicdo}
\algrenewcommand\algorithmicrequire{\textbf{Input}}
\algrenewcommand\algorithmicensure{\textbf{Output}}
\newtheorem{theorem}{Theorem}
\newtheorem{lemma}{Lemma}
\newtheorem{corollary}{Corollary}

\newtheorem{example}{\quad Example}
\usepackage{enumitem}

\def\*#1{\bm{#1}}
\def\btheta{\mathop{\theta\kern-.45em\hbox{$\theta$}}\nolimits}

\def\bfeta{\mathop{\eta\kern-.5em\hbox{$\eta$}}\nolimits}
\def\bxi{\mathop{\xi\kern-.45em\hbox{$\xi$}}\nolimits}

\def\P{\mathbb P}
\def\E{\mathbb E}
\def\A{\mathbf a}
\def\B{\mathbf b}
\def\C{\mathbf c}
\newcommand{\RR}{\mathbb{R}}
\newcommand{\R}{\mathbb{R}}

\newcommand{\diag}{\hbox{diag}}

\newcommand\numberthis{\addtocounter{equation}{1}\tag{\theequation}}

\expandafter\def\expandafter\normalsize\expandafter{%
    \normalsize%
    \setlength\abovedisplayskip{1.5pt}%
    \setlength\belowdisplayskip{2pt}%
    \setlength\abovedisplayshortskip{-8pt}%
    \setlength\belowdisplayshortskip{-2pt}%
}

\usepackage{titlesec}

\titleformat{\section}
  {\normalfont\normalsize\bfseries}
  {\thesection}{1em}{}

\titleformat{\subsection}
  {\normalfont\normalsize\bfseries}
  {\thesubsection}{1em}{}

\titlespacing*{\section}{0pt}{0pt}{0pt}
\titlespacing*{\subsection}{0pt}{0pt}{0pt}

\usepackage{lipsum}

\let\OLDthebibliography\thebibliography
\renewcommand\thebibliography[1]{
  \OLDthebibliography{#1}
  \setlength{\parskip}{1pt}
  \setlength{\itemsep}{1.5pt plus 0.3ex}
}

\title{\large
Repro Samples Method for a Performance Guaranteed
Inference in General and Irregular Inference Problems}
\author{Minge Xie and Peng Wang\footnote{Min-ge Xie is a Distinguished Professor, Department of Statistics, Rutgers, The State University of New Jersey, Piscataway, NJ 08854. Email: mxie@stat.rutgers.edu. Peng Wang is an Associate Professor, Department of Operations, Business Analytics and Information System,  University of Cincinnati, Cincinnati, OH 45221. Email: wangp9@ucmail.uc.edu.  The research is supported in part by NSF grants DMS2015373, DMS2027855, 
 NSF-DMS2311064, NSF-DMS2319260, NSF-DMS2515766 and NIH-1R01GM157610-1.}}
\date{}

\begin{document}
\def\spacingset#1{\renewcommand{\baselinestretch}%
{#1}\small\normalsize} \spacingset{1}

\maketitle

\thispagestyle{empty}
\begin{abstract} 
Rapid advances in data science require fundamentally new approaches to address prevalent inference problems for which regularity conditions and the large-sample~central limit theorem do not apply, such as those involving discrete or non-numerical~parameters or non-numerical data. This article presents an innovative and effective framework, called {\it repro samples method}, to conduct statistical inference in general settings, including the complex problems described above. The development leverages Fisher inversion techniques and artificial samples generated to mimic the observed data, and is broadly applicable and supported by rigorous theories. 
For commonly encountered irregular inference problems involving mixed types of (discrete/non-numerical and continuous) parameters, a three-step procedure is proposed to dissect the complicated inference task into manageable steps. The article also develops a unique matching scheme that turns parameter discreteness from an obstacle for forming inferential theories into an advantage for improving computational efficiency. The effectiveness of the proposed methodology is demonstrated using examples, including a case study on inference for a Gaussian mixture model with an unknown number of components that resolves a long-standing 
inference problem. Real data and simulation studies, with comparisons to existing approaches, demonstrate the superior performance of the proposed method. 
\end{abstract}

\vspace{-1.5mm}

\spacingset{2}

\newpage
 \setcounter{page}{1}

\vspace{-2mm}
\section{
Introduction}\label{sec:intr}
\vspace{-1mm}

The approaches of 
creating 
artificial data to
help assess uncertainty  
for inference have proven to be an effective method in the literature; 
see, e.g., \citet{EfroTibs93,
fishman1996, 
Robert2016, Hannig2016, dalmasso2022}, 
 and references therein. Except for a few recent works, 
most of the approaches, however, still rely on large-sample central limit theorem (CLT)
to justify their inference validity.
Their applicability to more complex problems, especially those involving discrete or non-numerical parameters and non-numerical data 
is limited. 
In this article, we 
develop a fundamentally new 
and wide-reaching artificial-sample-inspired inferential framework, with 
supporting theories, 
for these complex problems and more.
The 
framework, called {\it repro samples method}, does not need to rely on CLT or likelihood functions.
It is especially effective for difficult {\it irregular inference problems}.
Here, following \cite{wasserman2020universal}, the irregular problems refer to those ``highly non-trivial''  and ``complex'' setups to which the regularity assumptions for large-sample theories do not apply. 
Although our focus is mostly on finite-sample inference, the 
proposed framework
has direct extensions to the situations where the large-sample theorem holds as well. 

Suppose sample data $\*Y \in {\mathcal Y}$ 
are generated from a model (distribution) 
$
\*Y | {\btheta} \sim F_{{\bf \theta}}(\cdot),$ 
where
${\btheta} \in \Theta$ is model parameters, which can either be numerical, non-numerical, or a mixture of both, or even functions or model indices. 
In most publications $\*Y$ typically have $n$ data points $\*Y= (Y_1, \ldots, Y_n)^\top$, but we 
allow $n=1$ or also $\*Y$ to be a summary or function of $\*Y$; they can also be non-numerical types (e.g., image; voice; object; etc.). 
We assume that we know how to simulate $\*Y$ from $F_{\bf \theta}(\cdot)$, given $\btheta$. 
Often,  $\*Y | {\btheta} \sim F_{{\bf \theta}}(\cdot)$ 
can be re-expressed as:
in the form of {\it a generative 
model}, which we adopt in this paper:
\begin{equation}\label{eq:1}
\*Y = G({\btheta}, \*U), 
\end{equation} 
where $G(\cdot, \cdot)$ is a known mapping from $\Theta \times {\mathcal U} \mapsto {\mathcal Y}$ and $\*U = (U_1, \ldots U_r)^\top \in {\mathcal U}$,  for some $r >0$, 
is a random vector 
whose distribution, say $D_U(\cdot)$, is free of $\*\theta$.

The generative model  (\ref{eq:1})  is frequently used in modern data science literature  \citep[e.g.,][among others]{dalmasso2022} and very general. It is natural for problems where $\*Y$~are generated 
(simulated) 
through an algorithm. The classical statistical model specification~using a family of (parametric)
density functions $Y_i \sim f_\theta(y)$,  for $i = 1, \ldots, n$, can also be re-expressed in this form $Y_i = G(\theta, U_i)$, with $G(\theta, \cdot) = F_\theta^{-1}(\cdot)$, $F_\theta(\cdot)$ being the cumulative distribution function  and $U_i \sim U(0,1)$. 
Additionally, for model distributions with a corresponding
probability measure on Borel spaces, it is always possible to separate the model parameters from the stochastic component --- By the Borel isomorphism theorem \cite[][Theorem 4.19]{kallenberg1997foundations}, any probability measure on a Borel space can be represented by a random element on $[0,1]$, with the corresponding distribution generated from a $U(0,1)$ random variable. Thus, for any model with a corresponding probability measure
on a Borel space, its sample random variable(s) $\*Y$  can be expressed as  
$
\*Y = \Psi(U),
$  
where $U \sim U(0,1)$,   $\Psi(\cdot)$ is a deterministic measurable bijection from $[0,1]$ to $\mathcal{Y}$, and the randomness of $\*Y$ is entirely represented by the random component $U$. The model parameters  are typically  characteristics (features) of the mapping function $\Psi$,  
$
\btheta = \btheta(\Psi),
$ 
which can be any types and may even be $\Psi$ itself.
In this paper, for notational convenience and to facilitate 
the inversion techniques to be used, we explicitly separate the model parameters and the random components by rewriting $\*Y = \Psi(U)$ as  
$
 \*Y = G(\mathbf{\theta}, \*U),
$  
where, to better accommodate practical use, we also allow $\*U  \in {\mathcal U} \subset \RR^r$ to be a $r \times 1$ random vector from a known distribution 
on a probability space $(\Omega, {\mathcal F}, \P)$. 
For ease of discussion, we further assume that the target model parameter $\btheta = \btheta(\Psi)$ is well-defined and identifiable, and that the sample data $\*Y$ provide enough information to make inference about $\btheta$. An additional discussion on model identifiability~is~in~Section~\ref{sec:6}.

Let  
$\*y_{obs} = G({\btheta}_0, \*u^{rel})$
be the realized observation with true parameter value $\*\theta_0$ and 
$\*u^{rel}$ be the corresponding (unobserved) 
realization of $\*U$. 
The repro samples framework
uses a simple yet fundamental idea: study  artificial samples 
obtained by mimicking the sampling mechanism; 
then use
them to help
quantify inference uncertainty.
Particularly, 
for any $\btheta \in \Theta$ and 
a copy of artificial  $\*u^* \sim \*U$, we can use 
(\ref{eq:1}) to get 
an artificial sample $\*y^* = G({\btheta}, \*u^*)$, 
which we refer to as a {\it repro sample}.
If $\btheta = \btheta_0$ and $\*u^*$ matches $\*u^{rel}$, the corresponding $\*y^*$  matches $\*y_{obs}$.
Inversely, 
if $\*y_{obs} = \*y^*$,  an artificial sample generated using $\btheta$ and a likely  realization $\*u^*$ from $\*U$, 
then we cannot dismiss the possibility that such a $\btheta$ equals $\btheta_0$.

The above (inversion) idea can be traced back to R.A. Fisher in which one can solve a model equation to obtain $\*\theta_0$ if both $\*y_{obs}$ and $\*u^{rel}$ were completely given (known). However, Fisher's continuously treating the unknown realization $\*u^{rel}$ as random leads to unsolvable obstacles and his version of fiducial inference being viewed as his ``biggest blunder'' 
(e.g., \citealt{Efron1998, Hannig2016, Thornton2022}).   
In this paper, we will use two inversion ideas and artificial samples $\*u^*$ or $\*y^*$ to develop a general {\it frequentist} repro samples framework to construct level-$(1 -\alpha)$ confidence sets, $\alpha \in (0,1)$, for parameters of interest.
The development covers both regular and irregular inference problems under both finite and large-sample settings.
We show that the constructed confidence sets 
have a theoretically guaranteed frequency coverage rate. We also show that, for any confidence~set constructed by inverting a Neyman-Pearson test, 
the proposed method can always ~construct~either the same~or a better confidence set. 
We further provide discussions on handling nuisance~parameters.

We specifically 
focus on complex and difficult inference problems that involve mixed types of parameters $\btheta = (\eta, {\*\beta}^\top)^\top$, where $\eta$ are discrete or non-numerical and, when given $\eta$, the remaining parameters $\*\beta$ may 
be handled by a regular approach. 
An effective three-step approach
is proposed to dissect the difficult inference problems into manageable steps. Moreover, the discrete nature of $\eta$, often an obstacle for regular inference theory, becomes a beneficial feature to help greatly improve computational and inferential efficiency within our framework.
This development is applied to a case study example of a Gaussian mixture model with an unknown number of components, where we resolve a long-standing open question on quantifying uncertainty in estimating the discrete number of components and associated parameters.
Note that, for inference problems involving discrete or non-numerical parameters, 
their point estimators, if exist, do not provide any uncertainty quantification. The often-used bootstrap approaches lack theoretical support. Bayesian methods may construct credible sets, but the sets are highly sensitive to the prior choices and their frequency coverage performance is typically poor even when the sample sizes are large \citep{hastie2012,Kass1995}. 
The difficulties 
in these bootstrap and Bayesian methods are from the fact that the 
large sample 
CLT 
and Bernstein-von-Mises theorem do not apply for estimators of the discrete or non-numerical parameters. On the contrary, our repro samples method does~not~have~such~constraints.

\noindent{\bf Major contributions and significance:}

\vspace{-4mm}
\begin{enumerate}[leftmargin=*]
\item We develop 
a new  
and broadly applicable inferential framework to construct 
performance-guaranteed confidence sets 
    for parameters of interest under  
    generative model settings. 
 In addition to 
  conventional inference problems, 
  it
  is particularly 
  useful for irregular inference problems with  
 discrete or non-numerical parameters and non-numerical data.

\vspace{-4mm}
\item 
We show that 
a confidence set obtained by the repro samples method is 
either the same as or smaller than that obtained by inverting the Neyman-Pearson test.
We obtain~an optimality result that an optimal Neyman confidence set corresponds to an optimal confidence set by the repro samples method; we also provide examples that the repro samples method improves the result by the Neyman method when it is not optimal.

\vspace{-4mm}
\item The repro samples method provides explicit 
solutions in many cases. In other cases when explicit expressions are unavailable, we provide a practical guideline and an  generic algorithm, often with Monte-Carlo simulation, to 
compute confidence sets. Two techniques are provided to significantly reduce the computing costs: (a) for discrete parameters, we utilize a unique many-to-one inversion mapping 
to develop a novel data-driven approach, with supporting theories, to reduce the parameters' search space; (b) for continuous parameters, we adopt a quantile regression technique used  by \cite{dalmasso2022}~to~avoid~grid~search.  

\vspace{-4mm}
\item  
We provide a general profiling method to handle nuisance parameters 
for finite-sample inference in 
the repro samples framework.
Unlike the conventional `hybrid' or `likelihood profiling' methods that often rely on  asymptotics \citep[e.g.,][]{Chuang2000, dalmasso2022}, the proposed 
method is developed based on specifically constructed $p$-values,
and 
it
guarantees the finite-sample performance 
of the repro samples confidence sets. 

\vspace{-2mm}
\item 
In a case study illustration, we provide a solution to an open and “highly nontrivial” inference problem on how to quantify the uncertainty in estimating the unknown number of components (say $\tau_0$) and the associated parameters in a Gaussian mixture model~\citep[cf.,][]{wasserman2020universal}. 
Numerical studies based on a real data setting
show that~the~repro samples method is the only approach among existing frequentist and Bayesian approaches that can cover $\tau_0$ with the desired accuracy and also be effective~for~other~parameters.
\end{enumerate}

\vspace{-2mm}
\noindent
{\bf Relation to other work}: 
 Besides
the classical Neyman-Pearson test and the recent work on {\it simulation-based inference} (SBI) \citep{dalmasso2022},
the repro samples method~stems from and is closely related to several simulation-based approaches across Bayesian, frequentist and fiducial (BFF) paradigms \citep{Berger2020},  namely,  {\it Bootstrap} \citep{EfroTibs93},  {\it Approximate Bayesian Computation} (ABC) \citep{Robert2016}, 
{\it Generalized fiducial Inference} (GFI) \citep{Hannig2016} and {\it Inferential Model} (IM) \citep{Martin2015}. 
Specifically, the two key inversion techniques, the direct {\it Fisher inversion} 
in Section~\ref{sec:candidate_general} and its generalization to operate on Borel sets, the {\it Fisher--Dempster inversion}  
in Section~\ref{sec:general}, are also a key technique used in GFI \citep{Hannig2016} and IM \citep{Martin2015}, respectively, although our formulations and applications are a little different, and our repro samples method is fully frequentist and does not involve any fiducial justification or use of Dempster--Shafer calculus.
Moreover, 
bootstrap, GFI and most other large-sample methods
rely on CLT, except in  
special cases; they often cannot handle irregular inference problems involving discrete/non-numerical parameters or structures. Furthermore, 
the required sufficiency in ABC, the dependence of BvM-type theorems (thus CLT) for inference justification in GFI and ABC \citep{Hannig2016, Li2016} and the unaddressed 
$\epsilon$ (approximation error) question \citep{Li2016} in ABC and GFI have hampered the use of these approaches in practice.  
Finally, 
{\it universal inference} \citep{wasserman2020universal} is  another novel framework 
for performance-guaranteed  
inferences
without classical regularity conditions, although it relies on a Markov inequality to justify its validity (test size) at some expanse of power \citep{Dunn2022, TseDavison2023}.  
A Gaussian mixture example is also studied in \cite{wasserman2020universal} but, unlike our method, it cannot provide a two-sided confidence set for $\tau_0$ and has less power than ours. 
Due to space limit, we provide detailed reviews and comparisons of these related methods in Appendix~\ref{sec:BFFcomprison}.

\noindent
{\bf Organization of the remaining sections:} 
Section 2 develops a basic yet general framework of the repro samples method with supporting theories. A connection and comparison to the classical Neyman inversion approach
is also provided.  
Section 3 contains~further~developments on handling nuisance parameters,  
constructing a data-driven candidate set to improve computing efficiency, and developing an effective three-step procedure tailored 
for problems with mixed types of parameters. A practical guideline for the use of the repro samples method is also provided.
Section 4 is a case study that addresses a long-standing open inference problem in the Gaussian mixture model. 
Section 5 contains real data and simulation studies.
Section 6 includes further discussions.  
Due to space limit, proofs of all theorems,
as well as additional algorithms, examples, and numerical results, are provided in the appendices.

\section{
A general inference framework by repro samples}
\label{sec:general}

Let  $\*U \in {\mathcal U} \subset \RR^r$ be a measurable random vector 
on a probability space $(\Omega, {\mathcal F}, \P)$ with~$\{ \*U \leq \*u\} = \{\omega \in \Omega: \*U(\omega) \leq \*u\}$. 
The randomness of $\*Y = \*Y(\omega)$ is from $\*U = \*U(\omega)$. 
For a given $\btheta$,  
we call a simulated $\*u^* \sim \*U$ a {\it repro sample} of $\*u^{rel}$ and
$\*y^* = G({\btheta}, \*u^*)$  
a {\it repro~sample}~of~$\*y_{obs}$.

\subsection{\hspace{-3mm} Basic version: a level-(1 - $\alpha$) confidence set by Fisher-Dempster inversion}\label{sec:2.1}

Let $T(\cdot, \cdot)$ be a (user specified) mapping function from ${\mathcal U} \times \Theta \to$ ${\mathcal T} \subseteq \RR^{q}$, for some $q \leq n$, which we refer to as a {\it nuclear mapping}.
We first quantify the  uncertainty of the random $\*U$ via the function $T(\cdot, \cdot)$ and a Borel set $B_{1 -\alpha}(\btheta) \subset  {\mathcal T}$: 
\begin{equation}
    \label{eq:B}
    \P \left\{T(\*U, \btheta)  \in B_{1 -\alpha}(\btheta) \right\}\ge 1 - \alpha,
\end{equation}
for a given $\btheta$. Then, let ``s.t.'' be an abbreviation for ``such that,'' we construct a set in $\Theta$:
\begin{equation}
\label{eq:G1}
\Gamma_{1-\alpha}(\*y_{obs}) = \big\{\btheta: \exists \, \*u^* \in {\mathcal U} \mbox{ s.t. }  \*y_{obs} = 
G({\btheta}, \*u^*),
\,  
T(\*u^*, \btheta)  \in B_{1 - \alpha}(\btheta)  \big\} \subset \Theta. 
\end{equation}
That is, for a potential $\btheta \in \Theta$, if there exists a  realization $\*u^*$, with $T(\*u^*, \btheta)  \in B_{1 - \alpha}(\btheta)$,~such that the repro sample ${\*y}^* = G( {\btheta}, {\*u}^*)$ matches with  $\*y_{obs}$, then we keep this $\btheta$ in set $\Gamma_{1-\alpha}(\*y_{obs})$. Here, the observed $\*y_{obs}$ can be produced by any ${\btheta} \in \Gamma_{1-\alpha}(\*y_{obs})$, with a corresponding potential realization $\*u^*$ (satisfying (\ref{eq:B})), so we cannot rule out the possibility of these $\btheta = \btheta_0.$
Theorem~\ref{thm:1} below states that $\Gamma_{1-\alpha}(\*y_{obs})$ is a level $1 - \alpha$ confidence set in both finite-sample and large-sample (or other) approximations, followed by a corollary on hypothesis testing. 
\begin{theorem}\label{thm:1} Assume model (\ref{eq:1}) holds with $\*\theta = \*\theta_0$. If inequality (\ref{eq:B}) holds exactly for any fixed $\*\theta$, then the following inequality holds exactly 
\begin{equation}
\label{eq:v1}
\P\left\{ {\btheta}_0 \in \Gamma_{1-\alpha}(\*Y)\right\} \ge 1 - \alpha 
\,\,\, \hbox{for $ 0< \alpha <1$}.
\end{equation}
Furthermore, if the inequality (\ref{eq:B}) holds approximately 
 with $\P \big\{T(\*U, \btheta)  \in B_{1 -\alpha}(\btheta) \big\}\ge \alpha\{1 + o(\delta^{'})\}$, then (\ref{eq:v1}) holds approximately 
 with $\P\left\{ {\btheta}_0 \in \Gamma_{1-\alpha}(\*Y)\right\} \ge \alpha\{1 + o(\delta^{'})\}$, for $0< \alpha <1$, where $\delta^{'} > 0$ is 
a small value 
that may or may not
depend on sample size $n$. 
\end{theorem}

\begin{corollary}\label{col:test}
For $\Theta_0 \subsetneq \Theta$ and 
a test $H_0: \*\theta \in \Theta_0$ vs $H_1: \*\theta \not\in \Theta_0$, we can define a p-value 
\begin{align*} 
p(\*y_{obs})=   1- \inf\nolimits_{\*\theta \in \Theta_0}\left[\inf\left\{{1- \gamma}: 
\*\theta \in \Gamma_{1 - \gamma}
(\*y_{obs})\right\}\right],
\end{align*}
where $\Gamma_{1-\alpha}(\*y_{obs})$ is by  \eqref{eq:G1}.
Rejecting $H_0$ when
$p(\*y_{obs}) \leq \alpha$ 
leads to a size $\alpha$ test, $0 < \alpha < 1$. 
\end{corollary}

In Theorem~\ref{thm:1},
 $\delta^{'}$ 
 may depend  on $n$ with $\delta^{'} \to 0$ as $n \to 0$ for a large-sample~approximation, but there are also examples in which $\delta^{'}$ does not involve $n$.
For example,~if~$Y | \theta  = \lambda \sim $ $  \text{\small Poisson}(\lambda)$, then 
$U =$ $ \frac{Y -  \lambda}{\sqrt{\lambda}}
\to N(0,1)$ for a large $\lambda$. So, 
taking $T(U, \lambda) = U$, we~have~$\P \{T(U,  \lambda) $ $ \in B \} =  \int_{{t} \in B}
\phi({t}) d {t} \{1 + o(\frac1\lambda)\}$ for any set $B$; thus $\delta^{'} = \frac1\lambda$. 
Here, 
 $\phi(t)$ is $N(0,1)$ density~function. Also, in (\ref{eq:B}) and (\ref{eq:v1}), 
we keep 
the inequalities ``$\geq$'' to cover the situations when either~$\*Y$~or~$\*\theta_0$ is discrete. These two ``$\geq$'' can be replaced by
 ``='' in other situations and the method does~not suffer systemic power loss (unlike the universal inference method --- see an elaboration~in Appendix~\ref{sec:universal-power-loss}).  
 Further discussions on optimality in comparison with Neyman-Pearson test is in Section~\ref{sec:NeymanInversion}, where we show the repro samples method is broader and 
 can achieve the same or better results. Asymptotic discussions may help but are not needed for efficiency. 
 
In the repro samples development, the role of the nuclear mapping  $T(\cdot, \cdot)$ is similar to, but broader than, that of a test statistic in the classical Neyman-Pearson framework. We~assume throughout the paper that  $T(\cdot, \cdot)$ is provided (user-specified), and    
a general guide on its choice is given in  Section~\ref{sec:guide}.
Regardless of its choice, the repro samples method always has the desired coverage rate, but the power may be impacted.
Two obvious special cases~are: (i) 
$T(\*U,\btheta) = \*U$, unrelated to $\btheta$, and (ii) $T(\*U,\btheta) = \widetilde T(\*Y, \btheta)$ being a test statistic.   
The special case (i) 
is closely related to the IM development of \cite{martin2013inferential}, although IM attempts to achieve a higher level goal of producing a  probabilistic inference~for~$\btheta$~that requires random sets (not a single fixed Borel set) and use of imprecise probability system {\it Dempster-Shafer calculus}.
Our development also allows $T(\*U,\btheta)$ depending on $\btheta$ under consideration, and it provides a great  flexibility which is necessary for many complex inference problems.
The special case (ii)  is closely related to the classical Neyman inversion method, although $T(\cdot, \cdot)$ does not need to be a test statistic and is more general. Section~\ref{sec:NeymanInversion} next investigates in detail the special case (ii) and its connection to the classical Neyman-Pearson~method.

To illustrate the proposed repro samples method, we consider below a binomial example. 

\vspace{-2mm}
\begin{example}[Binomial sample with success probability $\theta_0$] \label{ex:bin}
Assume $Y \sim Binomial(r,\theta)$ $0 < \theta < 1$.  
In the form of  (\ref{eq:1}), $Y = \sum_{i =1}^r {\bf 1}{(U_i \leq \theta_0)}$ for $U_1, \ldots, U_r \sim U(0,1)$. Correspondingly, $y_{obs}$ $=  \sum_{i =1}^r {\bf 1}{(u_i^{rel} \leq \theta_0)}$, for a realized $y_{obs}$ from $\theta_0$ and $\*u^{rel} = (u_1^{rel} , \ldots, u_r^{rel} )$.~We~consider a nuclear mapping function
$T(\*u,  \theta) =  \sum_{i =1}^r {\bf 1}(u_i $ $\leq \theta)$, so $T(\*U,  \theta) =   \sum_{i =1}^r {\bf 1}{(U_i \leq \theta)} \sim Binomial(r, \theta)$, for each given $\theta \in \Theta = (0,1)$. Let $B_{1 -\alpha}(\theta) = [a_L(\theta)  , a_U(\theta)]$ in (\ref{eq:B}) be the shortest interval whose   
bounds are $(i,j)$ pairs in $\{(i,j): 
    \sum^j_{k=i} {r \atopwithdelims( ) k} \theta^k (1 - \theta)^{(r-k)} \geq 1 - \alpha\}$;~i.e.,  
\begin{equation}
    \label{eq:B-bounds}
\left(a_L(\theta), a_U(\theta)\right) = 
\arg\min\nolimits_{\big\{(i,j): \sum^j_{k=i} {r \atopwithdelims( ) k} \theta^k (1 -  \theta)^{(r-k)} \geq 1 - \alpha \big\}}
    |j-i|.
\end{equation} 
Equation (\ref{eq:B}) can be  directly verified  
and, after simplification, 
the confidence set (\ref{eq:G1}) becomes 
\begin{eqnarray}
\label{eq:BinExample}
\Gamma_{1-\alpha}(y_{obs}) &=& 
 \left\{ \theta \big |a_L(\theta) \leq y_{obs} \leq  a_U(\theta) \right\}.
\end{eqnarray}
Due to space limit, the derivation to get (\ref{eq:BinExample}) and also a numerical study are placed in~Appendix~\ref{sec:bin_example}. The numerical study shows the finite-sample interval 
(\ref{eq:BinExample}) has the desired coverage rate in all settings (including $n \theta_0 < 5$ cases). It performs similarly~to~(slightly~better than) Stern's exact method and outperforms other four existing, both exact and asymptotic,~methods. 
\end{example}

\subsection{Test statistic as nuclear mapping function and Neyman inversion of test}\label{sec:NeymanInversion}

Under the classical hypothesis testing setup $H_0: \btheta_0 = \btheta$ vs $H_1: \btheta_0 \not = \btheta$, we often construct a test statistic, say $\widetilde T(\*y_{obs}, \btheta)$, and assume we know the distribution of $\widetilde T(\*Y, \btheta)$ for $\*Y = G(\btheta,\*U)$ generated under $H_0: \btheta_0 = \btheta$. 
By the classical testing method, we can derive 
a level $1 - \alpha$ acceptance region $A_{1-\alpha}(\btheta) = \big\{\*y_{obs} \big| \widetilde T(\*y_{obs}, \btheta) \in B_{1 -\alpha}(\btheta)\big\}$, where set $B_{1 -\alpha}(\btheta)$ satisfies~$\P\big\{\widetilde T(\*Y, \btheta) \in   
B_{1 -\alpha}(\btheta)\big\} \ge $ $ 1 - \alpha$. Then, by the {\it Neyman's inversion} (test duality) method,  
a level $1 - \alpha$ confidence~set~is $\widetilde \Gamma_{1-\alpha}(\*y_{obs}) =  \left\{\btheta:  \*y_{obs} \in  A_{1-\alpha}(\btheta) \right\}
= \left\{\btheta:  \widetilde T(\*y_{obs}, \btheta)  \in B_{1 - \alpha}(\btheta) \right\}. $

Now suppose our nuclear mapping $T(\*u, \btheta)$
is defined through
the same test statistic as:
\begin{equation} \label{eq:tt1}
T(\*U, \btheta) = \widetilde T(\*Y, \btheta) = \widetilde T(G({\btheta}, \*U), \btheta), 
\end{equation} 
where $\*Y = G(\btheta,\*U)$ is from 
the same $\btheta$. Since $\P\left\{T(\*U, \btheta) \in  B_{1 -\alpha}(\btheta)\right\} = \P\big\{\widetilde T(\*Y, \btheta) \in  B_{1 -\alpha}(\btheta)\big\}$ 
$\ge 1 - \alpha$ for the same $B_{1 -\alpha}(\btheta)$, the repro samples confidence set (\ref{eq:G1}) becomes:
\begin{eqnarray*}
\Gamma_{1-\alpha}(\*y_{obs}) 
 =  \left\{\btheta: \exists \*u^* \in {\mathcal U} \mbox{ s.t. }  \*y_{obs} = G({\btheta}, \*u^*), \widetilde T(G({\btheta}, \*u^*), \btheta)  \in B_{1 - \alpha}(\btheta)  \right\}.
\end{eqnarray*}
Theorem~\ref{thm:NP} below states that 
$\widetilde \Gamma_{1-\alpha}(\*y_{obs})$ by  Neyman inversion contains the above
$\Gamma_{1-\alpha}(\*y_{obs})$. 
 \begin{theorem}\label{thm:NP}
 If the nuclear mapping function is defined through a test statistic as in (\ref{eq:tt1}), then we have $\Gamma_{1-\alpha}(\*y_{obs}) \subseteq  \widetilde \Gamma_{1-\alpha}(\*y_{obs})$. The two sets are equal  
$\Gamma_{1-\alpha}(\*y_{obs}) = \widetilde \Gamma_{1-\alpha}(\*y_{obs})$, when $\widetilde \Gamma_{1-\alpha}(\*y_{obs}) \subseteq \big\{\btheta:  \*y_{obs} = G({\btheta}, \*u^*), 
\exists \, \*u^* \in {\mathcal U} \big\}$.
 \end{theorem}
 
It is possible that $\Gamma_{1-\alpha}(\*y_{obs}) \subsetneq \widetilde \Gamma_{1-\alpha}(\*y_{obs})$ strictly. The following is such an example.

\begin{example}\label{ex:3}
  Let $\*y_{obs} = (y_1, \ldots, y_n)^\top$ be an observed sample with $\theta_0$ from the model~$Y_i$ $ = \theta + U_i,$ for  $U_i \sim U(-1,1),$  $i = 1, \ldots, n$, and $\Theta 
  = {\mathcal Y} = (-\infty, \infty)$. 
 An unbiased~(also $\sqrt{n}$~consistent) point estimator 
 is $\bar Y = \frac 1n \sum_{i=1}^n Y_i$. Suppose we
   use the test statistic~$\widetilde T(\*y_{obs}, \theta) = \bar y - \theta$. Since $n \{(\bar Y  - \theta) $ $ + 1\}/2 = \sum_{i=1}^n \frac{U_i + 1}2$ follows an Irwin-Hall distribution under $H_0: \theta_0 = \theta$,  
 a level $95\%$ confidence interval by the Neyman inversion 
 method  is $\widetilde \Gamma_{.95}(\*y_{obs}) = (\bar y - \frac2n q_{.975} + 1, \bar y + \frac2n q_{.975} - 1)$. Here, $q_{.975}$ is the $97.5\%$ quantile~of the Irwin-Hall distribution. 
 If we use the repro samples method, our confidence interval
 $\Gamma_{.95}(\*y_{obs})$ $ =  \big\{\theta: $ $ y_{i} = \theta +  u_i^*,
 i = 1, $ $\ldots,$ $ n; \, \bar y \in (\theta - \frac2n q_{.975} + 1, \theta + \frac2n q_{.975} - 1);
\exists  \*u^* \in (-1, 1)^n \big\} = \left\{\theta: \theta \in \cap_{i=1}^n (y_i -1, y_i+1) \right\} \cap \widetilde \Gamma_{.95}(\*y_{obs})$. Since
the constraint $\{\theta \in \cap_{i=1}^n (y_i -1, y_i+1)\}$ is often in tack,
the strictly smaller statement $\Gamma_{.95}(\*y_{obs}) \subsetneq \widetilde \Gamma_{.95}(\*y_{obs})$ holds for many~realizations of $\*y_{obs}$. 
Our numerical study
(due to space limit placed in Appendix~\ref{sec:example_sec2}) also confirms this conclusion empirically --- 
both intervals have desired empirical coverage rates $95.1\%$, but among the $1,000$ repetitions,  $219$ times 
$\Gamma_{.95}(\*y_{obs}) = \widetilde \Gamma_{.95}(\*y_{obs})$ and $781$ times $\Gamma_{.95}(\*y_{obs}) $ $\subsetneq \widetilde \Gamma_{.95}(\*y_{obs})$. The average length  of the repro samples intervals $|\Gamma_{.95}(\*y_{obs})| = 0.915$ is shorter than that of $|\widetilde \Gamma_{.95}(\*y_{obs})| = 1.292$. 
\end{example}

There are two immediate implications from the special case with a nuclear mapping $T(\cdot, \cdot)$ defined via a test statistic as in  (\ref{eq:tt1}).
First, by Theorem~\ref{thm:NP},the repro samples confidence set $\Gamma_{1 - \alpha}(\*y_{obs})$ 
is never worse than the one obtained by the classical testing approach, thus is~more desirable. Second, if the test statistic is optimal (in the sense it leads to a powerful test or an optimal confidence set), the corresponding repro samples confidence set $\Gamma(\*y_{obs})$~is~also optimal. 
Corollary~\ref{cor:UMA-NP} below is the formal statement,
where a level $1 - \alpha$ {\it uniformly most accurate} (UMA) confidence set (a.k.a. {\it Neyman shortest}) refers to a level $1 - \alpha$ confidence~set that minimizes the probability of false coverage (i.e., probability of covering a wrong parameter value) over a class of level $1 - \alpha$ confidence sets \citep[][\S 9.3.2]{tCAS90a}.  

\begin{corollary}\label{cor:UMA-NP}
(a) If a test statistic $\widetilde T(\*y_{obs}, \btheta)$  corresponds to the uniformly most powerful test and  $\widetilde \Gamma(\*y_{obs})$ is a level $1 - \alpha$ UMA confidence set, then the set $\Gamma(\*y_{obs})$ by the corresponding repro sample method is also a level $1 - \alpha$ UMA confidence set. \\
(b) If a test statistic $\widetilde T(\*y_{obs}, \btheta)$ corresponds to the uniformly most powerful unbiased test and  $\widetilde \Gamma(\*y_{obs})$ is a level $1 - \alpha$ UMA unbiased confidence set, then the confidence set $\Gamma(\*y_{obs})$ by the corresponding repro sample method is also a level $1 - \alpha$ UMA unbiased confidence set.
\end{corollary}

Finally, we would like to 
stress that the nuclear mapping $T(\cdot, \cdot)$ does not need to be a test statistic.
One such example, on inference for a nonparametric quantile with noise-perturbed observations (for privacy consideration), is provided in Example~\ref{ex:quantile_privacy} of Appendix~\ref{sec:example_sec2}, in which a nuclear mapping $T(\cdot, \cdot)$ is readily available but an effective test statistic is difficult to construct. 
Furthermore, Neyman-Pearson lemma suggests that the likelihood ratio~test~(LRT) is uniformly most powerful for a simple-versus-simple hypothesis, but it does not necessarily hold for a two-sided test. Thus, an interval constructed by inverting an LRT is not necessarily optimal. Example~\ref{ex:3a} of Appendix~\ref{sec:example_sec2}, a continuation of Example~\ref{ex:3}, demonstrates that the LRT confidence interval is not optimal, and a repro samples interval can improve it.

\subsection{Beyond a fully specified generative model and defining oracle parameters} \label{sec: extendedsetup}

The repro samples method is also applicable even in cases  when 
model \eqref{eq:1} is not fully given.
For instance, consider a nonparametric inference problem on the $\zeta$-th quantile  
of an 
unknown distribution $F$,  
say, 
$\theta_0 = F^{-1}(\zeta)$. Since sample
$Y \sim F$, we have $I(Y < \theta_0) \sim \text{Bernoulli}(\zeta)$~and  equation
$I(Y < \theta_0) = U$, where $U \sim \text{Bernoulli}(\zeta)$. Now suppose we observe data $\*y_{obs} =  (y_1^{obs}, $ $ \ldots, y_n^{obs})^\top$, for a fixed $n$, then the equation (treated as a given model specification) becomes 
\begin{equation} \label{eq:quantile}
    \sum\nolimits_{i=1}^{n} I(Y_i - \theta_0 <0) - \sum\nolimits_{i=1}^n U_i = 0,
\end{equation}
where $Y_i \overset{iid}{\sim} F$ and $U_i \overset{iid}{\sim} \text{Bernoulli}(\zeta)$. 
Equation (\ref{eq:quantile}) 
cannot be re-expressed in the form of (\ref{eq:1}).
Although based on (\ref{eq:quantile}),
we cannot generate a repro sample ${\*y}^*$ that is directly comparable to the observed  $\*y_{obs}$, we can still use the repro samples method. 
Particularly, let us
consider a {\it generalized generative 
model}, of which (\ref{eq:quantile}) is a special case: 
\vspace{-2mm}\begin{equation}
    \label{eq:AA}
    g(\*Y, \btheta, \*U) = \*0,
\end{equation} 
where $(\*Y, \btheta, \*U)$ is the same as in (\ref{eq:1}) and $g$ is a given  
function from ${\mathcal Y} \times \Theta \times {\mathcal U}  \to \R^s$.
The 
realization version is
  $g(\*y_{obs}, \btheta_0, \*u^{rel}) = \*0$.
In this case, 
we modify $\Gamma_{1-\alpha}(\*y_{obs})$ in (\ref{eq:G1})~as~follows: 
\begin{eqnarray}
\label{eq:G2}
\Gamma_{1-\alpha}(\*y_{obs}) = \big\{\btheta:  \exists 
\,\*u^* \in {\mathcal U} \mbox{ s.t. } 
 g(\*y_{obs}, \btheta, \*u^*) = 0, T(\*u^*, \btheta)  \in B_{1 - \alpha}(\btheta)
\big\} \subset \Theta.
\end{eqnarray}

A special case of (\ref{eq:AA}) that is relevant to machine learning and modern data science research is: 
Irrespective of the  assumed model structure, 
completely given or not, 
we only assume that the sample data $\*Y$  contain sufficient information to recover our target model parameter $\boldsymbol{\theta}$, in the sense that the parameter $\btheta$ can be recovered by an algorithm 
when given $\*Y$ and $\*U$:

\vspace{-4mm}
\begin{equation}
\label{eq:invertedH}
\btheta
= H(\*Y,\*U).
\end{equation}

\vspace{-2mm} \noindent
That is, when given observed $\*y_{obs}$ and its corresponding realized $\*u^{rel}$, we can express our target parameter $\btheta_0$ as $\btheta_0 = H(\*y_{obs}, \*u^{rel}).$
Here, $H(\cdot, \cdot)$ is a given algorithm or function. Under the generative model (\ref{eq:1}), for instance, 
$H(\cdot,\cdot)$ is typically an inversion  algorithm~that~solves~for 
\begin{equation}
\theta_0 = H(\*y_{obs}, \*u^{rel}) \overset{\rm def}{=} \arg \min\nolimits_{\theta} L\big(\*y_{obs}, G(\*\theta, \*u^{rel})\big),
    \label{eq:model-alg}
\end{equation}
where $L\big(\*y, \*y'\big)$ is a loss function that measures the difference between two copies of data $\*y$ and $\*y'$. Similarly, under (\ref{eq:AA}), 
$H(\cdot, \cdot)$ is the algorithm that solves $g(\*y_{obs}, \btheta, \*u^{rel}) = \*0$ for $\btheta_0$, though mathematically (\ref{eq:invertedH}) is a special case of (\ref{eq:AA}) with $g(\*y, \btheta, \*u) \overset{\text{def}}{=} H(\*y, \*u) - \btheta$. 
Under model (\ref{eq:invertedH}), we can modify the repro sample confidence set $\Gamma_{1-\alpha}(\*y_{obs})$ in (\ref{eq:G1}) as: 
\begin{eqnarray}
\label{eq:G2h}
\Gamma_{1-\alpha}(\*y_{obs}) = \big\{\btheta = H(\*y_{obs}, \*u^*):  \exists 
\,\*u^* \in {\mathcal U} \mbox{ s.t. } 
  T(\*u^*, \btheta)  \in B_{1 - \alpha}(\btheta)
\big\} \subset \Theta.
\end{eqnarray}

Note that, in more complex problems, especially when it involves nonconvex optimization, we may get a local solution from (\ref{eq:model-alg}) instead of the $\btheta_0$ that generated $\*y_{obs}$.
In this~case, we refer to such a local solution $\btheta_0 = H(\*y_{obs}, \*u^{rel})$ as an {\it oracle parameter}. Our~repro samples inference is then for this algorithm $H(\cdot, \cdot)$ specific parameter. Here, with slight abuse of notation, we still use $\btheta_0$ for the oracle parameter, and we further assume it is uniquely defined given the algorithm $H(\cdot, \cdot)$. 
The same model assumption  (\ref{eq:invertedH}) is also used in the development of {\it extended fiducial inference} \citep{liang2025extended}, although the authors imposed a further requirement that $\btheta_0 = H(\*y_{obs}, \*u^{rel})$ is the model parameters $\btheta_0$ that generated~$\*y_{obs}$.

In this paper, we treat $g(\cdot, \cdot)$ or $H(\cdot, \cdot)$ as part of given model specification. The coverage results~in Theorem~\ref{thm:1}~can be directly extended to the sets  
defined in (\ref{eq:G2}) and~(\ref{eq:G2h}),~where the statements also cover the oracle parameter defined above. Due to space limit,~the formal theorem, with its proof, is provided in Appendix~\ref{sec:the_G2}. Additionally, Appendix~\ref{sec:example_sec2} contains two 
illustrative examples of  \eqref{eq:G2}, where   
Example~\ref{ex:crq}  demonstrates how to make~a~nonparametric 
inference for the 
quantile $\theta_0 = F^{-1}(\zeta)$ with a completely unknown distribution $F$ and it shows that the repro samples method
works well even for $\zeta$ near $0$ or $1$, whereas~the conventional bootstrap method has issues on coverage.
The method in 
Example~\ref{ex:crq}~is~also further utilized 
to provide a robust approach to estimating location and scale parameters in Gaussian mixture models. Further details are in Example~\ref{ex:crq_robust} of Appendix~\ref{sec:example_sec2} and~Section~\ref{sec:mu-sigma}.   

Unless specified otherwise, in the rest of the paper,  
$\Gamma_{1-\alpha}(\*y_{obs})$ refers to that in (\ref{eq:G1}), instead of (\ref{eq:G2}) or (\ref{eq:G2h}). All results developed for (\ref{eq:G1}) have direct extensions for~(\ref{eq:G2})~and~(\ref{eq:G2h}).

\section{Further developments: repro samples method with mixed types of parameters}\label{sec:3}

Inference problems 
in contemporary statistics

often involves nuisance parameters and mixed types of parameters. 
In this section,~we partition the parameters by $\btheta = $ $(\eta, {\*\beta}^\top)^\top$, for $\eta \in {\it \Upsilon}$ and $\*\beta \in {\cal B}$.  
We start with Section~\ref{sec: nuisance_control} on a profile method for finite-sample inference for a general $\eta$. Then, we focus~on~the~case of mixed types of parameters in Sections~3.2 and 3.3, where $\eta$ is discrete or non-numeric and the remaining $\beta$ may or may not depend on $\eta$. A practical guideline on use of repro samples methods is provided in~Section~\ref{sec:guide}.

\subsection{A generic profile approach for finite inference}
\label{sec: nuisance_control}

The generic profile method below does not need a plug-in estimator of nuisance parameter nor large-sample results; it guarantees finite-sample inference performance of target~parameter. 

 Let $\*y = G(\btheta, \*u)$ be a copy of sample data and  $\Gamma_{1 - \gamma}(\*y)$ be a level ${1 -\gamma}$ confidence set for $\btheta = (\eta, {\*\beta}^\top)^\top$. 
Let us consider another copy of parameter vector $\widetilde \btheta = (\eta, \widetilde {\*\beta}^\top)^\top \in \Theta$ that has the same $\eta$ but perhaps a different $\widetilde {\*\beta} \not= \*\beta$. To explore the impact of a potentially misspecified nuisance parameter $\widetilde {\*\beta} \not= \*\beta$, we define
\begin{equation}
    \label{eq:nu}
\nu(\*u; \eta, \*\beta, \widetilde {\*\beta}) = \inf\nolimits \left\{1-\gamma: \widetilde\btheta \in \Gamma_{1-\gamma}(\*y)\right\} = \inf\nolimits\left\{1-\gamma: \widetilde\btheta \in \Gamma_{1-\gamma}(G(\btheta, \*u))\right\}. 
\end{equation}
From 
Corollary~\ref{col:test}, we can interpret $1-\nu(\*u; \eta, \*\beta, \widetilde {\*\beta})$ as a `$p$-value' when we test the null hypothesis that the data $\*y$ is from $\widetilde \btheta = (\eta, \widetilde {\*\beta}^\top)^\top$ but in fact $\*y = G(\btheta, \*u)$ is from $\btheta = (\eta, {\*\beta}^\top)^\top$. 

To make inferences for the target $\eta$, we propose a {\it profile nuclear mapping~function}: 
\begin{align}
\label{eq:T_p}
T_p(\*u, \btheta) = \min\nolimits_{\widetilde{\*\beta}} \,\, \nu\left(\*u; \eta, \*\beta, \widetilde {\*\beta}\right). 
\end{align}
Ideally, we would seek that  the profile nuclear mapping $T_p(\*u, \btheta) = T_p(\*u, \eta)$ free of $\*\beta$. However, even if 
 $T_p(\*u, \btheta) = T_p(\*u, (\eta, {\*\beta}^\top)^\top)$ still depends on  the unknown nuisance  $\*\beta$, its random version $T_p(\*U, \btheta)$
is dominated by a $U(0,1)$ random variable as stated in Lemma~\ref{lemma:nuisance}. 
\begin{lemma}
\label{lemma:nuisance}
Suppose $T_p(\*u, \btheta)$ is a profile nuclear mapping function from ${\mathcal U} \times \Theta \to [0,1]$ defined in (\ref{eq:T_p}). Then, we have $\P\left\{T_p(\*U, \btheta) \leq 1 - \alpha \right\} \ge 1 - \alpha.$ 
\end{lemma}

 By Lemma~\ref{lemma:nuisance}, the Borel set for our repro samples method is 
$B_{1 -\alpha} =(0, \alpha]$.  
We construct 
\begin{equation}
\label{eq:CIeta}
{\small \Xi_{1 -\alpha}(\*y_{obs}) = }\resizebox{.785\textwidth}{!}{${\small \bigg\{\eta: \exists \, \*u^* \in {\mathcal U} \, \text{and} \, \*\beta, 
\text{s.t.}  {\footnotesize \begin{pmatrix}\eta  \\ \*\beta
\end{pmatrix}}\in \Theta,  \*y_{obs} = 
G\left({\footnotesize \begin{pmatrix}\eta \\ \*\beta
\end{pmatrix}}, \*u^*\right),
\,
T_p\left(\*u^*, {\footnotesize \begin{pmatrix}\eta \\ \*\beta
\end{pmatrix}}\right) \leq 1 - \alpha   \bigg\}.} 
$}
\end{equation}

\noindent
Theorem~\ref{the:nuisance} below states that $\Xi_{1 -\alpha}(\*y_{obs})$ above
is a level $1 - \alpha$ confidence set for the target $\eta$.
\begin{theorem}\label{the:nuisance}
Suppose $\*Y = G(\btheta, \*U)$ with $\btheta = \btheta_0,$ where $\btheta_0 = (\eta_0, \*\beta_0^\top)^\top$ is the true parameter, and $\Xi_{1 -\alpha}(\cdot)$ is defined in \eqref{eq:CIeta}. Then,  we have $\P\left\{\eta_0 \in \Xi_{1 -\alpha}(\*Y)\right\} \ge 1 - \alpha$. 
\end{theorem}

In (\ref{eq:CIeta}), we only collect $\eta$. 
If collecting both $\eta$ and $\*\beta$, we obtain a joint confidence set for $\btheta = (\eta, \*\beta^\top)^\top$, but this joint set is unsuitable for joint inference because it imposes~no constraints on the elements of $\*\beta$. Nevertheless, since our interest is solely in $\eta$, we profile out the unconstrained
$\*\beta$ to obtain the confidence set of $\eta$ in (\ref{eq:CIeta}). In fact, we can obtain a confidence set for $\eta$ from any joint confidence set of $\btheta$ by profiling out the nuisance $\*\beta$. However, profiling out a set that places no restrictions on the nuisance parameter $\*\beta$ is preferable to profiling a well-defined joint confidence set that does impose constraints on $\*\beta$. For example, in a simple bivariate independent Gaussian case $(x_{i1},x_{i2})^\top \sim N\big((\nu_1, \mu_2)^\top, \text{diag}(1,1)\big)$, $i = 1, \ldots,n$,  profiling out nuisance parameter $\mu_2$ in the best $95\%$ joint confidence set $\big\{ (\mu_1, \mu_2): (\mu_1 - \bar{x}_1)^2 + (\mu_2 - \bar{x}_2)^2 \le \frac{\chi^2_{2,0.95}}{n} \big\}$ leads to a 95\% confidence interval of $\mu_1$, $\big( \bar{x}_1 - \sqrt{\frac{\chi^2_{2,0.95}}{n}}, \bar{x}_1 + \sqrt{\frac{\chi^2_{2,0.95}}{n}} \big)$, which is worse than $\big( \bar{x}_1 - \frac{1.96}{\sqrt{n}}, \bar{x}_1 + \frac{1.96}{\sqrt{n}} \big)$, 
by profiling out the joint set $\{ (\mu_1, \mu_2): $ $ |\mu_1 - \bar{x}_1| \le \frac{1.96}{\sqrt{n}} \}$ with no constraint on $\mu_2$. See also \cite{Michael2019} for a related~discussion.

Finally, 
note that 
$T_p(\*u, \btheta)$ in (\ref{eq:T_p}) is defined through $\nu(\cdot; \eta, \*\beta, \widetilde {\*\beta})$ in (\ref{eq:nu}), which often can be computed either directly or by a Monte-Carlo method. Due to space limit, we refer readers to Appendix~\ref{sec:depth_profile} and Lemma~\ref{lem:nuisance_depth} for the Monte-Carlo method and its justification.

\subsection{\hspace{-3mm} Fisher inversion and
space reduction for discrete/non-numerical parameters}
\label{sec:candidate_general}

In this subsection,  
 the target parameter $\eta \in {\it \Upsilon}$ is  discrete or non-numerical,  
and $\*\beta$~is~nuisance~parameters (of any types) that may even depend on $\eta$. We take advantage of an inherent ``many-to-one" inversion mapping to get a data-dependent candidate set of $\eta$
that is drastically smaller than ${\it \Upsilon}$. 
 By doing so, the discreteness of $\eta,$ a hurdle for many large~sample based inference methods, becomes
an advantage for computation and~theory~derivation.

We first note that $\*y_{obs} = G((\eta_0, \*\beta_0^\top)^\top,  \*u^{rel}),$ 
so we can rewrite $\eta_0$ as a solution of an optimization problem: 
$\eta_0 = \arg \min\nolimits_{\eta} \min\nolimits_{\*\beta}\|\*y_{obs} - G((\eta, \*\beta^\top)^\top, \*u^{rel})\|$ or, more generally, 
\begin{equation}
    \eta_0 = \arg \min\nolimits_{\eta} \min\nolimits_{\*\beta}L\big(\*y_{obs}, G((\eta, \*\beta^\top)^\top, \*u^{rel})\big),
    \label{eq:tau0}
\end{equation}
where $L\big(\*y, \*y'\big)$ is a continuous loss function that measures the difference between two copies of data $\*y$ and $\*y'$. However, since we do not observe the unknown  $\*u^{rel}$, we replace $\*u^{rel}$ with a repro copy $\*u^*$ in (\ref{eq:tau0}). It leads to a function that maps a value $\*u^* \in {\mathcal U}$ to a value $\eta^* \in {\it \Upsilon}$,
\begin{equation}
   \eta^* = \arg \min\nolimits_{\eta} \min\nolimits_{\*\beta}L\big(\*y_{obs}, G((\eta, \*\beta^\top)^\top, \*u^*)\big).
    \label{eq:tau-star} 
\end{equation} 
We refer this operation as {\it Fisher inversion}, which can be traced back to Fisher's fiducial development \citep[cf.,][]{Thornton2022}. 
Here,
${\mathcal U}$ is often uncountable and ${\it \Upsilon}$ is countable, so (\ref{eq:tau-star}) is ``many-to-one'': many $\*u^* \in {\mathcal U}$ correspond to a single $\eta^* \in {\it \Upsilon}$. Of~particularly~interest~is 
$S =\big\{\*u^*: \eta_0 = \arg \min\nolimits_{\eta} \min\nolimits_{\*\beta}L\big(\*y_{obs}, G((\eta, \*\beta^\top)^\top, \*u^*)\big)\big\} \subset {\mathcal U},$
in which the mapping (\ref{eq:tau-star}) produces $\eta^* = \eta_0$ for any $\*u^* \in S$. Since 
$\*u^{rel} \in S$, we know~$S \not = \emptyset$.

If the subset $S \subset {\mathcal U}$ is a nontrivial set,
i.e., $\P\big\{\*U \in S \big | \*y_{obs}\big\} > 0$ has a positive~probability, we can use a Monte-Carlo method to get a  (data-dependent) 
candidate set $\widehat {\it \Upsilon} = \widehat {\it  \Upsilon}(\*y_{obs})$
such that $\eta_0 \in \widehat {\it \Upsilon}(\*y_{obs})$ with a high (Monte-Carlo) probability. 
The idea is to simulate a sequence of $\*u^s \sim \*U$,~say, 
${\mathcal V} = $ $\big\{\*u_1^s, \ldots, $ $\*u_N^s \big\}$. If $|{\mathcal V}| = N$ is large enough, we have ${\mathcal V}\cap S \not = \emptyset$ and thus $\eta_0 \in \widehat{\it \Upsilon}_{{\mathcal V}}(\*y_{obs})$, with a high (Monte-Carlo) probability. Specifically, we define our candidate set as
\begin{equation}
\label{eq:cand-hat}
   \widehat {\it \Upsilon} =  \widehat{\it \Upsilon}_{{\mathcal V}}(\*y_{obs}) = \left\{\eta^s = \arg \min\nolimits_{\eta} \min\nolimits_{\*\beta}L\big(\*y_{obs}, G((\eta, \*\beta^\top)^\top, \*u^s)\big) \,\, \big| \,\, \*u^s \in {\mathcal V} \right\}.
\end{equation}

Formally,
let 
$\*U^*$ be an independent copy of $\*U$.
We assume 
the following~condition~holds:

\vspace{-4mm}
\begin{itemize}
    \item[(N1)] 
\it For 
any $\*u^{rel} \in \mathcal U$, there exists a 
neighborhood of $\*u^{rel}$, say $S_{nb}(\*u^{rel})$, 
such that
\begin{equation*}
    \resizebox{.84\hsize}{!}{$S_{nb}(\*u^{rel})\subseteq S =\big\{\*u^*: \eta_0 = \arg \min\nolimits_{\eta} \min\nolimits_{\beta}L\big(G((\eta_0, \*\beta_0^\top)^\top, \*u^{rel}), G((\eta, \*\beta^\top)^\top, \*u^*)\big)\big\}$},
\end{equation*}
and   $\P_{\*U}\left[\P_{\*U^*|\*U}\left\{\*U^* \in S_{nb}(\*U) \middle | \*U\right\}
   \geq c_{nb}\right] = 1$
   for a positive number $c_{nb} > 0$.

\end{itemize}

\vspace{-4mm}

\noindent
Although the neighborhood $S_{nb}$ and the number $c_{nb}$ are problem specific,  Condition (N1) is often satisfied 
when $L(\cdot, \cdot)$ is a continuous loss function and ${\mathcal U}$ is uncountable. 
Lemma~\ref{lemma:p_bound_C_D_general} below states that, 
under Condition (N1), 
$\widehat{\it \Upsilon}_{{\mathcal V}}(\*Y)$ 
in (\ref{eq:cand-hat}) contains $\eta_0$ with high probability for a large $|\mathcal V|$.  
Here, 
$\P_{\*U, {\mathcal V}}(\cdot)$ refers to the joint distribution of $\*U$ and  Monte-Carlo $\*U^*$'s in ${\mathcal V}$, and $\*Y = G(\btheta, \*U)$ is generated with $\*\theta = \*\theta_0 = (\eta_0, \*\beta_0^\top)^\top.$

\begin{lemma}
\label{lemma:p_bound_C_D_general} 
If Condition (N1) holds, 
then $
    \P_{\*U, {\mathcal V}}\big\{\eta_0 \not\in \widehat{\it \Upsilon}_{{\mathcal V}}(\*Y)\big\} \leq (1-c_{nb})^{|{\mathcal V}|}.
$
\end{lemma}
As long as Condition (N1) holds, a neighborhood to recover the true model with \eqref{eq:tau-star} always exists.
As a result, Lemma~\ref{lemma:p_bound_C_D_general} and the following Theorem~\ref{thm:p_bound_C_D_general} always hold even in challenging cases where the dimension of $\eta_0$ increases with $n.$ However, how small $c_{nb}$~is depends on how many alternative $\eta$ are close to $\eta_0.$ In challenging cases such as when dimension of $\eta_0$ increases rapidly with $n,$ or when weak signals are present in high-dimensional linear models, $c_{nb}$ could be very small, and recovering $\eta_0$ with $\widehat{\it \Upsilon}_{{\mathcal V}}(\*Y)$ may demand a large~$|\mathcal V|.$

In Theorem~\ref{thm:p_bound_C_D_general} below, the probability  $\P_{\*U|{\mathcal V}}(\cdot)$ refers to the condition distribution of $\*U$ given the Monte-Carlo set ${\mathcal V}$ and the ``in probability" statement is regarding to the Monte-Carlo randomness of ${\mathcal V}$. Theorem~\ref{thm:p_bound_C_D_general} 
has a clear 
frequentist statement: 
if we use the proposed method to obtain the candidate set $\widehat{\it \Upsilon}_{{\mathcal V}}(\*Y)$, we have a near $100\%$ chance (a Monte-Carlo probability statement regarding the Monte-Carlo set ${\mathcal V}$) that $\eta_0 \in \widehat{\it \Upsilon}_{{\mathcal V}}(\*Y)$ and in addition 
the frequency coverage probability of the repro sample confidence set $\Gamma_{1-\alpha}'(\*Y) = \Gamma_{1-\alpha}(\*Y) \cap \{ \btheta: \eta \in \widehat{\it \Upsilon}_{{\mathcal V}}(\*Y)\}$ is
at least $1 - \alpha$ (the regular frequentist statement regarding $\*Y$ or equivalently $\*U$). Here again, $\*Y = G(\btheta, \*U)$ is generated with $\*\theta = \*\theta_0 = (\eta_0, \*\beta_0^\top)^\top.$ 
\begin{theorem}
\label{thm:p_bound_C_D_general}
 Suppose Condition (N1) holds and $\widehat{\it \Upsilon}(\*y_{obs})$ is defined in \eqref{eq:cand-hat}. Then, (a) there exists a positive number $c$ such that we have, in probability, 
  $\P_{\*U|{\mathcal V}}\big\{\eta_0 \in \widehat{\it \Upsilon}_{{\mathcal V}}(\*Y)\big\} = 1 - o_p\big(e^{-c|{\mathcal V}|}\big)
 	$.
 (b) Let $\Gamma_{1-\alpha}'(\*Y) = \Gamma_{1-\alpha}(\*Y) \cap \{ \btheta: \eta \in \widehat{\it \Upsilon}_{{\mathcal V}}(\*Y)\}$ where $\Gamma_{1-\alpha}(\*Y)$ is defined in Theorem~\ref{thm:1} or \ref{thm:2}, then
 $\P_{\*U|{\mathcal V}}\big\{\btheta_0 \in \Gamma_{1-\alpha}'(\*Y)\big\} \geq 1- \alpha -o_p\left(e^{-c|{\mathcal V}|}\right)$, in probability.
\end{theorem}

In Theorem~\ref{thm:p_bound_C_D_general}, the term  
$o_p\left(e^{-c|{\mathcal V}|}\right)$ tends to $0$ as the size of the Monte-Carlo set $|{\mathcal V}| \to \infty$ and the results hold for any fixed (data) sample size $n < \infty$; so it is a finite-sample result.  
We will further discuss and elaborate Condition (N1) and Theorem~\ref{thm:p_bound_C_D_general} in Section~\ref{sec:mixture}.

Often times we can find an upper bound for the size of the candidate set $\widehat{\it \Upsilon}_{{\mathcal V}}(\*Y),$ because if $\eta$ is far away from $\eta_0$, it usually has a small probability of entering $\widehat{\it \Upsilon}_{{\mathcal V}}(\*Y).$ Hence, $\widehat{\it \Upsilon}_{{\mathcal V}}(\*Y)$ usually only contains $\eta$ that are relatively closer to $\eta_0.$ The bound however is problem specific and depends on the loss function used in the inversion~algorithm~(\ref{eq:tau-star}). In Theorem~\ref{the:upper_bound} of Section~\ref{sec:mixture}, we derive an upper bound for the candidate set of the unknown 
number of components in Gaussian mixture model. See also Theorem 3 of \cite{wang2022highdimenional} for a similar upper bound  result on the size of a model candidate set
for high-dimensional~regressions.

\subsection{\hspace{-3mm} A three-step procedure for inference with mixed-type parameters}\label{sec:ThreeSteps}

In many applications, 
the target model parameters $\*\theta = (\eta, \*\beta^\top)^\top$ are mixed types and the parameter $\*\beta$ may depend on $\eta$ (e.g., Section~\ref{sec:mixture}). 
In this subsection, we propose a three-step inference procedure to dissect this difficult inference problem into manageable steps, in~which we make inference for $\eta$, $\*\beta$ and jointly $(\eta, \*\beta^\top)^\top$, respectively. 
The proposed three-step 
procedure is as follows: 
\begin{itemize}[leftmargin=2.5em, labelsep=0.3em, 
align=parleft,   itemsep=1pt,
  parsep=0pt,
  topsep=1pt]
\item[Step]1  [Inference for $\eta_0$]  Use a repro samples method to construct a candidate set $\widehat{\it \Upsilon}_{\mathcal V}(\*y_{obs})$~and then a level $1 - \alpha$ confidence set $\Xi_{1-\alpha, 1}(\*y_{obs})$ for $\eta_0$, treating $\*\beta$ as nuisance parameters.
\item[Step]2: [Inference for $\*\beta_0$] For a given $\eta \in 
\widehat{\it \Upsilon}_{\mathcal V}(\*y_{obs})$, use a repro samples method (or a regular inference approach) to construct a level $1 - \alpha$ confidence set $\Xi_{1-\alpha, 2}(\*y_{obs})$~for~$\*\beta_0$.
\item[Step]3: [Joint Inference for $(\eta_0, \*\beta_0^\top)^\top$]
    Construct a level $1 - \alpha$ joint confidence set, say $\Xi_{\alpha, 3}(\*y_{obs})$, for $\*\theta_0 = (\eta_0, \*\beta_0^\top)^\top$ based on the results in Steps 1 and 2. 
\end{itemize}
The procedure breaks our task into three  manageable 
sub-tasks. 
In some situations, especially when $\*\beta$ depends on $\eta$, it is necessary to handle $\eta$ and $\*\beta$ separately by steps. For instance, in the Gaussian mixture model in Section~\ref{sec:mixture}, $\eta = \tau$ is the unknown number of components, and the dimension of $\*\beta = (\*\mu_\tau, \*\sigma_\tau)$ is a function of $\tau$ --- If $\tau =3$, $\*\beta$ has $6$~parameters, while~$\tau = 4$, $\*\beta$ has $8$~parameters; the $\mu_1$ in $\tau =3$ case may or may not line up with the $\mu_1$ in $\tau =4$ case. It is not clear to us how to make a meaningful inference for $\mu_1$ without first working on $\tau$.

In Step 1,  we first use Section~\ref{sec:candidate_general} and formula (\ref{eq:cand-hat}) to obtain a candidate set $\widehat{\it \Upsilon}_{\mathcal V}(\*y_{obs})$. 
Then, we treat $\*\beta$ as nuisance parameters and use Section~\ref{sec: nuisance_control} and formula (\ref{eq:CIeta}) to construct:
\begin{equation}
\label{eq:CIeta1}
\Xi_{1-\alpha, 1}(\*y_{obs}) = \bigg\{ \resizebox{.61\textwidth}{!}{$\eta: \exists \, \*u^* \in {\mathcal U} \, \text{and} \, \*\beta, 
\text{ s.t. }  {\footnotesize \begin{pmatrix}\eta  \\ \*\beta
\end{pmatrix}}\in \Theta,  \*y_{obs} = 
G\left({\footnotesize\begin{pmatrix}\eta \\ \*\beta
\end{pmatrix}}, \*u^*\right),
\,
T_p\left(\*u^*, {\footnotesize\begin{pmatrix}\eta \\ \*\beta
\end{pmatrix}}\right) \leq 1-\alpha$} \bigg\} \cap \widehat{\it \Upsilon}_{\mathcal V}(\*y_{obs}),      
\end{equation}
where $T_p\left(\*u^*,(\eta, \*\beta)^\top \right) $ is defined following  \eqref{eq:T_p} for the discrete 
 $\eta \in \widehat{\it \Upsilon}_{\mathcal V}(\*y_{obs})$.

In Step 2, for each $\eta \in \widehat{\it \Upsilon}_{\mathcal V}(\*y_{obs})$, we treat it as given so the only parameters now are $\*\beta$. By using the repro samples approach in Section~\ref{sec:2.1} (e.g., formula (\ref{eq:G1})) or using a regular inference approach, we can construct a {\it representative set} of $\*\beta_0$, given a $\eta \in \widehat{\it \Upsilon}_{\mathcal V}(\*y_{obs})$: 
\begin{equation}\label{eq:represent}
\Gamma_{1-\alpha}^{[\eta]}(\*y_{obs}) = \big\{\*\beta: \exists \, \*u^* \in {\mathcal U} \mbox{ s.t. }  \*y_{obs} = 
G^{[\eta]}(\*\beta, \*u^*),
\,  
T^{[\eta]}(\*u^*, \*\beta) \in B_{1 - \alpha}^{[\eta]}(\*\beta)\big\} \subset {\mathcal B}.
\end{equation}
Here, to highlight $\eta$ is given, we denote $G^{[\eta]}(\*\beta, \*u^*)=
G((\eta,\*\beta^\top)^\top, \*u^*)$, $T^{[\eta]}(\*u^*, \*\beta) = T(\*u^*, (\eta,\*\beta^\top)^\top)$ and $B_{1 - \alpha}^{[\eta]}(\*\beta) = B_{1 - \alpha}((\eta,\*\beta^\top)^\top)$. 
If $\eta = \eta_0$, 
$\Gamma_{1-\alpha}^{[\eta]}(\*y_{obs})$ is a level $1 - \alpha$ confidence~set of $\*\beta_0$. However, we do not know whether the given $\eta$ is $\eta_0$, 
so we instead refer to $\Gamma_{1-\alpha}^{[\eta]}(\*y_{obs})$ in (\ref{eq:represent}) as a {\it representative set} of $\*\beta_0$. 
To get a level $1 -\alpha$ confidence set for $\*\beta_0$ with an unknown $\eta_0$ and to accommodating the uncertainty of estimating $\eta_0$, we consider
\begin{align}  \label{eq:conf_beta_combine}
& \Xi_{1-\alpha, 2}(\*y_{obs})  =  \cup_{\*\eta \in \widehat{\it \Upsilon}(\*y_{obs})} \Gamma_{1-\alpha}^{[\eta]}(\*y_{obs}) \nonumber \\
& \qquad = \big\{\*\beta: \exists \, \*u^* \in {\mathcal U} \mbox{ s.t. }  \*y_{obs} = 
G^{[\eta]}(\*\beta, \*u^*),
\,  
T^{[\eta]}(\*u^*, \*\beta) \in B_{1 - \alpha}^{[\eta]}(\*\beta),  \eta \in \widehat{\it \Upsilon}_{\mathcal V}(\*y_{obs}) \big\}.
\end{align}
Alternatively,
we may replace $\eta \in \widehat{\it \Upsilon}_{\mathcal V}(\*y_{obs})$ with $\eta \in \Xi_{1-\alpha', 1}(\*y_{obs})$, where $0< \alpha' < \alpha$ and $\Xi_{1-\alpha, 1}(\*y_{obs})$ is a level-$\alpha'$ confidence set from $\eta_0$ obtained using Step 1. We then construct a representative set of $\*\beta$, say $\Gamma_{1-\alpha''}^{[\eta]}(\*y_{obs})$, as in (\ref{eq:represent}) for $\alpha'' =  \alpha - \alpha'$. 
Now, 
 the set $\Xi_{1-\alpha, 2}(\*y_{obs})$ in 
(\ref{eq:conf_beta_combine}) is replaced by 
\begin{align}
\label{eq:conf_beta_combine-prime}
& \Xi_{1-\alpha, 2}(\*y_{obs})  =  \cup_{\*\eta \in \Xi_{1-\alpha, 1}(\*y_{obs})} \Gamma_{1 -\alpha''}^{[\eta]}(\*y_{obs}) \nonumber \\
& \quad = \big\{\*\beta: \exists \, \*u^* \in {\mathcal U} \mbox{ s.t. }  \*y_{obs} = 
G^{[\eta]}(\*\beta, \*u^*),
\,  
T^{[\eta]}(\*u^*, \*\beta) \in B_{\alpha''}^{[\eta]}(\*\beta),  \*\eta \in \Xi_{1-\alpha', 1}(\*y_{obs}) \big\}. 
\end{align}
Since $\widehat{\it \Upsilon}(\*y_{obs})$ can be viewed as a level $100\%$ confidence set for $\eta_0$, (\ref{eq:conf_beta_combine}) 
is a special case of (\ref{eq:conf_beta_combine-prime}) with $\alpha' = 1$ and $\alpha'' = \alpha$. 
Because $\*\beta$'s length may vary for different $\eta$ values, $\*\beta$'s parameter space ${\mathcal B}$ (also the space $\Gamma_{1-\alpha''}^{[\eta]}(\*y_{obs})$ is in)  
may have different dimensions for different $\eta$ values. Thus, $\Xi_{\alpha,2}(\*y_{obs})$ in (\ref{eq:conf_beta_combine}) and (\ref{eq:conf_beta_combine-prime}) is possibly a union of sets of different~dimensions. 

In Step 3, the construction of a joint confidence set for $\*\theta_0 =(\eta_0, \*\beta_0^\top)^\top$
 is similar to that of (\ref{eq:conf_beta_combine-prime}), but the confidence set is a subset in the space $\Theta = {\it \Upsilon} \otimes {\mathcal B}$, instead of ${\mathcal B}$; thus we use a product instead of a union to construct the joint confident set:
\begin{align}
    \label{eq:conf_beta_combine-prime2}
& \Xi_{\alpha, 3}(\*y_{obs}) = \Xi_{1-\alpha', 1}(\*y_{obs}) 
\otimes \Gamma_{1-\alpha''}^{[\eta]}(\*y_{obs})  
\\ & = \big\{\*\theta = (\eta, \*\beta): \exists \, \*u^* \in {\mathcal U} \mbox{ s.t. }  \*y_{obs} = 
G^{[\eta]}(\*\beta, \*u^*),
\,  
T^{[\eta]}(\*u^*, \*\beta) \in B_{\alpha''}^{[\eta]}(\*\beta),  \*\eta \in \Xi_{1-\alpha', 1}(\*y_{obs}) \big\}. 
\nonumber 
\end{align}
Here, the product $\otimes$ refers to that an element in ${\it \Upsilon}$ and a (corresponding) element in ${\mathcal B}$ jointly form an element in $\Theta$.  
Also, $(\alpha', \alpha'')$ and other terms involved are defined as in (\ref{eq:conf_beta_combine-prime}). 

Coverage statements of the above confidence sets are below, with a proof in Appendix~\ref{sec:proof_sec3}. 
\begin{corollary}
\label{cor:3steps}
Under this subsection's setup and the conditions in  
Theorems~\ref{thm:1} and~\ref{thm:p_bound_C_D_general},~we have (a) $\P\{ \eta_0 \in \Xi_{1-\alpha, 1}(\*Y) \} \ge 1 -\alpha - o_p(e^{-c_1 |{\mathcal V}|});$
(b) $\P\{ \*\beta_0 \in \Xi_{1-\alpha, 2}(\*Y) \} \ge 1 - \alpha  - o_p(e^{-c_2 |{\mathcal V}|});$ 
(c) $\P\{ \*\theta_0 \in \Xi_{\alpha, 3}(\*Y) \} \ge 1 - \alpha - o_p(e^{-c_3 |{\mathcal V}|}).$ 
Here, $c_i > 0$ is some constant for~$i= 1,2,3$. 
\end{corollary}
 
We remark that the three-step procedure can be extended to the more general case of $\*\theta_0 = (\eta_0, \*\beta_0^\top, \*\varsigma_0^\top)^\top$, where the inference targets are $\eta_0$, $\*\beta_0$ and $(\eta_0, \*\beta_0^\top)^\top$, and $\*\varsigma_0$ are additional nuisance parameters.   In this case, we need to slightly modify the formulas in (\ref{eq:represent}) - (\ref{eq:conf_beta_combine-prime2}), and still have results similar to Lemma~\ref{cor:3steps}. 
Due to space limit, the general formulas are not included in the paper; but we instead will consider some specific examples in Section~\ref{sec:mixture}.  
For instance, in Section~\ref{sec:mixture-mu-sigma-brief} (also Section~\ref{sec:mu-sigma}), 
the target parameters are 
$(\tau_0, \mu_{j}^{(0)})$,~for a 
given $j$, $1 \leq j \leq \tau_0$,  where $\tau_0$ is the 
unknown number of components and $\mu_j^{(0)}$ is the location parameter of the $j$th component in a Gaussian mixture model. In the example,~$\eta_0 = \tau_0$, $\*\beta_0 = \mu_{j}^{(0)}$ and $ \*\varsigma_0 =$ $ (\*M_{\tau_0}, \*\mu_{(-j)}^{(0)}, \*\sigma_0)$,  where  
$\*M_{\tau_0}$ is the membership assignment matrix, $(\*\mu_0, \*\sigma_0)$ are the location and scale parameters of the $\tau_0$ components, 
and
$\*\mu_{(-j)}^{(0)}$ is $\*\mu_0$ minus $\mu_{j}^{(0)}$; see Section~\ref{sec:mixture} for more discussions. 
Another example is in \cite{wang2022highdimenional} on the high-dimensional linear model,  in which a repro samples method is used to construct a confidence~set for a single coefficient, say $\*\beta_j$, corresponding to covariate $X_j$. There, $\eta_0$ is the true sparse model, $\*\beta_0 = \beta_{j}^{(0)}$ and $\*\varsigma_0$ is the remaining $p-1$ regression coefficients plus~the~error~variance.

\subsection{A practical guide and useful Monte-Carlo techniques}\label{sec:guide}

Based on the theory, the repro samples method is broadly applicable and covers most settings in which classical inference methods apply. In practice, it is particularly valuable in situations where the classical procedures are infeasible or perform poorly. In this subsection, we provide some guidelines and useful Monte-Carlo techniques for the practical use of the method.

To use a repro samples method, we often begin with a question: whether a classical~inference approach is feasible and, if feasible, whether its performance  is satisfactory.
If the answer is no,  the repro samples method could be a good alternative. As discussed~earlier, most classical approaches rely on asymptotic theories with regularity conditions, and they may not perform well when the conditions are violated. Also, for discrete or non-numerical parameters and non-numerical data, inference approaches are limited. The repro samples method is most helpful in these cases. 

In a classical inference approach, we often start with a probabilistic model, then identify a related test statistic and its distribution. Similarly, in our implementation, we also start with a model, either in a generative model form (\ref{eq:1}) or the more general setup of (\ref{eq:quantile})~or~(\ref{eq:AA}); we then identify 
a nuclear mapping $T(\*U, \btheta)$ for our inference problem and find a corresponding level $1 - \alpha$ Borel set $B_{1 -\alpha}(\btheta).$ The choice of $T(\*U, \btheta)$ is similar to the choice of a test statistic in the classical approach that is problem specific. 
What $T(\*U, \btheta)$ we use may~affect the efficiency and size of confidence set, but the validity of the confidence set is always~guaranteed.

A very simple choice is $T(\*U,\btheta)=T(\*U)$, free of $\btheta$,  as in Examples~\ref{ex:3} and~\ref{ex:crq} (Appenedix~\ref{sec:example_sec2}); however, the suitability 
is problem dependent, and sometimes it may be~inefficient  
or, in extreme cases, even can yield overly large and impractical confidence sets.
Another common choice is $T(\*U, \btheta) = \widehat {\btheta} (\*Y_{\*\theta}),$~a point estimator of $\btheta$ using sample $\*Y_{\*\theta} = $ $G(\*U, \btheta)$,  hopefully an efficient estimator (for efficiency considerations). Moreover, in the case with nuisance parameters (Section~\ref{sec: nuisance_control}), we may use
$T(\*u, \btheta) = \widehat{\eta} (\*Y_{\*\theta})$ to construct $\Gamma_{1 - \gamma}(\*Y_{\*\theta})$
in (\ref{eq:nu}) and obtain the profile nuclear mapping $T_p(\*u, \btheta)$ 
in (\ref{eq:T_p}). Here, and $\widehat{\eta}$ is an estimator of the target parameter $\eta$.  Conditional methods may also be used to deal with nuisance parameters in some cases; cf.,  \cite{wang2022highdimenional}. 
We may also~explore other options when the point-estimator approach is undesirable; cf., Examples~\ref{ex:3a} and \ref{ex:quantile_privacy} of Appendix~\ref{sec:example_sec2}.

After setting up the model and nuclear mapping, a generic implementation algorithm~is:

\vspace{2mm}
\begin{algorithm}
[ht] \caption{A generic implementation  algorithm}
\label{alg:Ag} 
\begin{algorithmic}[1] 
\State For a given value of $\btheta \in \Theta$, (a) Obtain 
a level $1 - \alpha$ Borel set $B_{1 -\alpha}(\btheta)$ in (\ref{eq:B});  
 (b) Check whether there exists a $\*u^* \in {\mathcal U}$ such that 
${\*y_{obs}} = G({\btheta}, \*u^*)$ and $T(\*u^*, \btheta)  \in B_{1 - \alpha}(\btheta)$. If both of the above criteria are satisfied, keep the $\btheta$.
\State Collect all kept $\btheta$ to form a level $1 - \alpha$ confidence set $\Gamma_{1-\alpha}(\*y_{obs})$. 
\end{algorithmic}
\end{algorithm}
\noindent
When the distribution of the nuclear mapping $T({\*U}, \btheta)$ is explicit, the confidence set $\Gamma_{1-\alpha}(\*y_{obs})$ often also has an explicit form; cf., e.g., Examples~\ref{ex:bin}, \ref{ex:3}, \ref{ex:3a} and~\ref{ex:crq}. 
Under more complex settings with the distribution of $T({\*U}, \btheta)$ not explicitly available, we may need to use Monte-Carlo methods. For instance, 
in Step 1(a),
we can typically collect many simulated copies~of  ${\*u^s} \sim \*U$ to form a set ${\mathcal V}$.
We then compute  $\{T({\*u^s}, \btheta), {\*u^s} \in {\mathcal V}\}$, for a given $\btheta \in \Theta$ and derive an empirical distribution of $T(\*U, \btheta)$, based on which we can obtain a level $1 - \alpha$ Borel set $B_{1 -\alpha}(\btheta)$. 
In particular, 
if $T(\*u, {\btheta})$ is a scalar, we can directly 
use its (upper/lower) quantiles of 
to get a level $1 - \alpha$ interval as the Borel set $B_{1 - \alpha}({\btheta})$. 
When $T(\*u, {\btheta})$ is a vector in ${\mathcal T} \subseteq \RR^q$,
we may use the data depth approach 
and follow \citet[][\S 3.2.2]{Liu2021}
to construct 
the Borel set $B_{1 - \alpha}({\btheta})$; see further details in Appendix~\ref{sec:data_depth}. Alternatively, we may also get an (empirical) central region 
$B_{1 -\alpha}({\btheta})$ on ${\mathcal T}$ using  {\it empirical center-outward quantiles} (developed based on optimal transport mapping); cf.,  \cite{hallin2021distribution} for further details.

 When the distribution of $T(\*U, \btheta)$ depends on $\btheta$ and the corresponding $\Gamma_{1-\alpha}(\*y_{obs})$ does not have an explicit form, Algorithm~1 needs to search over $\btheta \in \Theta$, which can be computationally expensive.
Fortunately, in most situations, we can avoid such an exhaustive search over $\btheta$.  
In this paper, we highlight two techniques that can help mitigate this potential computational issue: one is the previously discussed technique in Section~\ref{sec:candidate_general} for discrete or non-numerical parameters where candidate sets are used to reduce search spaces, and the other is for continuous parameters using a quantile regression approach implemented in \cite{dalmasso2022}. In the case of mixed types of parameters, the two techniques can be used together.  

Specifically, on the quantile-regression technique, suppose our target parameter $\btheta$ is in a continuous space $\Theta$ and that 
$T(\cdot, \btheta)$  is continuous in $\btheta$. We use the following  technique to improve our computation efficiency: {\it 
Let $\pi(\btheta)$ be a proposal distribution that has no-zero probability support on $\Theta$. In addition to $|{\mathcal V}|$ copies of $\*u^s \sim \*U$, we  simulate $|{\mathcal V}|$~copies of $\btheta^s \sim \pi(\btheta)$. We then fit a level $1 - \alpha$ nonparametric quantile regression model with~$t^s = T(\*u^s, {\btheta}^s)$ as response and ${\btheta^s}$ as covariates to output the level $1- \alpha$ conditional quantile curve, say $\widehat C_{1-\alpha}(\theta)$. 
The Borel set $B_{1 -\alpha}(\btheta)$ for our repro samples method is either $\big(-\infty, \widehat C_{1 -\alpha}(\btheta)\big)$ or $\big(\widehat C_{\frac\alpha2}(\btheta), $ $\widehat C_{1 -\frac\alpha2}(\btheta)\big)$.}
This method is similar to Algorithm 1 of \cite{dalmasso2022} where the~authors use a quantile regression to get a level $\alpha$ rejection region in a Neyman test problem. The difference is that the $t^s$ in \cite{dalmasso2022} is computed from a test statistic, while~our $t^s$ comes from a nuclear mapping. Therefore we need not generate any data~points $\*y^s = G(\btheta^s, $ $\*u^s)$, making it simpler, especially when generating~$\*y^s$ is computationally expensive. See Appendix~\ref{sec:bin_example} for an illustration of the method for Example~\ref{ex:bin}.
\citet[][Theorems~1]{dalmasso2022} show that, 
as $|{\mathcal V}| \to \infty$, $\widehat C_{q}(\btheta)$ is the conditional level $q$ quantile, if the quantile regression estimator is consistent. Similarly, in our case,  if the quantile regression estimator~is consistent, then $B_{1 -\alpha}(\btheta) = \big(- \infty, \widehat C_{1-\alpha}(\theta)\big)$ or $B_{1 -\alpha}(\btheta) =\big(\widehat C_{\frac \alpha2}(\btheta), \widehat C_{1-\frac\alpha2}(\btheta)\big)$ satisfies (\ref{eq:B}). Furthermore, the  machine learning technique to learn a test statistic in {\it Simulation Based Inference} (SBI) \citep{dalmasso2022, tomaselli2025robustsimulationbasedinference} may also be adopted to learn nuclear mapping $T(\*U, \btheta)$; cf., \cite{ tomaselli2025robustsimulationbasedinference} for further details on SBI.  More detailed discussion of the quantile regression technique can further be found in Appendix~\ref{sec:quantile}.

When $\btheta$ is discrete or non-numerical, we can directly apply the technique in Section~\ref{sec:candidate_general} to significantly reduce the search space of $\btheta$. More commonly when $\btheta = (\eta, \*\beta)$ is mixed~types, the three-step approach in Section~\ref{sec:ThreeSteps} is effective. To save computing costs, in Step~1 we can use the technique from Section~\ref{sec:candidate_general}, and in Steps~2-3, for a given $\eta$, we may use the quantile regression technique described above, especially when the representative set of $\*\beta$ does not have an explicit form. See Sections~\ref{sec:mixture} and~\ref{sec: numerical} below for some implementation~details.

\section{Gaussian mixture model with unknown number of components} 
\label{sec:mixture}

The Gaussian mixture model 
has a long history, 
and frequently used. 
However, 
making inference for the unknown number of mixture components, say $\tau_0$,  
remains  
a ``highly non-trivial'' problem \citep[cf.,][]{wasserman2020universal}.  
Although there exist several 
procedures to obtain point estimators of $\tau_0$, assessing the uncertainty associated with these estimations remains largely open.
This difficulty stems from the fact that the parameter $\tau_0$ is a discrete integer 
and the standard tools relying on CLT are no longer applicable.
Moreover, the uncertainty in $\tau_0$
 propagates to the inference on other model parameters, 
 whose dimension and structure depend on the unknown $\tau_0$.
 This section addressed these challenges by using the repro samples method to construct a confidence set for $\tau_0,$  and we also provide inference solutions for other parameters while 
accommodating the uncertainty in the unknown $\tau_0$.

Consider a $\tau$-component Gaussian mixture model: $Y_i, i=1,\ldots,n$, 
is drawn from one of the component distributions $\{N(\mu_k,\sigma_k^2): k=1,\ldots,\tau\}.$
Let $M_\tau=(m_{ik})\in\{0,1\}^{n\times\tau}$ be the (unobserved) membership matrix, with
$m_{ik}=1$ if $Y_i \sim N(\mu_k,\sigma_k^2)$ and $m_{ik}=0$ otherwise. Then, the model parameter is $\*\theta=(\tau, \*M_\tau, {\*\mu}_\tau, {\*\sigma}_\tau^2)^\top,$
where ${\*\mu}_\tau=(\mu_j)_{\tau \times 1}$ and ${\*\sigma}_\tau^2=(\sigma_j^2)_{\tau \times 1}$. 
And the true values  $\*\theta_0=(\tau_0, \*M_0, {\*\mu}_0, {\*\sigma}_0^2)^\top$, where $\*\mu_0 = (\mu_j^{(0)})_{\tau \times 1}, {\*\sigma}_0^2 = ({\sigma_j^{(0)}}^2)_{\tau \times 1}. $

We can re-express $Y_i$ as a generative model $Y_i = \sum_{k =1}^\tau m_{ik} (\mu_k + \sigma_{k}U_{i})$, $U_i \sim N(0,1)$. In a vector form, ${\*Y} = {\*M} {\*\mu} + \diag(\*U)\*M  {\*\sigma} = {\*M} {\*\mu} + \diag(\*M {\*\sigma}) \*U,$ 
where ${\*Y} = (Y_1, \ldots, Y_n)^{\top}$ 
 and $\*U = (U_1, \ldots, $ $ U_n)^{\top}$. 
The realization 
of $\*Y$ with true $\*\theta_0$ is 
${\bm y_{obs}} = {\bm M}_{0} {\bm \mu}_{0} + \, \diag({\bm  u}^{rel}){\bm M}_{0} \*\sigma_0.$
An alternative  
setup is $Y_i \sim \sum_{k =1}^\tau p_{k} N(\mu_k, \sigma_k^2)$, which further assumes that ${\*m}_i= (m_{i1}, \ldots, $ $m_{i\tau})^{\top}$ is a random draw from multinomial distribution $\text{Multinomial}(\tau; p_1, \ldots,$ $ p_{\tau})$ with unknown proportion parameters $(p_1, \ldots, p_{\tau})$, $\sum_{k =1}^\tau p_k = 1$ \citep{chen_inference_2012}. 
In this case, 
$\*M_0$ is the realized membership assignment 
matrix that generates $\*y_{obs}$.

It is known in the literature that there is an identifiability issue in a mixture model; i.e., multiple sets of parameters could possibly generate the same data \citep{LiuShao2003}.~To account for this  identifiability 
issue, 
we re-define the true parameters as 
the ones with the smallest $\tau$, i.e., 
\begin{align}
\label{eq:C1}
(\tau_0, \*M_0, {\*\mu}_0, \*\sigma_0^2) =  
{\arg \min}_{\{(\tau, \*M_{\tau}, {\*\mu}_{\tau}, \*\sigma^2)|  {\*M}_{\tau} {\*\mu}_{\tau} = {\*M}_{0} {\*\mu}_{0}, \,   {\*M}_{\tau} {\*\sigma}_{\tau} = {\*M}_{0} {\*\sigma}_{0}\} }  \tau.   
\end{align}
Here, $\*M_0$ and ${\bm \mu}_0$ on the right-hand side are the true parameter values that generate $\*y^{obs}.$ 
  With a slight abuse of notation (but for simplicity), we still use the same notations on the left-hand side to denote the ones 
with the smallest $\tau$. 
We further assume that $(\tau_0, \*M_0, {\*\mu}_0, \*\sigma_0^2)$ defined in (\ref{eq:C1})
is unique and 
 $0 <  \sigma_{\max} = \|\*\sigma_0\|_\infty < \infty.$
Thus, among the $\tau_0$ components, 
no $\sigma_k^2$
can completely dominate another $\sigma_{j}^2$, for any $j \not= k$. 
Moreover, we follow the existing  literature to assume that the smallest number of components~in~\eqref{eq:C1}~$\tau_0 = O(\log(n))$, since increasing $\tau_0$ beyond $O(\log(n))$ only delivers indistinguishable gains in the likelihood function and one can confine attention to mixing distributions with no more than $O(\log(n))$ components; cf., \cite{ghosal2001entropies} and \cite{wu2025nonparametric}.
The inference procedure we develop in this paper is for the set of parameters defined in (\ref{eq:C1}). Due to space limit, 
we only include in the main text our inference for $\tau_0$,
with the detailed discussion for the location and scale parameters $\{(\mu_j^{(0)},\sigma_j^{(0)}): j=1,\ldots,\tau_0\}$ placed in  Appendix~\ref{sec:mu-sigma}.

\subsection{A confidence set for the unknown number of components $\tau_0$}  \label{sec:components}

When we work on $\tau$,  the other $(\*M_\tau, \*\mu_\tau, \*\sigma_\tau)$ are nuisance parameters whose dimensions are functions of $\tau$.  Given $ (\tau, \*M_\tau)$, we can obtain sufficient statistics for $(\*\mu_\tau,\*\sigma^2_\tau)$ by projecting~$\*Y$ onto the space expanded by $\*M_\tau$. Thus, $T(\*U, \btheta)= \widetilde T(\*Y, (\tau,{\*M_\tau}))$ is free of $(\*\mu_\tau, \*\sigma_\tau)$.~We~first use the profiling method in Section~\ref{sec: nuisance_control} (setting $\eta=\tau$ and $\*\beta=\*M_\tau$) to handle nuisance parameter $\*M_\tau$, and then use a conditional method to address nuisance parameter $(\*\mu_\tau,\*\sigma_\tau)$. 

Assume, temporarily, we can obtain a nuclear mapping $T(\*U, \btheta)= \widetilde T(\*Y, (\tau,{\*M_\tau}))$
and have a corresponding Borel set $B_{1 -\alpha}$ that satisfies $\P(T(\*U, \btheta) \in B_{1 -\alpha}) \geq 1- \alpha.$
Then,~following Section~\ref{sec: nuisance_control}, 
for $\*y$ that is generated from given parameters $(\tau, \*M_\tau)$, 
we test whether it is~from $ (\tau, \widetilde{\*M}_\tau)$,  $\widetilde{\*M}_\tau \not= \*M_\tau$. As in \eqref{eq:nu} and \eqref{eq:T_p}, we define a profile nuclear mapping~for~$\tau$: 
\begin{align}
\label{eq:nuclear_final_mixed}
    T_p(\*u, \*\theta) = \widetilde T_p(\*y, \tau) = \min\nolimits_{\widetilde{\*M}_\tau} \widetilde \nu(\*y, \tau,  \widetilde{\*M_\tau}), 
\end{align}
where $ \nu(\*u, \tau,  \widetilde{\*M_\tau})= \widetilde \nu(\*y, \tau,  \widetilde{\*M_\tau})  =  \inf\{1-\gamma: \widetilde T(\*y, (\tau,\widetilde{\*M_\tau})) \in B_{1-\gamma}\}.$
By Lemma~\ref{lemma:nuisance}~and~\eqref{eq:CIeta},
\begin{align}
\Xi_{1-\alpha}(\*y_{obs})   & = 
     \big\{\tau: \widetilde T_p(\*y_{obs},\tau)\leq 1-\alpha \big\};
   \label{eq:mix4}
\end{align}
detailed steps of deriving \eqref{eq:mix4} using \eqref{eq:CIeta} are provided in Appendix \ref{sec:27derivation}. 
Corollary~\ref{cor:mix4_new} below states~that $\Xi_{1-\alpha}(\*y_{obs})$ in (\ref{eq:mix4}) is a 
level $1 - \alpha$ confidence set; A proof is in Appendix~\ref{sec:B3-1}.
\begin{corollary}
\label{cor:mix4_new}
Let $\*y =\*M_{\tau} \*\mu + \diag(\*M_\tau  {\*\sigma}) \*u$,  for fixed $\btheta = (\tau,{\*M_\tau}, \*\mu, \*\sigma)$, and $T(\*U, \btheta) = \widetilde T(\*y,  (\tau,{\*M_\tau}))$. Suppose $\P(T(\*U, \btheta) \in B_{1 -\alpha}) \geq 1- \alpha$ for a Borel set $B_{1 -\alpha}$, given $\*\theta.$  If  $\widetilde T_p(\*y, \tau)$ is defined by \eqref{eq:nuclear_final_mixed} and $\Xi_{1-\alpha}(\*y_{obs})$  is defined by \eqref{eq:mix4}, then $\P(\tau_0 \in \Xi_{1-\alpha}(\*Y)) \geq 1-\alpha$ for a random sample $\*Y$ that is generated with the set $\*\theta = \*\theta_0 = (\tau_0, {\*M_0}, \*\mu_0, \*\sigma_0)$. 
\end{corollary}

We now move on to find  $T(\*U, \btheta) = \widetilde T(\*y,  (\tau,{\*M_\tau}))$ that satisfies the requirement in~Corollary~\ref{cor:mix4_new}. 
To do so, let $\hat\tau(\*y) $ be an 
estimator of $\tau$; in our paper we a BIC-type estimator:
\begin{align}
    & \hat \tau(\*y)  = \arg\min_{\tau  \in {\mathcal T}}\left[\left\{-2 \max\nolimits_{(\tau, \*M _\tau, {\bm \mu }, \*\sigma ) } \log \ell(\*M _\tau, {\*\mu }, \*\sigma |\*y)\right\} + 2\tau  \log(n)\right], \nonumber \\
    & \,\,\,\, =  \arg\min_{\tau  \in {\mathcal T}} \min_{\*M _\tau}\left\{2\sum\nolimits_{k=1}^{\tau } n_k\log\left(\|(I - \*H_{\tau}^{(k)})\*y^{ (k)}\|\right) - \sum\nolimits_{k=1}^{\tau }n_k\log n_k + 2\tau  \log(n) \right\}, \label{eq:BIC-est}
\end{align}
where 
$\ell(\tau, \*M _\tau, {\bm \mu} , \*\sigma |\*y ) = \sum\nolimits_{k=1}^{\tau } \{-n_k\log(\sigma^{(k)})  -\frac{1}{2(\sigma^{(k)})^2} \|\*y^{ (k)} - \mu^{(k)}\|^2\}$ is the log-likelihood function of {\small $(\tau, \*M _\tau, {\bm \mu }, \*\sigma )$} given {\small $\*y$,}   $\*H_\tau^{(k)}$$=$$ \*H_\tau[\mathcal I_k, \mathcal I_k]$ is the submatrix of $\*H_\tau = {\*M}_{\tau} ({\*M}_{\tau}^\top {\*M}_{\tau})^{-1}{\*M}_{\tau}^\top$ corresponding to the $k$th component, and ${\*y}^{ (k)} = \*y [\mathcal I_k]$ is the corresponding subvector of $\*y $. Here, $\mathcal I_k=\{i: m_{ik}=1\}$ is the index set
for component $k$ and, following the common practice in the literature \citep[cf.,][]{Chen2017, FraleyRaftery1998, FraleyRaftery2002}, we set ${\mathcal T} = \{1, \dots, K_{\max}\}$ for a fixed $K_{\max}$.
In the next paragraph,  we define our $T(\*U, \btheta) = \widetilde T(\*y,  (\tau,{\*M_\tau}))$ via this point estimator ${\hat\tau}(\*Y )$ given $ (\tau, \*M_\tau)$, 
while accommodating the unknown~$(\*\mu_\tau, \*\sigma_\tau)$. 

Suppose for the moment that we can compute $P_{{\*Y} }(\bar \tau) = \P_{{\*Y} }\{\hat\tau({\*Y} ) = \bar \tau\}$, 
for any $\bar \tau \in {\mathcal T}$, where $\P_{{\*Y} }\{\cdot\}$ 
refers the probability of the random sample $\*Y$ that is generated using $\*U$ and the set of fixed $\*\theta = (\tau, {\*M_\tau}, \*\mu, \*\sigma)$. 
Then, following the practice of defining a $p$-value, we assess how `extreme'  the realized $\hat\tau({\*y} )$ is, 
i.e., 
${\mathcal P}(\*y ) =  \sum\nolimits_{\{\bar \tau: \, P_{{\*Y} }(\bar \tau) > P_{{\*Y} }(\hat \tau(\*y )) \}} P_{{\*Y} }(\bar \tau).$
We expect the distribution of ${\mathcal P}(\*Y )$ is closely related to $U(0,1)$. However, 
$(\*\mu_\tau, \*\sigma_\tau)$ values are unknown, so we cannot directly compute the distribution $P_{{\*Y} }(\bar \tau)$, nor simulate $\*Y  =  \*M_{\tau} \*\mu_\tau + \diag(\*M_\tau  {\*\sigma_\tau}) \*U$, when we are only given $(\tau, \*M_\tau)$. 
 Nevertheless, when $(\tau, \*M_\tau)$ is given, $\big(\A_{\tau}(\*y ), $ $\B_{\tau}(\*y )\big)$ are sufficient statistics 
 for the unknown $(\*\mu_\tau, \*\sigma_\tau)$, where $\A_{\tau}(\*y ) = 
(\*M_{\tau}^\top \*M_{\tau})^{-1}\*M_{\tau}^\top\*y $ and
 $\B_{\tau}(\*y ) = \big(\|(I - \*H_{\tau}^{(1)})\*y^{ (1)}\|, $ $\cdots, \|(I - \*H_{\tau}^{(\tau)})\*y^{ (\tau)}\|\big)^T$
are the ``point estimators'' of $(\*\mu_\tau, \*\sigma_\tau)$ with observing sample data $\*y $ and given $(\tau, \*M_\tau)$. Thus, we can generate a repro sample $\*Y'$ using $(\tau, M_\tau, \A_{\tau}(\*y ), \B_{\tau}(\*y ))$ instead. Specifically, we simulate an independent copy ${\*U}' \sim \*U$
and let ${\*Y}' =  \*M_{\tau} \A_{\tau}(\*y) + \diag\left\{\*M_\tau\B_{\tau}(\*y)\right\}\C_{\tau}({\*U'}),$ where $\C_{\tau}(\*U') = \big(\frac{\{(I - \*H_{\tau}^{(1)})\*U'^{(1)}\}^T}{\|(I - \*H_{\tau}^{(1)})\*U'^{(1)}\|}, \cdots,$ $\frac{\{(I - \*H_{\tau}^{(\tau)})\*U'^{(\tau)}\}^T}{\|(I - \*H_{\tau}^{(\tau)})\*U'^{(\tau)}\|}\big)^T$ and $\*U'^{(k)} = \*U'[\mathcal I_k].$ 
We then define our mapping as 
\begin{align}
\label{eq:mix_nuc_new}
    T(\*U, \btheta) = \widetilde T(\*y, (\tau,{\*M_\tau})) =  \sum\nolimits_{\{\bar \tau: P_{{\*Y}'}(\bar \tau) > P_{{\*Y}'}(\hat \tau(\*y)) \}} P_{{\*Y}'}(\bar \tau). 
\end{align}
Since $\C_{\tau}({\*U'}) \sim \C_{\tau}(\*U)$ and $\C_{\tau}(\*U)$ is independent of $(\A_{\tau}(\*Y), \B_{\tau}(\*Y)),$ the repro samples 
${\*Y}'$ and ${\*Y}$
have the same conditional distribution, i.e.,  
${\*Y}' \big| \{\A_{\tau}(\*y),  \B_{\tau}(\*y)\} 
\sim \,\, \*Y \big| \{\A_{\tau}(\*Y) \B_{\tau}(\*Y)\}=\{\A_{\tau}(\*y),  \B_{\tau}(\*y)\}.$ 
Based on this fact, we have the theorem below; a proof is in  Appendix~\ref{sec:B3-1}. 
\begin{theorem}
\label{the: nuc_new_unif} Given $\*\theta$ and $T(\*U, \btheta)$  in (\ref{eq:mix_nuc_new}), we have
$\P(T(\*U, \btheta) \leq 1-\alpha) \geq 1-\alpha.$
\end{theorem}
 By Theorem~\ref{the: nuc_new_unif} and (\ref{eq:mix_nuc_new}),
$\P(\widetilde T(\*Y, (\tau,{\*M_\tau})) \in B_{1 -\alpha}) \geq 1-\alpha$, for $B_{1 -\alpha} = [0, 1-\alpha].$
Thus,  the function
$T_p(\*u, \*\theta)$   
in (\ref{eq:nuclear_final_mixed}) can be simplified to be $T_p(\*u, \*\theta)  
    = \min_{\widetilde{\*M}_\tau}\inf\{1-\alpha': \widetilde T(\*y, (\tau,\widetilde{\*M_\tau})) \leq $ $1-\alpha'\}
    = \min_{\widetilde{\*M}_\tau}\widetilde T(\*y, $ $ (\tau,\widetilde{\*M_\tau})).
$
Therefore, the confidence set in \eqref{eq:mix4} becomes 
\begin{equation}
    \label{eq:mix-simple}
\Xi_{1-\alpha}(\*y_{obs}) 
 = \big\{\tau: \min\nolimits_{_{{\*M}_\tau}}\widetilde T(\*y_{obs}, (\tau,{\*M_\tau}))\leq 1-\alpha \big\}.
 \end{equation} 

To obtain \eqref{eq:mix-simple}, it 
 suffices to  compute $\widetilde T(\*y_{obs}, (\tau, {\*M_\tau}))$ for any given $(\tau, {\*M}_\tau)$.
  We use a Monte-Carlo method to compute $\widetilde T(\*y_{obs}, (\tau, {\*M_\tau}))$. Specifically, we simulate many~$\*u^s \sim$~$\*U$ and collect them to form a set $\mathcal V$. For each $\*u^s \in \mathcal V$ and given  $(\tau, {\*M}_\tau),$ we compute
$\*y^s = $ ${\*M}_\tau\A_{\tau}(\*y_{obs})+\diag\{{\*M}_\tau\B_{\tau}(\*y_{obs})\}\C_{\tau}(\*u^s)$
and $\hat\tau (\*y^s)$.  Then $\widehat P_{\mathcal V}(\bar \tau) = $ $\frac{1}{|\mathcal V|}\sum_{\*u^s \in \mathcal V}$ $I(\hat \tau(\*y^s) =\bar \tau)$~is an approximation of $P_{{\*Y}'}(\bar \tau)$ for 
{\small ${\*Y}' 
=  \*M_{\tau} \A_{\tau}(\*y_{obs}) + \diag\left\{\*M_\tau\B_{\tau}(\*y_{obs})\right\}\C_{\tau}({\*U'})$, ${\*U}' \sim \*U$}.~When $|\mathcal V|$ is large enough, we can approximate 
$
    \widetilde T(\*y_{obs}, (\tau, {\*M_\tau})) \approx \sum_{\{\bar \tau: \hat P_{\mathcal V}(\bar \tau) > \hat P_{\mathcal V}(\hat \tau(\*y_{obs})) \}}\hat P_{\mathcal V}(\bar \tau). 
$

\subsection{A candidate set of $(\tau_0,\*M_0)$ and a modified confidence set for $\tau_0$} \label{sec:finding_candidates}

To avoid searching through all values of $(\tau,\*M_\tau)$ in (\ref{eq:mix-simple}), 
we follow  Section~\ref{sec:candidate_general} to construct a data-driven 
candidate set $\widehat {\it \Upsilon}(\*y_{obs})$. In the notation of  Section~\ref{sec:candidate_general}, the target parameter for the development is $\eta = (
\tau, \*M_\tau)$. 
Corresponding to (\ref{eq:tau-star}), we use  a modified BIC as our lost function and define 
our many-to-one mapping for $\eta^* = (\tau^*, \*M^*_\tau)$. Specifically, we define
\begin{equation}
    \label{eq:modefied_BIC}
    \resizebox{.8\hsize}{!}{$(\tau^*, \*M^*_{\tau^*}) =\arg\min\limits_{(\tau, \*M_\tau)} \min\limits_{(\*\mu_\tau, \bar\sigma)}  \left\{n\log \big(\frac{\|\*y_{obs} - {\*M}_{\tau} {\*\mu}_{\tau} - \bar\sigma \*u^*\|^2
+1}{n}\big) + 2\lambda\tau\log(n)\right\}$}
\end{equation}
with a tuning $\lambda > 0$, where we only use a scalar $\bar \sigma$ instead of a vector $\*\sigma_\tau$ (which significantly reduces the complexities of computing and theoretical developments). 
Due to space limit, 
the rational of using (\ref{eq:modefied_BIC}) is provided in Appendix~\ref{sec:BIC}.
By \eqref{eq:cand-hat}, a  candidate set for $(\tau_0, \*M_0)$~is 
\begin{align}
\label{eq:candidate_M}
\resizebox{0.92\hsize}{!}{$\widehat{\it \Upsilon}_{{\mathcal V}}(\*y_{obs}) 
  =  \bigg\{ 
  (\tau^*, \*M^*_{\tau^*})  
  ={\footnotesize \arg\min\limits_{(\tau, \*M_\tau)} \min\limits_{(\*\mu_\tau, \bar\sigma)}  \bigg\{n\log \big(\frac{\|\*y_{obs} - {\*M}_{\tau} {\*\mu}_{\tau} - \bar \sigma \*u^s \|^2
+1}{n}\big) +
\lambda\log(n)(2\tau)
  \bigg\}, 
 \*u^s \in {\mathcal V}
  \bigg\}},$}
\end{align}
where ${\mathcal V} = \{\*u_1^s, \ldots, \*u_N^s\}$ is a collection of simulated $\*u_j^s\sim N(\*0, \*I)$, $j = 1, \ldots, N = |{\cal V}|$. 
Therefore, as in Theorem~\ref{thm:p_bound_C_D_general} and using (\ref{eq:mix-simple}),
we define a modified confidence set for $\tau_0$ as
\begin{align}
\label{eq:cs_mixture}
\Xi'_{1-\alpha}(\*y_{obs}) = \big\{\tau: \min\nolimits_{(\tau, {\*M}_\tau) \in \widehat{\it \Upsilon}_{{\mathcal V}}(\*y_{obs})}\widetilde T(\*y_{obs}, (\tau,{\*M_\tau}))\leq 1-\alpha \big\}.
\end{align}

Theorem~\ref{thm:bound_of_C_d_mixture} below states that $\widehat{\it \Upsilon}_{{\mathcal V}}(\*y_{obs})$ in (\ref{eq:candidate_M}) add $\Xi'_{1-\alpha}(\*y_{obs})$ in~(\ref{eq:cs_mixture}) have the desired coverage rates when $|{\mathcal V}|$ is  large enough. 
One complication preventing us from directly applying Theorem~\ref{thm:p_bound_C_D_general} here to the Gaussian mixture model is: Condition (N1) may not hold for two events of 
 arbitrarily small probabilities.  They are: (i) when $\*u^{rel}$ hits certain~directions ($\diag(\*M_0\*\sigma_0) \*u^{rel}$ and $\{I-\*M_{\tau}(\*M_{\tau}^\top\*M_{\tau})^{-1}$ $\*M_{\tau}^\top\}\*M_0\*\mu_0$ are on the same direction for some $\*M_{\tau} \not = \*M_0$, $\tau \leq \tau_0$); (ii) when $\|\*u^{rel}\|$ is extremely large. Fortunately, we can control~both events so condition (N1) still holds in probability. 
 Thus,  $\widehat{\it \Upsilon}_{{\mathcal V}}(\*Y)$ and 
$\Xi'_{1-\alpha}(\*y_{obs})$ maintain the desired coverage rates except for an arbitrarily small $\delta>0$. 
A proof is in Appendix~\ref{sec:B3-2}. 
\begin{theorem}
\label{thm:bound_of_C_d_mixture}
Suppose $n-\tau_0>4.$ For any arbitrarily small $\delta >0,$ there exist $\gamma_\delta>0, c>0$ and also an interval of positive width $\Lambda_{\delta} =[\frac{\gamma_{\delta}^{1.5}}{\log(n)/n},$ $\frac{\log\{0.5C^2_{\min}\gamma_{\delta} + 1\}}{2 \tau_0 \log(n)/n}]$ such that when $\lambda \in \Lambda_{\delta},$ 
(a) $\P_{\*U | {\mathcal V}}\{(\tau_0, \*M_0) \in \widehat{\it \Upsilon}_{{\mathcal V}}(\*Y)\}$ $\geq 1- \delta - o_p(e^{-c|\mathcal V|})$, and  (b) $\P_{\*U | {\mathcal V}}\{\tau_0 \in \Xi'_{1-\alpha}(\*Y)\}  \geq 1-\alpha- \delta - o_p(e^{-c|\mathcal V|}).$ 
Here, a random sample $\*Y$ is generated with the set $\*\theta = \*\theta_0 = (\tau_0, {\*M_0}, \*\mu_0, \*\sigma_0)$. 
Also, $\Xi'_{1-\alpha}(\cdot)$ is defined in \eqref{eq:cs_mixture} and $C_{\min} = 
{\min}_{\small \{\*M_\tau: \, |\tau| \leq |\tau_0|, \tau \neq \tau_0\}} 
\|(\*I - \*M_\tau)\*M_0\*\mu_0\| > 0$.
 \end{theorem}

Theorem~\ref{thm:bound_of_C_d_mixture} states that we can recover both $\tau_0$ and $\*M_0$ with a high certainty, if we have enough computing resource to allow for a large $|\mathcal V|$ . Nevertheless, getting each element in $\*M_0$ correctly is a much harder task than recovering $\tau_0$, requiring a much larger $|\mathcal V|.$ This is because, if an alternative ${\*M}_\tau$ is not too different from $\*M_0$, a repro sample $\*y^s$ generated using the ${\*M}_\tau$ can very well produce the same $\tau$ estimate as that generated by using $\*M_0.$  
In practice,
 our computing power is limited so we focus only the task of making inference for $\tau_0$. Our numerical studies in Section~\ref{sec: numerical} suggest that the confidence set of $\tau_0$ in \eqref{eq:cs_mixture} can achieve its desirable coverage rate empirically without~$|\mathcal V|$~being~excessively~large.

Finally, to gain further insight, we provides below an upper bound for  $\tau_{up}(\*y_{obs})$ {\small $=\sup\{\tau: \exists \*M_\tau, \mbox{ s.t. } (\tau, {\*M}_\tau) \in \widehat{\it \Upsilon}_{\mathcal V} \},$} 
the largest $\tau$ that can be included in the candidate set $\widehat{\it \Upsilon}_{\mathcal V}.$ 
\begin{theorem}
\label{the:upper_bound}
   If  $\lambda = O(n/\log(n))$ as in Theorem~\ref{thm:bound_of_C_d_mixture},  there exist constants $a, a'$ s.t. 
       $ \P(\tau_{up}(\*Y) \geq  \tau) \leq a \exp\{-a'(n^{O((\tau-\tau_0)/\log(n)) -1/2} - n^{1/2})+ \log |\mathcal V|\}, $ 
    where $\tau \geq \tau_0.$
\end{theorem}
\noindent
When $\tau - \tau_0 \geq O(\log(n)),$ the probability bound decays fast with $n,$ indicating the upper bound of the confidence set is $\tau_0 + O(\log(n))$ with large probability.

\subsection{\hspace{-3mm}
Inference for the component mean $\mu_j^{(0)}$ and standard deviation $\sigma_j^{(0)}$ }
\label{sec:mixture-mu-sigma-brief}

We use the three-step procedure in Section~\ref{sec:ThreeSteps} to conduct inference for $\mu_j^{(0)}$ and $\sigma_j^{(0)}$, while accommodating the uncertainty of the unknown $\tau_0$. To mitigate the impact of potential membership mislabeling issues in the mixture, more robust nonparametric quantile estimators are also employed in Steps~2-3. Due to space limit, details are provided in Appendix~\ref{sec:mu-sigma}.

\section{Numerical studies: real and simulated data}\label{sec: numerical}
We conduct a comprehensive numerical study using both real and simulated data.{\small~}The~real~data consist of activity measurements of red blood cell sodium–lithium countertransport (SLC) 
from 190 individuals \citep{dudley_assessing_1991}. The SLC measurement is known to be correlated with blood pressure and is considered by many researchers to be an important contributor to hypertension. The same SLC data set has previously been analyzed using Gaussian mixture models with an unknown number of components by \citet{roeder_graphical_1994} and \citet{chen_inference_2012}. 

We apply the repro samples method to the SLC data and obtain more comprehensive and informative inference results.
The simulation studies are under two settings that mimic the SLC data and demonstrate the superior empirical performance of our proposed method. Due to space limit, in the main text we present inference results for the discrete parameter $\tau_0$ only. Comparisons with existing approaches, along with additional inference results for $\mu_j^{(0)}$ and $\sigma_j^{(0)}$, are summarized in the main text, with full details provided in Appendix~\ref{sec:numerical-details}.

 \begin{figure}[!t]

     \centering     \includegraphics[width=12cm, height=7cm]{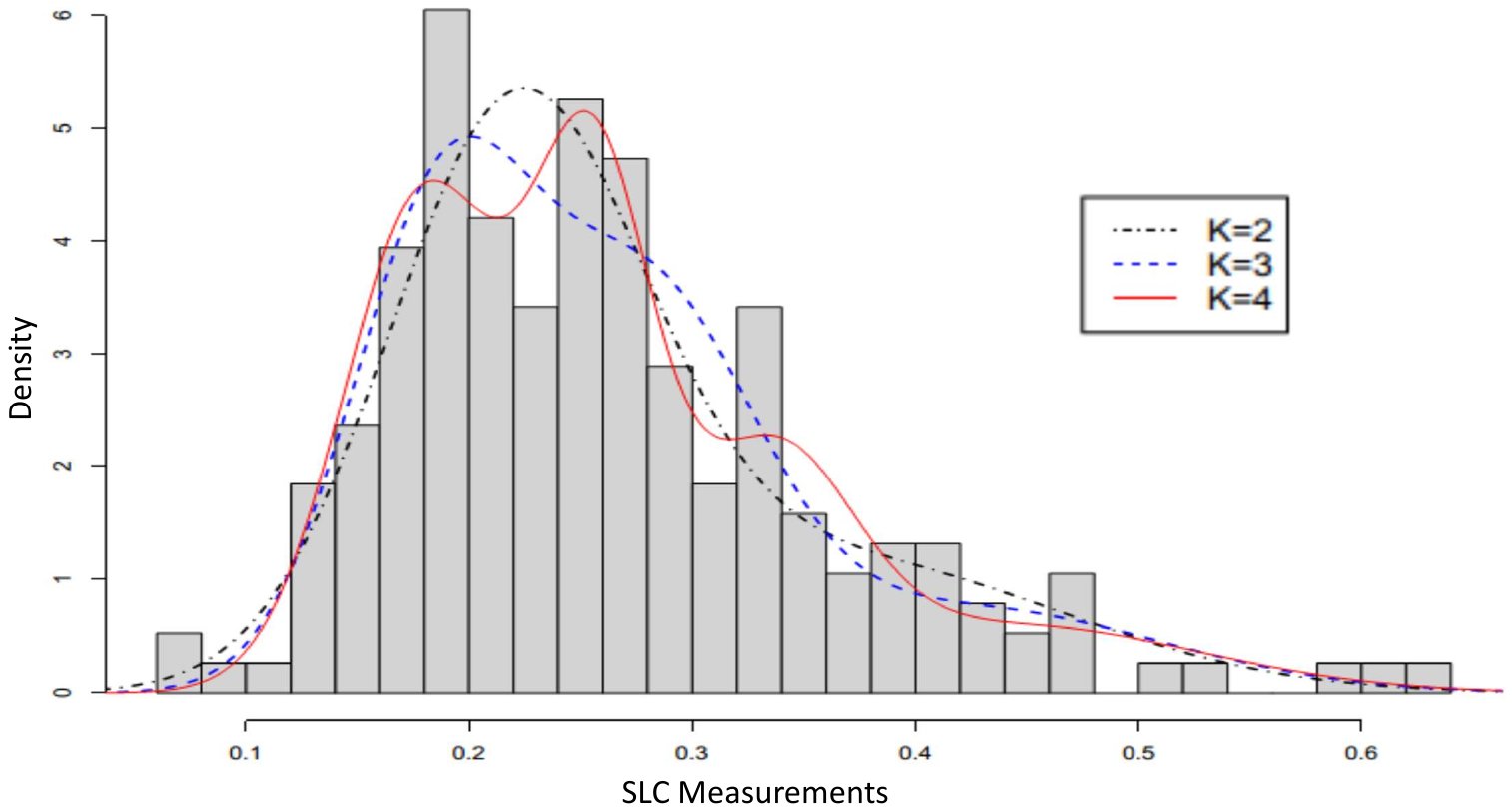}
       \captionsetup{font= small}
     \caption{Histogram of 190 SLC measurements and estimated Gaussian mixture density curves.} 
\label{fig:SLC_mixture} 
 \end{figure}

\subsection{Inference on the discrete unknown number of components $\tau_0$}\label{sec:tau-numerical}

\cite{roeder_graphical_1994} and \cite{chen_inference_2012} used the classical framework to test the unknown number of components $H_0: \tau_0 = k$ vs $H_1: \tau_0>k$~for~a~small integer $k.$ \cite{roeder_graphical_1994} concluded that $\tau_0=3$ is the smallest value of $\tau_0$ not rejected by the data under the assumption that $\sigma_j^2$'s are
the same.  
\cite{chen_inference_2012} concluded instead the smallest 
$\tau_0 =2$ without the equal variance assumption.
Our first goal is to obtain a   
95\% confidence set for the discrete unknown $\tau_0$. 
To do so we first obtain a candidate set for $(\tau_0, \*M_0)$. 
Particularly, for each $\tau$, we solve the objective function in  $\eqref{eq:modefied_BIC}$ via fitting a mixture linear regression model with $\*y_{obs}$ as the response vector and $\*u^*$ as the covariate  
using the R package {\it flexmix} \citep{grun2007fitting}, and we get $(\tau^*, \widehat M_{\tau}^*)$ by (\ref{eq:modefied_BIC}). The collection of these $(\tau^*, \widehat M_{\tau}^*)$'s~forms a candidate set $\widehat{\it \Upsilon}(\*y_{obs})$. To evaluate the nuclear mapping in \eqref{eq:BIC-est}, we set $K_{\max} = \tau_{up}(\*y_{obs}) \vee 10,$ with $\tau_{up}(\*y_{obs})$  defined in Theorem~\ref{the:upper_bound}. We then proceed and use formula (\ref{eq:cs_mixture}) to obtain a 95\% confidence set for $\tau_0$, which is $\Xi'_{1-\alpha}(\*y_{obs}) $ $= \{2, 3, 4\}.$
The confidence set  
 suggests
that both $\tau_0=2$ and~$3$ are plausible, consistent with earlier studies.  In genetics, a two-component mixture 
corresponds to a simple dominance model and a three-component corresponds to an additive model \citep{roeder_graphical_1994}.~Evidently,~we~cannot rule out these two models based on the SLC data alone. Our confidence set also includes
$\tau_0 = 4$.
Figure~\ref{fig:SLC_mixture} is the histogram of the SLC measurements together with estimated Gaussian mixture densities with $\tau =2,3,4$ using~an~EM algorithm (R package {\it ClusterR}).  Notably, only the  
model with $\tau=4$ captures all the three spikes in the histogram, while smoothly fitting the rest of the histogram.
Therefore, the model~with $\tau_0=4$ also demands further investigation, as it appears to 
well represent the data empirically. An advantage of our method is that our
confidence set 
provides both lower and upper bounds, for plausible  $\tau_0$  values supported by the data, 
whereas inverting the test methods proposed by \cite{roeder_graphical_1994} and \cite{chen_inference_2012} can only give us one-sided confidence~sets. 

\begin{figure}[!t]
     \centering     \includegraphics[width=0.75\textwidth, height=8.5cm]{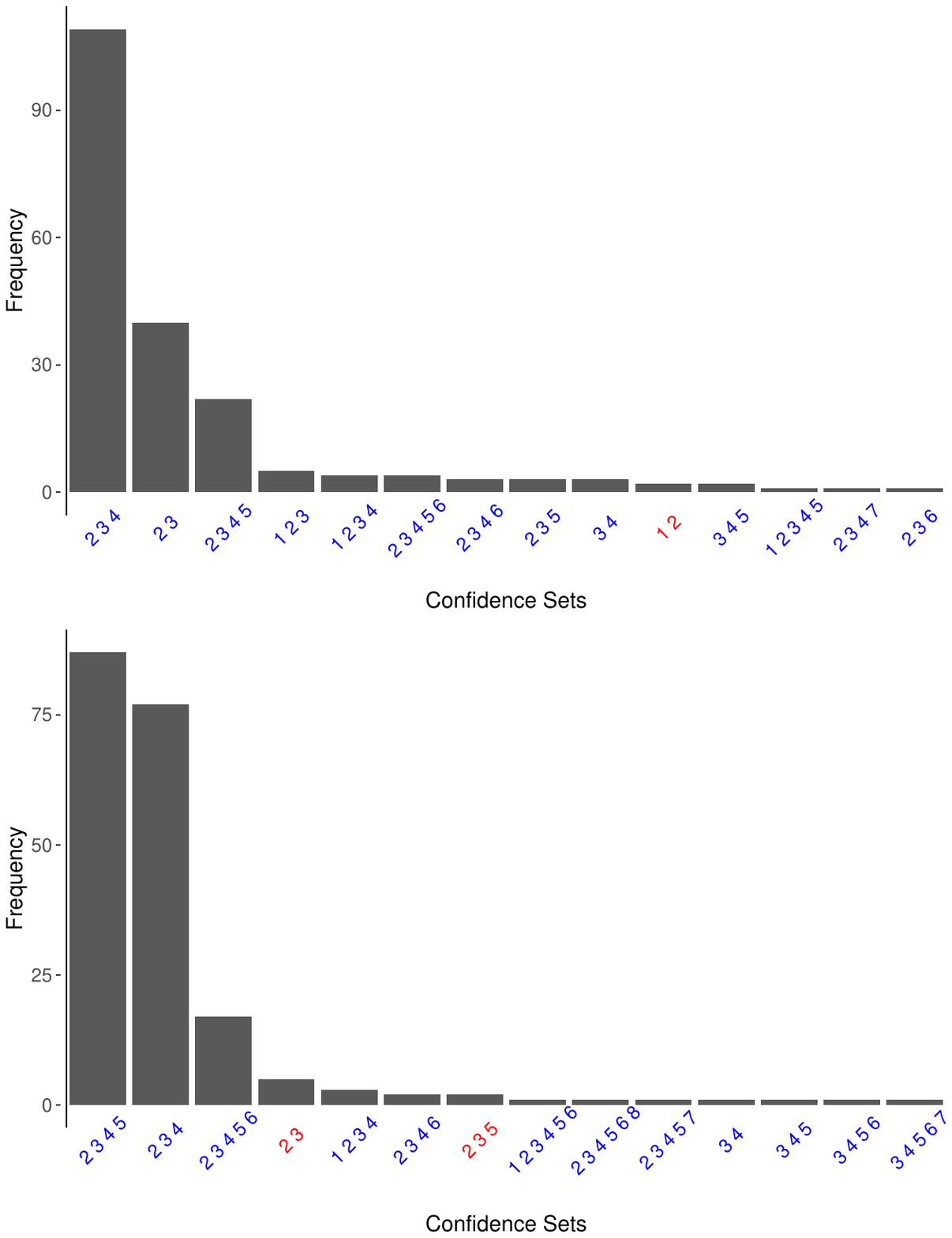}
            
    \caption{Bar plots of 200 level-95\% repro samples confidence sets obtained in the simulation study (200 repetitions). 
    The upper panel is for true $\tau_0=3$; lower for true $\tau_0=4.$} \label{fig:simulation_mixture} 
    \end{figure}

We use simulation studies to further demonstrate the performance of our proposed method, where the simulation settings mimic that of the real SLC data. Specifically, the true parameters of our simulation 
models are those  in Table~\ref{tab:SLC}(a)  (from fitting SLC data) with $\tau =3$ and $4$, respectively. In each setting of  $\tau_0=3$ or $4$, we simulate $n = 190$ data points using the corresponding parameters.
For each simulated data, we apply the proposed repro samples method. 
The experiments are repeated~200~times~in~each~setting of $\tau_0=3$~or~$4$.

The bar plots in Figure~\ref{fig:simulation_mixture} summarize the 200 level-95\% confidence sets of $\tau_0$ obtained~in the 200 repetitions using the repro samples method for $\tau_0 =3$ and $4$, respectively. The results  
demonstrate good empirical performances of the proposed confidence set in both settings, with the $\tau_0$ coverage rate 99\% and 96.5\%, respectively.\spacingset{1.77}
Over-coverage~is~observed~but~also expected, since the parameter space of $\tau_0$ is discrete and it is often 
impossible to achieve the exact level of 95\%. 
In the case of $\tau_0=3,$ the average size of the proposed confidence~set is about $3$ (standard error $0.05$), with  
the set $\{2, 3, 4\}$  being the majority among the $200$ simulations.  
As for $\tau_0=4,$ more than 80\% of the times, a confidence set is either $\{2, 3, 4\}$~or $\{2, 3, 4, 5\}$, with the average size of the confidence set around $3.65$ (standard error~$0.05$).

\subsection{Comparisons to existing approaches and additional inference on $\mu_j^{(0)}$ and~$\sigma_j^{(0)}$}
\label{sec:compare_numerical}

We have conducted a comprehensive study to compare the numerical performance of our method with that of several existing frequentist and Bayesian approaches using both the SLC real data and the same $2 \times 200$ simulated data sets. Due to space limit, we only~state our conclusion here and the detailed results are reported in  Appendix~\ref{sec:bayes_comp}. First, point estimators perform poorly in our simulation studies that mimic the SLC data (with engaged 
overlapping components).
For instance, 
the BIC point estimator, although consistent theoretically, rarely recovers the true $\tau_0$: It correctly identifies the true value of $\tau_0$ only 6.5\% of the time for the setting of $\tau_0 = 3$ and never for the setting of $\tau_0=4$. This underscores the critical importance of quantifying the estimation uncertainty through a confidence set~for~$\tau_0$. 
Second, the approaches developed under the classical test framework using (modified) penalized likelihood ratio tests \citep[e.g.,][]{roeder_graphical_1994, chen_inference_2012} perform only a one-sided test; thus it is often proposed to adopt ``the smallest $\tau_0$ that can not be rejected.'' We find 
in our studies  
the method by \cite{chen_inference_2012} often inaccurately favors smaller models, incorrectly maintaining $H_0: \tau_0=2$ 74\% of the time when actually $\tau_0 = 3$, and failing to reject $H_0: \tau_0=3$ more than 93\% of the time when actually $\tau_0 = 4$.  The results highlight the tendency of these one-sided tests to erroneously prefer a wrong $\tau$ over the correct one. Third, the performance of Bayesian methods \citep[e.g.,][]{richardson_bayesian_1997} is highly sensitive to the choices of priors on $\tau$, $\*\mu$, and $\*\sigma$,  and the credible sets do not have frequentist coverage rates.  
See Appendix~\ref{sec:bayes_comp} for more~details.

We also apply the three-step procedure in Section~\ref{sec:ThreeSteps} to conduct inference for the location and scale parameters $\mu_j^{(0)}$ and $\sigma_j^{(0)}$. Since Step~1 has already been completed in Section~\ref{sec:tau-numerical}, we proceed directly to Steps~2–3. In the real data analysis, the proposed method yields stable and interpretable confidence intervals for both parameters, providing a coherent joint quantification of uncertainty across discrete and continuous components. The simulation studies further show that our method is the only approach that consistently achieves~the desired empirical coverage for both means and variances, highlighting its robustness and practical reliability, particularly in finite samples. Detailed results are~reported~in~Appendix~\ref{sec:simu_location_scale}.

Based on the numerical results and comparisons in Sections~\ref{sec:tau-numerical}--\ref{sec:compare_numerical}, we can see that~the repro samples method is the only method that gives us a two-sided confidence set for~$\tau_0$~with the desired performance. The repro samples method can also effectively quantify uncertainty and makes joint inferences for both the discrete $\tau_0$ and the continuous  $\mu_j^{(0)}$ and $\sigma_j^{(0)}$, $j = 1, \ldots, \tau_0$. In summary, the repro samples method gives us an effective approach for analyzing Gaussian mixture data. It provides a performance-guaranteed comprehensive solution for a set of complex irregular inference problems, while the existing procedures~can~not.

\section{Concluding remark and further discussions}
\label{sec:6}

We have developed an effective and flexible new  framework for inference by utilizing, explicitly or implicitly, synthetic data that mimic the observed sample. A notable strength~is~that~it directly uses inversion techniques and does not rely on large-sample theory or likelihood functions, making it particularly effective for some previously unaddressed challenging  
problems, including those involving discrete or non-numerical parameters or non-numerical~data. By accommodating diverse parameters, models, and data types, it can greatly extend the~scope of statistical inference. We also compare the proposed method with the classical Neyman– Pearson framework. When the same test statistic is used as a nuclear mapping, the repro samples method provides  the same or better inference results than those obtained via Neyman inversion. The case study of Gaussian mixture  illustrates the method’s effectiveness in resolving the long-standing challenging inference problem of the unknown number of components, and  numerical studies further show that the proposed method is the only approach that can achieve the desired coverage rates for both discrete and continuous~parameters.

The generative model specification (\ref{eq:1}) may not be  unique --- the same probabilistic model~(i.e., the same likelihood) can possibly be represented through different generative models \citep[Remark~6]{Hannig2009}. For example, the uniform model $Y_i \sim U(\theta,\theta^2)$, $\theta>1$,  may be expressed in two different forms of (\ref{eq:1}), one using the full data and the other using the minimal sufficient statistics (the sample minimum and maximum) \citep[Example~4]{Hannig2016}. Within the repro samples framework, these different model formulations still lead to identical inference on $\theta$ if the same nuclear mapping is used; see Example~\ref{ex:fiducial} in Appendix~\ref{sec:App-G-unique-example}. This invariance 
contrasts with other Fisher inversion-based GFI and IM approaches, in which the resulting inference depends on the distributional form of the latent $U$'s and thus on the chosen generative model; cf., \cite{Hannig2016,martin2013inferential}.

Our primary developments are under the generative model~(\ref{eq:1}), but can be extended to more general settings, such as~(\ref{eq:invertedH}) that only requires  the target (true or oracle) parameter be recovered by a given algorithm. Nevertheless, all models are approximations, making it important to consider misspecification.
We have begun extending the repro samples method to settings with model misspecification, where a family of working models is used to analyze data generated from a completely unknown model. For example, in \cite{hou2025repro} on high-dimensional data binary
classification,  although the true model is completely unknown, inference is conducted using misspecified sparse working models, targeting oracle parameters that minimize an appropriate loss function and represent the best working-model  (closest to observed data). The development is a further extension of (\ref{eq:invertedH}) and aligns with several relevant research in model calibration, variational Bayes,   robust universal inference, robust~SBI,~among others \citep[e.g.,][]{smith2007calibration,  wang2019variational, park2023robust, tomaselli2025robustsimulationbasedinference}.

Finally, despite the two techniques in Section~\ref{sec:guide}, the computational cost can remain high for complex inference problems. For example, obtaining candidate sets for discrete parameters may require repeated runs of the inversion algorithm, with a cost comparable to bootstrap. This poses further challenges when an effective inversion algorithm is unavailable. Nevertheless, with continued growth in computing power, we view this as a characteristic shared by many modern simulation-based approaches rather than a fundamental limitation.

\section*{Acknowledgment} 
This is a substantially revised version of  https://arxiv.org/abs/2402.15004 and \\ https://arxiv.org/abs/2206.06421. 
We sincerely thank the editors, reviewers, and many colleagues in the community for their constructive comments and suggestions that have helped us greatly improve the development, presentation, and quality of the paper.



\appendix

\newpage

\spacingset{1.8}

\setcounter{example}{0}
\renewcommand{\theexample}{A\arabic{example}}

\setcounter{figure}{0}
\renewcommand{\thefigure}{A\arabic{figure}}

\setcounter{table}{0}
\renewcommand{\thetable}{A\arabic{table}}

\renewcommand{\theequation}{A\arabic{equation}}
\setcounter{equation}{0}

\setcounter{lemma}{0}
\renewcommand{\thelemma}{A\arabic{lemma}} 

\setcounter{theorem}{0}
\renewcommand{\thetheorem}{A\arabic{theorem}} 

\setcounter{corollary}{0}
\renewcommand{\thecorollary}{A\arabic{corollary}}

\setcounter{page}{1}
\renewcommand{\thepage}{A\arabic{page}} 


\spacingset{1.4} 

\centerline{\LARGE {\sc Appendices A - E}}

\centerline{\large (Supplementary Materials)}

\vspace{6mm}\noindent
The appendices are organized in five sections:

\begin{itemize}
[leftmargin=1em, labelsep=0.3em, 
align=parleft,   itemsep=1pt,
  parsep=0pt,
  topsep=1pt]
    \item[] {\bf Appendix A} Proofs, additional materials and examples in Section 2
    \item[] {\bf Appendix B} Proofs, additional materials and examples in Section 3
    \item[] {\bf Appendix C} Proofs, additional materials and examples in Section 4
     \item[] {\bf Appendix D} Additional implementation details and numerical results in Section 5
 \item[] {\bf Appendix E} A brief review and comparisons with existing and relevant inference procedures.
\end{itemize}

\section{Proofs, additional materials and examples in Section~\ref{sec:general}}
\subsection{Proofs of theoretical results in proof~\ref{sec:general}}
\label{sec:proof_sec2}

\begin{proof}[Proof of Theorem~\ref{thm:1}]

Since $\*y_{obs} = G({\btheta}_0, {\*u}^{rel})$ always holds, if 
$T({\bf u}^{rel}, \btheta_0)  \in B_{1 - \alpha}(\btheta_0)$, then $\btheta_0 \in \Gamma_{1-\alpha}(\*y_{obs}) \not = \emptyset$. That is, 
$\{\btheta_0 \in \Gamma_{1-\alpha}(\*y_{obs})\} \supseteq \{T({\bf u}^{rel}, \btheta_0)  \in B_{1 - \alpha}(\btheta_0)\}$. 
Similarly, for the random version~$\*Y = $ $ G({\btheta}_0, {\bf U})$, we have $\left\{\btheta_0 \in \Gamma_{1-\alpha}({\*Y})\right\} \supseteq \left\{T({\bf U}, \*\theta_0) \in B_{1 -\alpha}(\*\theta_0)\right\}$. Thus, 
$\P\{\btheta_0 \in \Gamma_{1-\alpha}(\*Y)\} \ge \P\{T(\*U, \btheta_0) \in B_{1 -\alpha}(\btheta_0)\} \ge 1- \alpha$, where the last equation is by \eqref{eq:B}. The same proof applies in the situation when \eqref{eq:B} holds approximately.
\end{proof}

\begin{proof}[Proof of Corollary 1]

 Under $H_0$, we have $\*\theta_0 \in \Theta_0$ and thus
    \begin{align*} \P\{p(\*Y) \leq \gamma\}
    \leq  \P \left[\inf\left\{1-\alpha': \*\theta_0 \in \Gamma_{1-\alpha'}(\*Y)\right\} \geq 1 - \gamma \right]  
    = \P \left\{ \*\theta_0 \not\in \Gamma_{1-\alpha' }(\*Y),  
    \forall \, \alpha' <  \gamma \right\} \leq \gamma,
    \end{align*}
    for any significance level  $\gamma.$ Thus, the Type I error is controlled. 
\end{proof}

\begin{proof}[Proof of Theorem \ref{thm:NP}]
\begin{eqnarray}
\Gamma_{1-\alpha}(\*y_{obs}) 
& = & \big\{\btheta: \exists \*u^* \in {\mathcal U} \mbox{ s.t. }  \*y_{obs} = G({\btheta}, \*u^*), T(\*u^*, \btheta)  \in B_{1 - \alpha}(\btheta) \big\}  \nonumber \\
& = &\big\{\btheta: \exists \*u^* \in {\mathcal U} \mbox{ s.t. }  \*y_{obs} = G({\btheta}, \*u^*), \widetilde T(G({\btheta}, \*u^*), \btheta)  \in B_{1 - \alpha}(\btheta) \big\} \nonumber
\\
& = &\big\{\btheta: \exists \*u^* \in {\mathcal U} \mbox{ s.t. }  \*y_{obs} = G({\btheta}, \*u^*), \widetilde T(\*y_{obs}, \btheta)  \in B_{1 - \alpha}(\btheta),  \big\} \nonumber 
\\
 & = &\big\{\btheta:  \*y_{obs} = G({\btheta}, \*u^*),  \exists \, \*u^* \in {\mathcal U} \big\} 
 \cap \big\{\btheta:  \widetilde T(\*y_{obs}, \btheta)  \in B_{1 - \alpha}(\btheta) \big\}
\nonumber \\
& = &\big\{\btheta:  \*y_{obs} = G({\btheta}, \*u^*),  \exists \, \*u^* \in {\mathcal U} \big\} 
 \cap \widetilde \Gamma_{1-\alpha}(\*y_{obs})
\nonumber
\\
& \subseteq & \widetilde \Gamma_{1-\alpha}(\*y_{obs}), \nonumber 
\end{eqnarray}
The two sets  
$\Gamma_{1-\alpha}(\*y_{obs}) = \widetilde \Gamma_{1-\alpha}(\*y_{obs})$, when $\widetilde \Gamma_{1-\alpha}(\*y_{obs}) \subseteq \big\{\btheta:  \*y_{obs} = G({\btheta}, \*u^*), 
\exists \, \*u^* \in {\mathcal U} \big\}$.
\end{proof}

\begin{proof}[Proof of Corollary \ref{cor:UMA-NP}]

 Let $\btheta_0$ be the true parameter with $\*y_{obs} = G(\btheta_0, \*u^{rel})$ (realized version) and $\*Z = G(\btheta_0, \*U)$ (random version). Let $\btheta' \not = \btheta_0$ be a false parameter value. 

(a) Since $\widetilde
\Gamma_{1-\alpha}(\*Y)$ is an UMA confidence set, by the definition of UMA \citep[][Section 9.2.1]{tCAS90a}, we have 
$$
\P\left\{\btheta' \in \widetilde
\Gamma_{1-\alpha}(\*Y) \right\} \leq \P\left\{ \btheta' \in C_{1-\alpha}(\*Y) \right\},
$$
where $C_{1-\alpha}(\*y_{obs})$ is any level $1 - \alpha$ confidence set. However, by Theorem~\ref{thm:NP}, we have $\Gamma_{1-\alpha}(\*Y) \subseteq \widetilde \Gamma_{1-\alpha}(\*Y)$. It follows that
$$
\P\left\{\btheta' \in 
\Gamma_{1-\alpha}(\*Y) \right\} \leq 
\P\big\{\btheta' \in \widetilde
\Gamma_{1-\alpha}(\*Y) \big\} \leq \P\left\{ \btheta' \in C_{1-\alpha}(\*Y) \right\}.
$$
So, the repro sample level $1 - \alpha$ confidence set $\Gamma_{1-\alpha}(\*y_{obs})$ is also UMA. 

(b) Since $\widetilde \Gamma_{1-\alpha}(\*Y)$ is unbiased, we have $\P\left\{\btheta' \in \widetilde \Gamma_{1-\alpha}(\*Y)\right\} \leq 1-\alpha $, for any $\btheta' \not = \btheta_0$. By Theorem~\ref{thm:NP}, we have $\Gamma_{1-\alpha}(\*Y) \subseteq \widetilde \Gamma_{1-\alpha}(\*Y)$, it follows that $\P\left\{\btheta' \in \Gamma_{1-\alpha}(\*Y)\right\} \leq \P\left\{\btheta' \in \widetilde \Gamma_{1-\alpha}(\*Y)\right\} \leq 1-\alpha$, for any $\btheta' \not = \btheta_0$. So $\widetilde \Gamma_{1-\alpha}(\*Y)$ is also unbiased. The UMA part of proof is the same as in~(a). 
\end{proof}

\subsection{A theorem on the confidence set in \eqref{eq:G2} and (\ref{eq:G2h})}
\label{sec:the_G2}
The following theorem states that the set $\Gamma_{1-\alpha}(\*y_{obs})$ defined in (\ref{eq:G2}) and (\ref{eq:G2h}) are level $1 -\alpha$ confidence sets for $\btheta_0$. 

\begin{theorem}\label{thm:2} Assume Eq. (\ref{eq:AA}) or Eq. (\ref{eq:invertedH}) holds with $\*\theta = \*\theta_0$. If the inequality (\ref{eq:B}) holds exactly for any fixed $\*\theta$, then for $\Gamma_{1-\alpha}(\*y_{obs})$  in (\ref{eq:G2}) or (\ref{eq:G2h}) the following inequality holds exactly, 
\begin{equation}
\label{eq:V2}
\P\left\{ {\btheta}_0 \in \Gamma_{1-\alpha}(\*Y)\right\} \ge 1- \alpha 
\quad
\hbox{for $0< \alpha <1$.}
\end{equation}
Furthermore, if (\ref{eq:B}) holds approximately 
with $\P \left\{T(\*U, \btheta)  \in B_{1 -\alpha}(\btheta) \right\}\ge (1-\alpha)\{1 + o(\delta^{'})\}$, for~a~small $\delta^{'} > 0$, then (\ref{eq:V2}) holds approximately 
with $\P\left\{ {\btheta}_0 \in \Gamma_{1-\alpha}(\*Y)\right\} \ge (1-\alpha)\{1 + o(\delta^{'})\}$, for $0< \alpha <1$. 
\end{theorem}

\begin{proof}[Proof of Theorem~\ref{thm:2}] The proof is similar to that of Theorem~\ref{thm:1}. Specifically, since $g(\*Y, \btheta, $ $ \*U)$ $ =0$ holds for $\*\theta = \*\theta_0$, 
we have 
$\big\{\btheta_0 \in \Gamma_{1-\alpha}({\*Y})\big\} \subseteq \big\{T({\bf U}, \*\theta_0) \in B_{1 -\alpha}(\*\theta_0)\big\}$ for the set $\Gamma_{1-\alpha}({\*Y})$ defined in (\ref{eq:G2}). Thus, 
The coverage of $\Gamma_{1-\alpha}(\*y_{obs})$ in \eqref{eq:G2} follows from  
$\P\{\btheta_0 \in \Gamma_{1-\alpha}(\*Y)\}  \geq \P\{T(\*U, \btheta_0) \in B_{1 -\alpha}(\btheta_0)\} \ge 1 - \alpha,$ where the last inequality is by \eqref{eq:B}.  The same proof applies in the situation when \eqref{eq:B} holds approximately. 

Finally, it follows immediately by making $g(\*Y, \btheta, \*U) = \btheta - H(\*Y, \*U)$ that \eqref{eq:V2} and its approximation also holds for $\Gamma_{1-\alpha}(\*y_{obs})$ in \eqref{eq:G2h}. 
\end{proof}

\subsection{Derivation of confidence set in \eqref{eq:BinExample} and its simulation results}
\label{sec:bin_example}
From~equation (\ref{eq:G1}), we have
\begin{eqnarray*}
\Gamma_{1-\alpha}(y_{obs}) &=& \bigg\{ \theta \big |  \exists \, \*u^* \in {\mathcal U} \, \hbox{s.t.} \,\, y_{obs} =  \sum\nolimits_{i =1}^r {\bf 1}{(u_i^* \leq \theta)},   
 \sum\nolimits_{i =1}^r {\bf 1}{(u_i^* \leq \theta)}
\in [a_L(\theta), a_U(\theta)] \bigg\}  \nonumber \\  
& = & \left\{ \theta \big |y_{obs} = \sum\nolimits_{i =1}^r {\bf 1}{(u_{i}^* \leq \theta)}, y_{obs}\in [a_L(
\theta), a_U(\theta)],  \exists \, \*u^* \in {\mathcal U} \right\}
\nonumber \\ 
& = & \left\{ \theta \big |y_{obs} = \sum\nolimits_{i =1}^r {\bf 1}{(u_{i}^* \leq \theta)},  \exists \, \*u^* \in {\mathcal U} \right\} \cap \left\{ 
\theta\big | y_{obs} \in [a_L(\theta), a_U(\theta)]\right\}
\nonumber \\ 
& = & \left\{ \theta \big |a_L(\theta) \leq y_{obs} \leq  a_U(\theta) \right\}. \nonumber 
\end{eqnarray*}
The last equation holds, because for any given $\theta \in \Theta = (0,1)$ there always exists at least a $\*u^* \in {\mathcal U}$ such that $y_{obs} = \sum_{i =1}^r {\bf 1}{(u_{i}^* \leq \theta)}$, i.e., $\{ \theta \big |y_{obs} = \sum_{i =1}^r {\bf 1}{(u_{i}^* \leq \theta)}, \exists \, \*u^* \in {\mathcal U} \} = \Theta$.
\hfill $\square$

Table~\ref{tab:simulation_binary_example}  provides a numerical study comparing the empirical performance of $\Gamma_{.95}(y_{obs})$ (Repro) against 95\% confidence intervals obtained using the traditional Wald method and the fiducial approach (GFI) of \cite{Hannig2009}, both of which are asymptotic methods that can ensure coverage only when $n$ is large. The repro sampling method is an exact method that improves the performance of existing methods, especially when $n \theta_0 < 5$.  We also apply two other exact methods to compute the confidence intervals, Clopper-Pearson and Stern. The Clopper-Pearson intervals are in general wider than those from the repro samples approach, and the Stern intervals have a similar performance to the repro samples. {Moreover, we also applied the quantile regression approach to compute the Borel set $B_{1-\alpha}(\theta)$, as discussed in Section~\ref{sec:guide} and Appendix~\ref{sec:quantile}. The results are close to the exact repro samples approach in Table~\ref{tab:simulation_binary_example}, implying that the quantile regression approach is effective in obtaining $B_{1-\alpha}(\theta)$ and the repro samples confidence sets. }

\begin{table}[!t]
   \centering 
\resizebox{\textwidth}{!}{\begin{tabular}{cccccccccc}
\hline\hline
 && \multicolumn{2}{c}{$n=20, \theta_0=0.1$} && \multicolumn{2}{c}{$n=20, \theta_0=0.4$} && \multicolumn{2}{c}{$n=20, \theta_0=0.8$}\\ \cline{3-4} \cline{6-7} \cline{9-10}
  && Coverage           &    Width      & & Coverage &Width   & &  Coverage    & Width         \\ \hline
 Repro & &     0.949(0.007)      &  0.281(0.059)   &      & 0.963(0.006)           &       0.408(0.026)   & &    0.959(0.006)       &    0.342(0.045)  \\ \hline
 Wald & &        0.877(0.010)  &     0.236(0.052)      & &  0.927(0.008)         &    0.418(0.028)       &  &  0.915(0.009)       &0.332(0.071)\\ \hline
 GFI & &   0.988(0.003)       &  0.293(0.065)        & &          0.963(0.006) &      0.438(0.022)     &  &    0.973(0.004)     &      0.365(0.056)     \\ \hline
Clopper-Pearson && 0.988(0.003) &  0.293(0.065) & & 0.963(0.006) &      0.438(0.022)  & &   0.973(0.004)     &      0.365(0.056)    \\ \hline
 Stern && 0.949(0.007)& 0.283(0.058) & &   0.963(0.006)   &       0.408(0.026) &&  0.959(0.006)       &    0.343(0.046)\\ \hline \hline
\end{tabular}}
 \captionsetup{font= small}
    \caption{Comparison of $95\%$ confidence intervals by the repro, Wald and GFI methods in Binomial$(n, \theta_0)$ data; repetitions $= 1000$; standard errors (sd's) are enclosed in brackets.}
    \label{tab:simulation_binary_example}
\end{table}

\subsection{Additional Examples in Section~\ref{sec:general}}
\label{sec:example_sec2}

In this appendix, we provide four additional illustrative examples that are relevant to the derivations in Sections~\ref{sec:general}, and Example~\ref{ex:crq_robust} is also relevant to Appendix~\ref{sec:mu-sigma}. Specifically, 
Example~\ref{ex:crq} provides an illustration of using (\ref{eq:G2}) 
under the extended framework
to construct confidence intervals for nonparametric quantiles. 
The repro samples method works well and the intervals by the repro samples method 
are compared much favorably to those obtained using large-sample 
bootstrap methods.  
Example~\ref{ex:crq_robust} extends  Example~\ref{ex:crq} to develop a robust method for analysis of 
contaminated Gaussian data, such as those in a Gaussian mixture model where the membership assignments are often difficult to recover accurately. Example~\ref{ex:quantile_privacy} studies the same problem as Example~\ref{ex:crq}, but with all the observations privatized. It presents a case where the nuclear mapping is not test statistics. Finally, in Example A4, which is a continuation of Example~\ref{ex:3}, we show that the LRT confidence interval is not optimal, and the proposed repro sample procedure is more efficient. 

\begin{example}[Nonparametric inference for population quantiles] \label{ex:crq} 
Assume $\*y_{obs} = (y_1, \ldots,$ $ y_n)^\top$ are from an unknown distribution $F$. Our goal is to make inference about the population quantile $\theta_0 = F^{-1}(\zeta)$, for a given $0<\zeta <1$. 
In the form of (8), 
we have 
$$
g(\*Y, \theta, \*U) = \sum_{i = 1}^n I(Y_i \leq \theta) - \sum_{i =1}^n U_i = 0,$$ 
where $Y_i \sim F$ and $U_i \sim \text{Bernoulli}(\zeta)$. 
Let us define $T(\*U, \theta) =  T(\*U) = \sum_{i =1}^n U_i$. It follows that $T(\*U) \sim \text{Binomial}(n, \zeta)$ and thus $B_{1 -\alpha} = [a_L(\zeta), a_U(\zeta)]$, where $a_L(\cdot)$ and $a_U(\cdot)$ are defined 
(6) 
with $r = n$. By (9), 
a finite-sample level $1 - \alpha$ confidence set is
\begin{align*}
\Gamma_{1-\alpha}(\*y_{obs}) & = \big\{\theta: \sum_{i = 1}^n I(y_i \leq \theta) - \sum_{i =1}^n u_i = 0, 
\sum_{i =1}^n u_i \in [a_L(\zeta), a_U(\zeta)]
\big\} 
\\ & = \big\{\theta: a_L(\zeta) \leq \sum_{i = 1}^n I(y_i \leq \theta) \leq a_U(\zeta))
\big\} \\ & = \big[y_{(a_L(\zeta))}, y_{(a_U(\zeta)+1)}\big),
\end{align*}
where 
$y_{(k)}$ is the $k$th order statistic
of the sample $\{y_1, \ldots, y_n\}$,  
$y_{(0)}$ and $y_{(n+1)}$ are the infimum and supremum of support of $F.$ 

Figure~\ref{fig:quantile} provides a 
summary of
a numerical study that  compares the empirical performance of $\Gamma_{.95}(y_{obs})$ 
against 95\% intervals obtained by the conventional bootstrap method under two types of distributions $F$: (a) $F$ is $\text{Cauchy}(0,1)$ and (b) $F$ is negative binomial  $\text{NB}(2, 0.1)$.
The repro samples method works well in terms of coverage rate and interval length even for $\zeta$'s close to $0$ or $1$, whereas the bootstrap method has coverage issues under these settings.

See also Example 2B of \cite{XieWang2022} for related work on much more complex finite-sample semi-parametric inference development on a censored regression model \citep{powell_censored_1986}.

\begin{figure}[ht]
    \centering
    \begin{subfigure}[b]{\textwidth}
    \centering
      \includegraphics[width=\textwidth, height= 0.23\textheight]{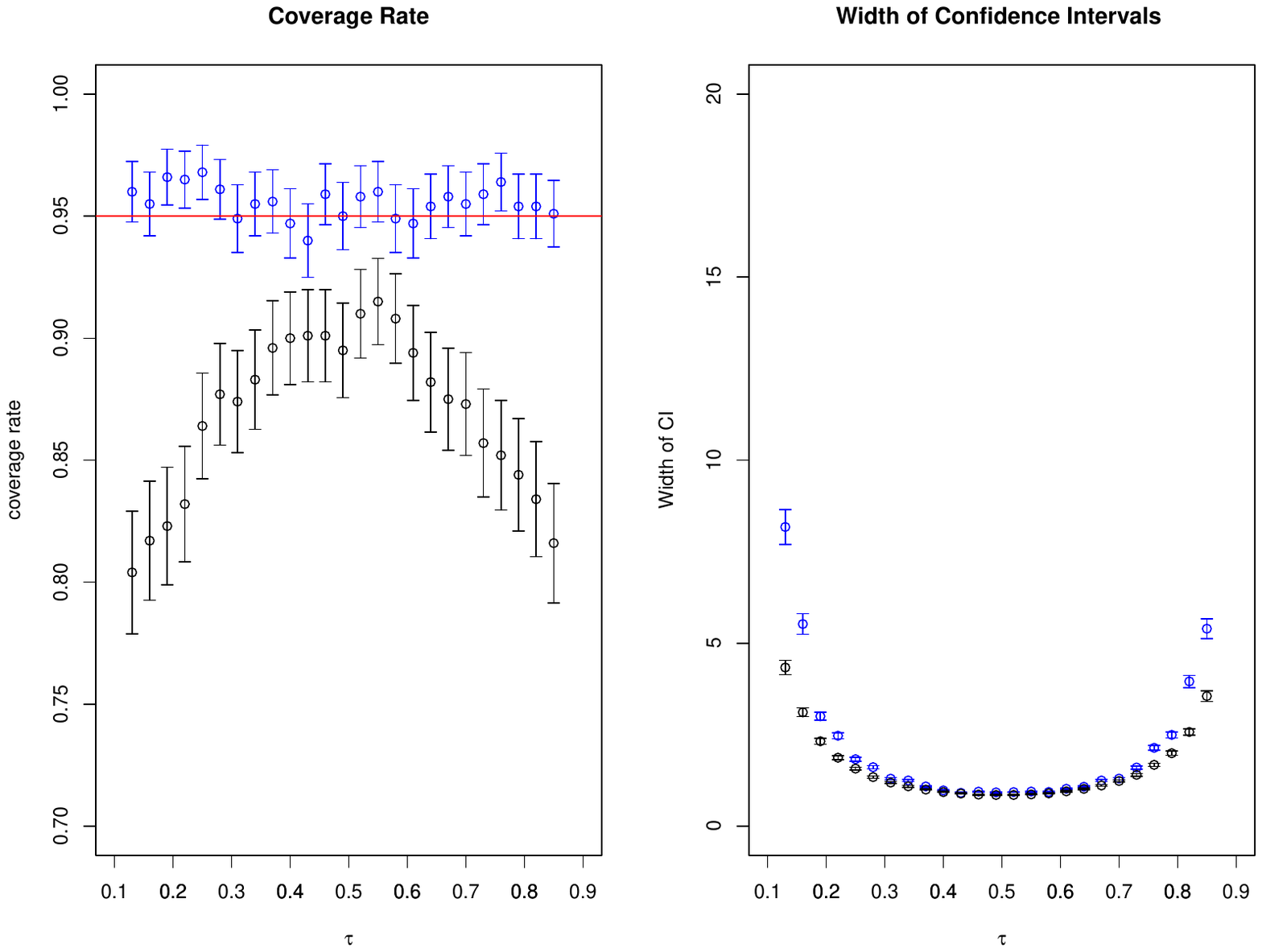}
      \caption{Cauchy Distribution}
    \end{subfigure}
   \begin{subfigure}[b]{\textwidth}
   \centering
\includegraphics[width=\textwidth,  height= 0.23\textheight]{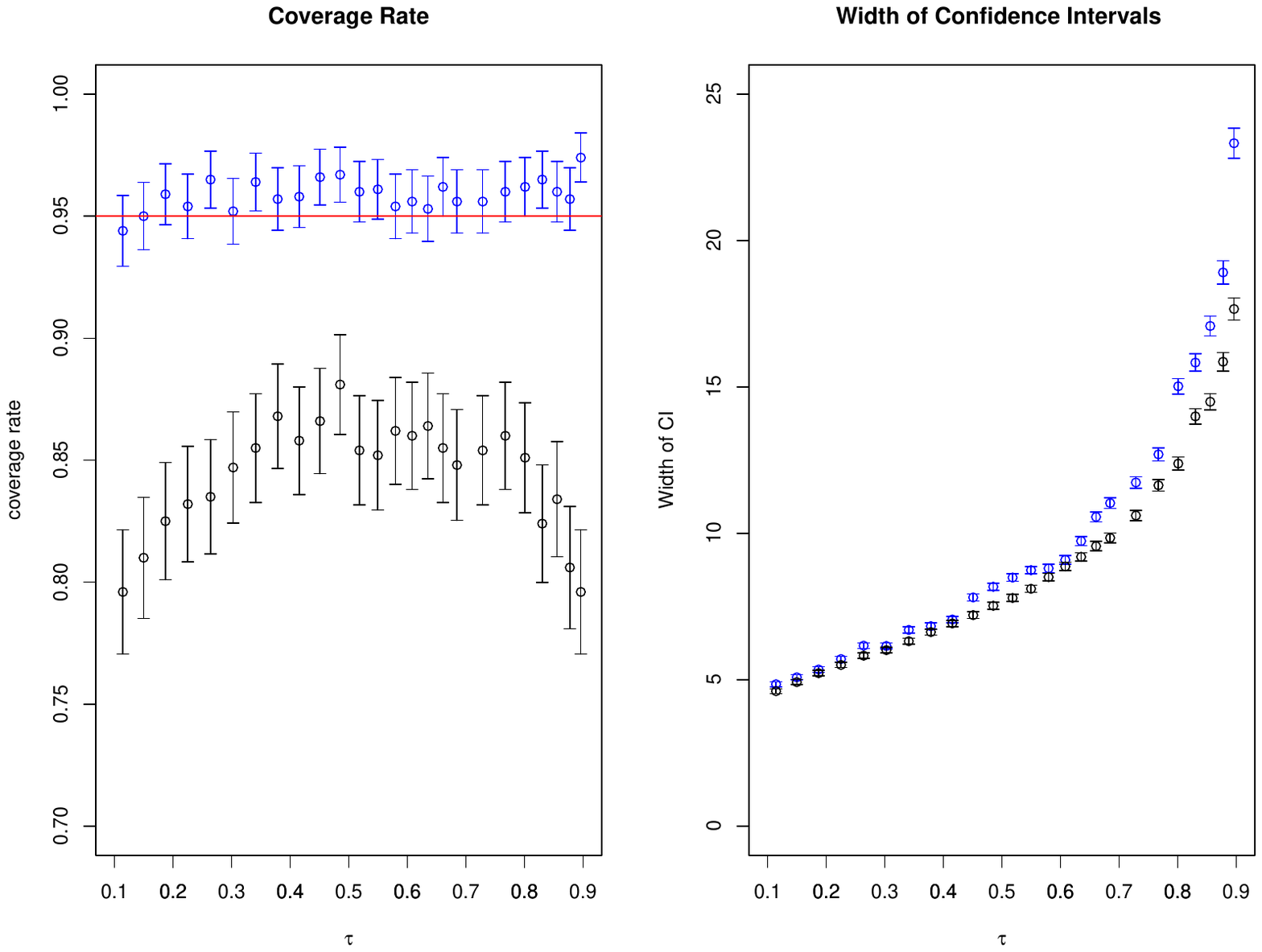}
    \caption{Negative Binomial Distribution}
   \end{subfigure}
    \caption{\small Coverage rates and widths of 95\% confidence intervals for $\theta_0 = F^{-1}(\zeta)$ of an (unknown) distribution $F$ for a range of percentile $\zeta$ values. Bars and dots colored in blue are produced by the repro samples method and those in black are 
    by the conventional bootstrap method. In (a) $F$ is $\text{Cauchy}(0,1)$ and in (b) $F$ is $\text{NB}(2, 0.1)$. Both cases have a sample size of $n=60.$ 
    The simulations are replicated for $1000$ times. The error bars display the mean~$\pm 2$~standard~errors.}
  \label{fig:quantile}
\end{figure}

\end{example}

\begin{example}[Robust method for contaminated Gaussian data]
\label{ex:crq_robust}
Suppose the observations $\*y_{obs} =(y_1, \ldots,$ $ y_n)^\top$ in Example~\ref{ex:crq} are from
$N(\mu_0, \sigma_0^2)$, a small percentage of which, however, are contaminated. If $I(Y_i \leq \mu_0) 
\sim {\rm Bernoulli}(.5)$ holds for the contaminated $Y_i$'s, the approach in Example~\ref{ex:crq} can be directly applied to obtain a confidence set 
for $\mu_0$ with $\zeta = .5$.  That is, a level $1 - \alpha$ confidence interval for $\mu_0$ is
\begin{align}
\big\{\mu: a_L(.5) \leq \sum_{i = 1}^n I(y_i \leq \mu) \leq a_U(.5))
\big\} = \big(y_{(a_L(.5))}, y_{(a_U(.5)+1)}\big)
\label{eq:quant_ci}
\end{align}
where $a_L(\cdot)$ and $a_U(\cdot)$ are defined  in \eqref{eq:B-bounds} with $r = n$; i.e.,
$$
\left(a_L(.5), a_U(.5)\right) = 
{\arg\min}_{\left\{(i,j): \sum^j_{k=i} {n \choose k} \left(\frac12\right)^n \geq 1- \alpha \right\}}
    |j-i|. 
$$
It can be shown that this interval in \eqref{eq:quant_ci} corresponds to conducting the sign test on median $H_0: \mu = \mu_0$ versus $H_a: \mu \not= \mu_0$. Since the sample median has a high breakdown point $BP = 0.5$, the method has a high breakdown point as well.

In cases that 
the assumption of $I(Y_i \leq \mu_0) 
\sim {\rm Bernoulli}(.5)$ does not hold, 
we typically have 
$I(Y_i \leq \mu_0 + \delta) 
\sim {\rm Bernoulli}(.5)$ instead, where $\delta = \mu_c - \mu_0$ with $\mu_c$ being the (population) median of the contaminated data $Y_i$.  
In this case, 
we may 
adjust the confidence interval in \eqref{eq:quant_ci} by  $\big\{\mu: a_L(.5) \leq \sum_{i = 1}^n I(y_i \leq \mu + \delta) \leq a_U(.5))
\big\} = \big(y_{(a_L(.5))} - \delta, y_{(a_U(.5)+1)} -\delta\big)$. 
Practically, $\delta$ is often unknown, however. We need some empirical knowledge of $\delta$ or a way to estimate 
$\delta$. 
In our numerical studies in Section~5  
for the Gaussian mixture model, for each given Gaussian component we estimate the unknown $\delta$ by $\hat \delta = M - \widetilde M,$ where $M$ is the sample median, and $\widetilde M$ is the weighted sample median.  
Although a more sophisticated weighting scheme could be developed, for simplicity we simply assigned a weight of $1.5$ for all observations overlapping with other components. Here, we say an observation overlaps with another component if it falls in the span of any other components given the (estimated) membership matrix, where the span of a  component refers to the interval from the minimum to the maximum of 
members in
the component assigned by the membership matrix.
Our empirical studies suggest that this simple weighting scheme works well in the settings of Gaussian mixture models.

To obtain a robust confidence interval for $\sigma_0$, we utilized a robust and consistent estimator of $\sigma_0$: $\hat \sigma = {\it MAD}/\Psi_n^{-1}(.5)$, 
where ${\it MAD} = {\it med}\{|y_i - M|: i = 1, \ldots, n\}$ is the median absolute deviation, $M = {\it med}\{y_i: i = 1, \ldots, n\}$ and $\Psi_n(t) = P(Z_i - M_z < t)$ is the cumulative distribution function of $Z_i - M_z$, $i = 1, \ldots, n$. Here, $Z_i \overset{iid}{\sim} N(0,1)$ and $M_z = {\it med}\{Z_i: i = 1, \ldots, n\}.$ 
We have a generalized data generating model~in~the~form~of~(\ref{eq:quantile}), 
$$\sum_{i = 1}^n I(|Y_i -M| \leq \Psi_n^{-1}(.5)\sigma ) - \sum_{i =1}^n U_i' = 0,$$ 
where $U'_i = I\{|Z_i - M_z| \leq {\small \Psi_n^{-1}(.5)}\}$ are dependent Bernoulli$(.5)$ variables. Since the distribution of $\sum\nolimits_{i = 1}^n U'_i$ can be simulated, we can get 
a Borel interval  $B_{1 -\alpha} = [a'_L, a'_U]$, similarly as in (\ref{eq:B-bounds}); i.e., 
$\left[a_L', a_U'\right] = 
     \arg \min_{\left\{(l,r): 
  P( l < \sum_{i = 1}^n U'_i < r )
     \geq 1-\alpha\right\}} |r - l|$. A level $1 - \alpha$ confidence set of $\sigma_0$ is then $$
     \big\{\sigma: a'_L \leq \sum_{i =1}^n I(|y_i - M| \leq {\small \Psi_n^{-1}(.5)} \sigma ) \leq a'_U \big\}.
     $$
     Since $MAD$ has a high breakdown point of 50\%, the above method is expected to have a high breakdown point as well.  

\end{example}

\begin{example} \label{ex:quantile_privacy}
(Nonparameteric quantile example where the observed sample is privacy-tized) In the setup of Example~\ref{ex:crq}, suppose the realized $y_i$'s are not given to us. Instead, due to privacy consideration, we are given 
$
\widetilde {\*y}_{obs} = (\widetilde y_1, \ldots, \widetilde y_n)^\top
$
where 
$$
\widetilde y_i^{obs} = y_i + z_i^{rel}
$$
and $z_i^{rel}$ is the realized (privacy) noise from $N(0, \delta)$. Here, $\delta$ is given and the $N(0, \delta)$-distributed $Z$ noise is independent of $Y \sim F$, but $z_i^{rel}$ is not given to us. We still want to make inference for the same quantile parameter of $F$, $\theta_0 = F^{-1}(\zeta)$, for a given $0<\zeta <1$. In this case, we have 
$$
g(\widetilde{\*Y}, \theta, \*U) = \sum_{i = 1}^n I(\widetilde Y_i - Z_i \leq \theta) - \sum_{i =1}^n B_i = 0,
$$ 
where $\widetilde Y_i = Y_i + Z_i$, $Y_i \sim F$, $Z_i \sim N(0,\delta)$ and $B_i \sim \text{Bernoulli}(\zeta)$. Also, $U_i = (Z_i, B_i)$. We would like to make inference for $\theta_0$ for the given privatized sample $\widetilde {\*y}_{obs}$. When $F$ is unknown, it appears there no good way to obtain a good test statistic for this inference problem.

Here we define the nuclear mapping as 
\begin{align*}
     T(\*u, \theta)  = \left(\sum_{i=1}^n B_i, \sum z_i^2/\delta  \right),
\end{align*}
and the Borel set $B_{1 -\alpha} = [a'_L(\eta), a'_U(\eta)) \times (0, F^{-1}_{\chi^2_n}(1+\alpha'-\alpha)),$ where $\alpha' < \alpha <1,$
\begin{equation}
\left(a'_L(\theta), a'_U(\theta)\right) = 
\arg\min\nolimits_{\big\{(i,j): \sum^j_{k=i} {r \atopwithdelims( ) k} \theta^k (1 -  \theta)^{(r-k)} \geq 1- \alpha' \big\}}
    |j-i|, \nonumber
\end{equation}
following \eqref{eq:B-bounds}, but with $\sum^j_{k=i} {r \atopwithdelims( ) k} \theta^k (1 -  \theta)^{(r-k)} \geq 1-\alpha'$ instead of $1 - \alpha.$ In our implementation, we set $\alpha' = \alpha - 0.01$ and $1+\alpha' - \alpha = 0.99.$

See also \cite{Awan03072025} for an application of the repro samples method under a more general privacy model setup.

\begin{table}[]
    \centering
    \caption{Simulation Results for 95\% confidence intervals for $\eta$-quantile from privatized data, generated by $Y \sim Gamma(2,4) + N(0, \delta)$}
    \begin{tabular}{c|c c|cc}
    & \multicolumn{2}{c|}{$\eta = 0.7, n = 50$} & \multicolumn{2}{c}{$\eta = 0.75, n= 150$}\\
         $\delta$ & Coverage & Width & Coverage & Width \\ \hline
         $\delta = 4$ & 0.985 & 11.389  & 1.00 & 9.383  \\
         $\delta = 2.25$  & 0.985 & 10.189 & 0.995 & 8.262   \\ 
         $\delta = 1$ & 0.980 &  8.902 & 0.990 & 7.064 \\
         $\delta = 0.25$ & 0.975 & 7.414 &  0.990 & 5.620\\
         $\delta  =0 $ & 0.970 & 5.214  & 0.960 & 3.284\\
         \hline
    \end{tabular}
    \label{tab:my_label}
\end{table}

\begin{example}\label{ex:3a}
[Likelihood inference; Example~\ref{ex:3} continues] 
The likelihood function under the setup of Example~\ref{ex:3} is
$L(\theta|\*y^{obs}) = \prod\nolimits_{i =1}^n
    \big[ I\{-1 < y_i - \theta \leq 1\}/2\big] 
    \propto I\big\{ \max\nolimits_{1\leq i \leq n}|y_i - \theta| < 1 \big\}
$. The LRT statistic for testing $H_0: \btheta_0 = \btheta$ vs $H_1: \btheta_0 \not = \btheta$ is 
$\lambda_{\theta}(\*y_{obs}) 
= I\big\{ \max\nolimits_{1 \leq i \leq n}|y_i - \theta| < 1 \big\}$, but the regularity conditions do not hold and $-2 \log\{\lambda_{\theta}(\*Y)\}$ 
does not converge to a $\chi^2$ distribution. 
Nevertheless, since $\lambda_{\theta}(\*y_{obs})$ is 
decreasing in $\widetilde T(\*y_{obs}, \theta) = \max\nolimits_{1 \leq i \leq n}|y_i - \theta|$, the LRT rejects $H_0$ with a large $\widetilde T(\*y_{obs}, \theta)$, which leads to 
a level $1 - \alpha$ confidence set $$\widetilde \Gamma_{1-\alpha}(\*y_{obs}) = \{\theta: \max\nolimits_{1 \leq i \leq n} |y_i - \theta| < (1-\alpha)^{1/n}\} 
=  (y_{(n)} - (1-\alpha)^{1/n}, y_{(1)} + (1-\alpha)^{1/n}).$$ 
We get the same confidence set if we 
directly define our nuclear mapping function through $\widetilde T(\*y_{obs}, \theta) =  \max\nolimits_{1 \leq i \leq n} |y_i - \theta|$. Alternatively, 
we consider another choice 
of 
a vector nuclear mapping function $T(\*u, \theta) = (u_{(1)}, u_{(n)})$,
where $u_{(k)}$ is the $k$th order statistic of a sample from $U(-1,1)$. 
Let $c_{1-\alpha} \in (0,1)$ be a solution of
$({c_{1-\alpha}+1})^n - 2^{n-1}c_{1-\alpha}^n = 2^{n-1}\alpha$. We can show that $\P\{(U_{(1)}, U_{(n)}) \in B_{1 -\alpha}\} = 1- \alpha$, for $B_{1 -\alpha} = (-1, -c_{1-\alpha}) \times (c_{1-\alpha}, 1)$. Thus, a level $1 - \alpha$ confidence set is $\Gamma_{1-\alpha}(\*y_{obs})$ $ = \{\theta: y_{(n)} - 1 < \theta < y_{(1)} + 1; 
 y_{(1)} - \theta < -c_{1-\alpha}, y_{(n)} - \theta > c_{1-\alpha}\},
 $
 which is simplified~to $$\Gamma_{1-\alpha}(\*y_{obs})
= \big(\max\big\{y_{(1)} + c_{1-\alpha}, y_{(n)} - 1\big\}, \min\big\{y_{(n)} - c_{1-\alpha}, y_{(1)} + 1\big\}\big).$$ 
Note that, if we consider the space of $(U_{(1)} , U_{(n)})$, the
Borel set for the LRT confidence set $\widetilde \Gamma_{1-\alpha}(\*y_{obs})$ is $\widetilde B_{1 -\alpha} = (- (1-\alpha)^{1/n}, (1-\alpha)^{1/n}) \times (- (1-\alpha)^{1/n}, (1-\alpha)^{1/n})$, which is much bigger than the Borel set $B_{1 -\alpha} = (-1, -c_{1-\alpha}) \times (c_{1-\alpha}, 1)$ of the alternative method. 

We have conducted a numerical study with true $\theta_0 = 0$ and $n = 5, 20, 200$, respectively. In all cases, the coverage rates of both confidence intervals are right on target around $95\%$ in 3000 repetitions.  
The intervals by the repro sample method are on average consistently shorter than those obtained using the LRT method across all sample sizes, as evidenced empirically in Figure~\ref{fig:lrt} as well. 
Although a LRT is uniformly most powerful for a simple-versus-simple test by the Neyman-Pearson Lemma, it is not the case for the two-sided test in this example.  Here, we can explore and use the repro samples method to obtain a better confidence interval.   
\end{example}

\begin{figure}
    \centering
      \includegraphics[width=\textwidth, height= 0.18\textheight]{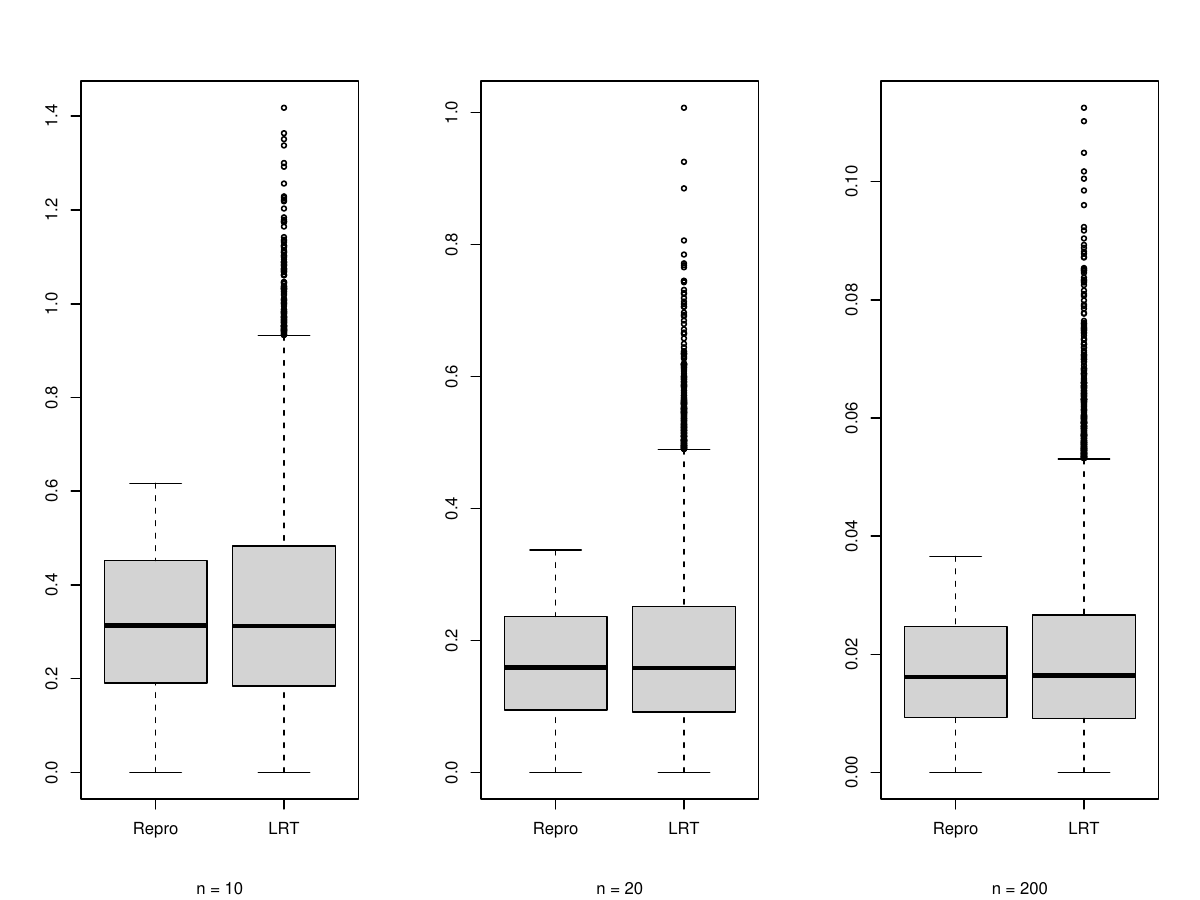}
    \caption{{\small Side-by-side box plots comparing the widths of the two 95\% confidence intervals discussed in Example 4, for $n = 10, 20$ and $200$, respectively; 
    The number of repetitions $=5000$.     
    The~empirical coverage rates are 
    $(.952,  .954,  .949)$ for the repro method and $(.950,  .952, .950)$ for the LRT method, respectively, all on the target. The average interval lengths are $(.320, .166, .017)$, respectively, for the repro method, all smaller than the corresponding lengths $(.351, .184,  .020)$ for the LRT method.
    }}

    \label{fig:lrt}
\end{figure}

\end{example}

\section{Proofs, additional materials and examples in Section~\ref{sec:3}}

\subsection{Proofs of theoretical results in Section~\ref{sec:3} }
\label{sec:proof_sec3}

\begin{proof}[Proof of Lemma~\ref{lemma:nuisance}] 

It follows from the definitions in \eqref{eq:nu} and \eqref{eq:T_p} that
$\P\left\{T_p(\*U, \btheta) \leq 1-\alpha \right\} \ge \P\left\{\nu (\*U, \eta, \*\beta,\*\beta) \ \leq 1-\alpha \right\} \geq 
\P\{\btheta \in \Gamma_{1-\alpha}(G(\*U, \btheta))\} \ge 1-\alpha,$ 
where $\Gamma_{1-\alpha}(G(\*u, \btheta))$ is that used in \eqref{eq:nu}. 
\end{proof}

\begin{proof}[Proof of Theorem~\ref{the:nuisance}] 
Similar to the proof of Theorem~\ref{thm:1}, we can show that $\{\eta_0 \in \Xi_{1-\alpha}(\*Y)\} 
\supseteq
\{T_p(\*U, \*\theta_0) $ $\leq 1-\alpha\}$.
The coverage probability is therefore bounded by
$\P\left\{\eta_0 \in \Xi_{1-\alpha}(\*Y)\right\} \geq \P\{T_p(\*U, \*\theta_0)\leq 1-\alpha\}.$ Thus,  Theorem~\ref{the:nuisance} holds by Lemma~\ref{lemma:nuisance}.
\end{proof}

\begin{proof}[Proof of Lemma~\ref{lemma:p_bound_C_D_general}]
Let $\P_{\*U, \mathcal V}$ denote the joint probability function with respect to $\*U$ and $\mathcal V,$ where $\mathcal V$  is the set of a collection of independent samples of $\*U^s \sim \*U,$ 
then
\begin{align*} 
& \P_{\*U, \mathcal V}\big\{\eta_0 \not\in \widehat{\it \Upsilon}_{{\mathcal V}}(\*Y)\big\}  = \P_{\*U, \mathcal V}\left\{\eta_0 \not\in \widehat{\it \Upsilon}_{{\mathcal V}}(\*Y), \bigcap_{\*U^s \in \mathcal V}\{\*U^s \not\in S_{nb}(\*U)\}\right\} \\ & \qquad \qquad  
 +  \P_{\*U, \mathcal V}\left\{\eta_0 \not\in \widehat{\it \Upsilon}_{{\mathcal V}}(\*Y), \bigcup_{\*U^s \in \mathcal V}\{\*U^s \in S_{nb}(\*U)\}\right\}\\ & \qquad   = \P_{\*U, \mathcal V}\left\{\eta_0 \not\in \widehat{\it \Upsilon}_{{\mathcal V}}(\*Y), \bigcap_{\*U^s \in \mathcal V}\{\*U^s \not\in S_{nb}(\*U)\}\right\} \\
& \qquad \leq \P_{\*U, \mathcal V}\left\{ \bigcap_{\*U^s \in \mathcal V}\{\*U^s \not\in S_{nb}(\*U)\}\right\}\\
& \qquad = E_{\*U}\left[\P_{\mathcal V|U}\left\{\bigcap_{\*U^s \in \mathcal V}\{\*U^s \not\in S_{nb}(\*U)\}\middle |\*U\right\}\right] =  E_{\*U}\left[\left\{1- \P_{\*U^s|\*U}\left(\*U^s \in S_{nb}(\*U)\middle |\*U\right)\right\}^{|\mathcal V|}\right]\\
&\qquad  \leq (1-c_{nb})^{|\mathcal V|}.
\numberthis \label{eq: 1-P_delta}
\end{align*}
The second equality holds since the set inside the second probability term is empty.  This is because 
$\{\*U^s \in S_{nb}(\*U)\}$ implies $ \{\eta_0 = \arg \min\nolimits_{\eta} \min\nolimits_{\*\beta}L\big(\*Y, $ $G((\eta, \*\beta^\top)^\top, \*U^s)\big)
\}$
by the definition of $S_{nb}(\*U)$, and thus  $\bigcup_{\*U^s \in \mathcal V}\{\*U^s \in S_{nb}(\*U)\}$ holds implies $\{\eta_0 \in  \widehat{\it \Upsilon}_{{\mathcal V}}(\*Y)
\}$ by the definition of $\widehat{\it \Upsilon}_{{\mathcal V}}(\*Y)$, which contradicts the other statement in the set that $\{\eta_0 \not\in  \widehat{\it \Upsilon}_{{\mathcal V}}(\*Y)
\}$. 
\end{proof}

\begin{proof}[Proof of Theorem~\ref{thm:p_bound_C_D_general}]
	By Lemma~\ref{lemma:p_bound_C_D_general} and Markov Inequality, 
	\begin{align*}
	\P_{\mathcal V}\left[\P_{\*U|\mathcal V}\{\eta_0 \not\in \widehat{\it \Upsilon}_{{\mathcal V}}(\*Y)\} \geq (1-P_{nb})^{|\mathcal V|/2}\right] & \leq \frac{E_{\mathcal V}\P_{\*U|\mathcal V}\{\eta_0 \not\in \widehat{\it \Upsilon}_{{\mathcal V}}(\*Y)\}}{(1-P_{nb})^{|\mathcal V|/2}}\\
	& \leq \frac{\P_{\*U,\mathcal V}\{\eta_0 \not\in \widehat{\it \Upsilon}_{{\mathcal V}}(\*Y)\}}{(1-P_{nb})^{|\mathcal V|/2}} = (1-P_{nb})^{|\mathcal V|/2}.	
	\end{align*}
Theorem~\ref{thm:p_bound_C_D_general}(a) then follows by making $c<-\frac{1}{2}\log(1- P_{nb}).$ 

Moreover, we have
\begin{align*}
   \P_{\*U|\mathcal V}\big\{\btheta_0 \in \Gamma_{1-\alpha}(\*Y)\big\} &  = \P_{\*U|\mathcal V}\big\{\btheta_0 \in \Gamma_{1-\alpha}(\*Y), \eta_0 \in\widehat{\it \Upsilon}_{{\mathcal V}}(\*Y)\big\} +  \P_{\*U|\mathcal V}\big\{\btheta_0 \in \Gamma_{1-\alpha}(\*Y),\eta_0 \not\in\widehat{\it \Upsilon}_{{\mathcal V}}(\*Y)\big\}  \\
   &  = \P_{\*U|\mathcal V}\big\{\btheta_0 \in \Gamma_{1-\alpha}'(\*Y)\big\} +  \P_{\*U|\mathcal V}\big\{\btheta_0 \in \Gamma_{1-\alpha}(\*Y), \eta_0 \not\in\widehat{\it \Upsilon}_{{\mathcal V}}(\*Y)\big\} \geq 1 - \alpha.
\end{align*}
It then follows that
   \begin{align}
   \label{eq:coverage_inequality}
       \P_{\*U|\mathcal V}\big\{\btheta_0 \in \Gamma_{1-\alpha}'(\*Y) \big\} \geq 1- \alpha - \P_{\*U|\mathcal V}\big\{\btheta_0 \in \Gamma_{1-\alpha}(\*Y), \btheta_0 \not\in \widehat{\it \Upsilon}_{{\mathcal V}}(\*Y)\big\} \geq 1- \alpha - \P\big\{\theta_0 \not\in \widehat{\it \Upsilon}_{{\mathcal V}}(\*Y)\big\}.
   \end{align}
Theorem ~\ref{thm:p_bound_C_D_general}(b) then follows immediately from Theorem~\ref{thm:p_bound_C_D_general}(a).
\end{proof}

\begin{proof}[Proof of Corollary~\ref{cor:3steps}]
    Following the definition in \eqref{eq:CIeta}, by Theorem~\ref{thm:p_bound_C_D_general},
    \begin{align*}
        \P\{\eta_0 \not\in \Xi_{1-\alpha, 1}(\*Y)\} \leq \P(T_p(\*U, \*\theta_0) > 1- \alpha) + \P\{\eta_0 \not\in \hat\Upsilon_{\mathcal V}(\*Y)\}\leq 1-\alpha + o_p(e^{-c_1|\mathcal V|}),
    \end{align*}
    for some $c_1>0.$ Corollary~\ref{cor:3steps}(a) then follows immediately. 
    Similarly, by definition \eqref{eq:conf_beta_combine-prime}, it follows from the above that
    \begin{align*}
         \P\{\*\beta_0 \not\in \Xi_{1-\alpha, 2}(\*Y)\} & \leq \P(T_p(\*U, \*\theta_0) > 1- \alpha') + \P\{\eta_0 \not\in \Xi_{1-\alpha'', 2}(\*Y)\} \leq \alpha' +  \alpha'' + o_p(e^{-c_2|\mathcal V|}) \\
         & =  \alpha + o_p(e^{-c_2|\mathcal V|}),
    \end{align*}
    for some $c_2>0,$ which proves Corollary~\ref{cor:3steps}(b). Finally by \eqref{eq:conf_beta_combine-prime2}
    \begin{align*}
         \P\{\*\theta_0 \not\in \Xi_{\alpha, 3}(\*Y)\} & \leq \P(T_p(\*U, \*\theta_0) > 1-\alpha') + \P\{\eta_0 \not\in \Xi_{1-\alpha'', 2}(\*Y)\} \leq \alpha' + \alpha'' + o_p(e^{-c_3|\mathcal V|}) \\
         & = \alpha + o_p(e^{-c_3|\mathcal V|}),
    \end{align*}
    for some $c_3>0,$ from which Corollary~\ref{cor:3steps}(c) follows. 
\end{proof}

\subsection{A Monte-Carlo method to compute the profile nuclear mapping in \eqref{eq:T_p} }\label{sec:depth_profile}
The profile nuclear mapping 
$T_p(\*u, \btheta)$ in (\ref{eq:T_p}) is defined through $\nu(\cdot; \eta, \*\beta, \widetilde {\*\beta})$ in (\ref{eq:nu}), which often can be computed either directly or by a Monte-Carlo method.
Specifically, if $\nu(\cdot; \eta, \*\beta, \widetilde {\*\beta})$ has no explicit expression, 
Lemma~\ref{lem:nuisance_depth} below suggests that we can use a Monte-Carlo method to get $\nu(\cdot; \eta, \*\beta, \widetilde {\*\beta})$ for a given $(\eta, \*\beta, \widetilde {\*\beta})$. A proof of  Lemma~\ref{lem:nuisance_depth} is in Appendix~\ref{sec:proof_sec3}.

\begin{lemma}
\label{lem:nuisance_depth}
Let $D_{{\mathcal S}_{\theta}}(\*t)$ be the empirical data depth function \citep{liu_multivariate_1999-1}  computed using the Monte-Carlo points in ${\mathcal S}_{\theta}$, where 
${{\mathcal S}_\theta} = \{T_p(\*u^s, \btheta), \*u^s \in {\mathcal V}\}$ is the Monte-Carlo set 
for a given $\btheta = (\eta, {\*\beta}^\top)^\top$ and ${\mathcal V}$ is a set of simulated $\*u^s \sim \*U$. Denote the empirical cumulative distribution function of $D_{{\mathcal S}_{\theta}}\big(T_p(\*u^s, \btheta)\big)$ as $F_{\mathcal V|D}(s) = \frac{1}{|{\mathcal V}|} \sum_{\*u^s \in {\mathcal V}}I\big\{D_{{\mathcal S}_{\theta}}\big(T_p(\*u^s, \btheta) \big) \leq s \big\}$. Then, for $\widetilde\btheta = (\eta, \widetilde {\*\beta}^\top)^\top,$
$\nu(\*u; \eta, \*\beta, \widetilde {\*\beta}) = \inf\nolimits_{\{\*u^*: G(\widetilde\theta,\*u^*)= G(\theta, \*u)\}} \left\{1- F_{\mathcal V|D} \left(D_{{\mathcal S}_{\widetilde\theta}}\big(T_p(\*u^*, \widetilde\btheta) \big)\right)\right\}.$
\end{lemma}

Using this lemma, we can evaluate  
 $T_p(\*u, (\eta, {\*\beta}^\top)^\top)$ in (\ref{eq:CIeta}). 
Specifically, for a given $(\*u, (\eta, {\*\beta}^\top)^\top)$,  $\nu(\*u; \eta, \*\beta, \widetilde {\*\beta})$ in Lemma~\ref{lem:nuisance_depth}  is a function of $\widetilde {\*\beta}$ only. We can use \eqref{eq:T_p} and a build-in optimization function in R (or another computing package) to obtain $T_p(\*u, (\eta, {\*\beta}^\top)^\top)$.  An illustration of this approach is in the case study example in Section~\ref{sec: numerical}, where we use it to construct a confidence set of \eqref{eq:CIeta} for the unknown number of components $\tau_0$ in a Gaussian mixture model. See, also   \cite{chen2023exact, wang2022highdimenional} for additional examples.

Finally, We provide the proof of Lemma~\ref{lem:nuisance_depth} as follows. 
\begin{proof}[Proof of Lemma~\ref{lem:nuisance_depth}]
Let $B_{1 -\alpha}(\btheta) = \{\*t: F_{\mathcal V|D}\big(D_{{\mathcal S}_{\theta}}(\*t)\big) \geq \alpha\}.$ Then $\Gamma_{1-\alpha}(\*y)$ inside \eqref{eq:nu} is  
\begin{align*}
    \Gamma_{1-\alpha}(\*y) = \left\{\btheta: \exists \*u^*, \mbox{\,s.t. } \*y = G(\*u^*,\btheta),  F_{\mathcal V|D}\left(D_{{\mathcal S}_{\theta}}\big(T_p(\*u^*, \btheta)\big) \right)  \geq \alpha \right\}.
\end{align*}
Therefore, it follows that
\begin{align*}
    \nu (\*u, \eta, \*\beta, \widetilde{\*\beta})& =\inf\big\{1-\alpha': \exists \*u^* \mbox{ s.t. }  G(\*u,\btheta) = G(\*u^*,\widetilde{\btheta}),   1-\alpha' \geq 1-F_{\mathcal V|D}\big(D_{{\mathcal S}_{\widetilde\theta}}\big(T_p(\*u^*, \widetilde{\btheta})\big) \big) \big\} \\
    & = \inf\nolimits_{\{\*u^*: G(\*u^*, \widetilde\theta)= G(\*u, \theta)\}} \big\{1- F_{\mathcal V|D}\big(D_{{\mathcal S}_{\widetilde\theta}}\big(T_p(\*u^*, \widetilde\btheta) \big)\big)\big\}.
\end{align*}  
\end{proof}

\subsection{Techniques to obtain $B_{1 -\alpha}(\btheta)$ and to improve computing efficiency}
\label{sec:computation}

\subsubsection{Using data depth approach to find $B_{1 -\alpha}(\btheta)$} \label{sec:data_depth}
Specifically, let $D_{{\mathcal S}_{\theta}}(\*t)$ be the empirical depth function for any $\*t \in \mathcal T$ that is computed based on the Monte-Carlo points in ${\mathcal S}_{\theta} = \{T(\*u^s, {\btheta})$, $\*u^s \in {\mathcal V}\}$. Then, $F_{{\mathcal V}|D} (t) = \sum_{\*u^s \in {\mathcal V}} I\big(D_{{\mathcal S}_{\theta}}(T(\*u^s, $ ${\btheta})) \leq t \big)\big/|{\mathcal V}|$ is the empirical cumulative distribution function of $D_{{\mathcal S}_{\theta}}(T(\*u^s, {\btheta})).$  By Lemma~1(a) of \cite{Liu2021}, the level $1 - \alpha$ (empirical) central region $B_{1 -\alpha}({\btheta}) = \{\*t: F_{{\mathcal V}|D}(D_{{\mathcal S}_{\theta}}(\*t)) \geq \alpha\}$ is a level $1 - \alpha$ Borel set on ${\mathcal T}$. Here, the empirical depth function $D_{{\mathcal S}_{\theta}}(\cdot)$ can be  computed, for example, by R package {\it ddalpha} \citep{pokotylo_depth_2019}.

\subsubsection{ Use of quantile regression to find $B_{1 -\alpha}(\btheta)$}
\label{sec:quantile}

When the target parameter(s) are discrete, we develop in Sections 3.2 and 4.2 an effective method to significantly reduce the search space of the discrete parameters. Here, we consider the case when the target parameters are continuous. We discuss how to utilize the continuity and quantile regression techniques to side-step the issues of (grid) search of the entire parameter space. 

 Without loss of generality, we first assume that our parameter of interest is $\btheta$ in a continuous parameter space $\Theta$ and we do not have any nuisance
 parameters. 
We further assume that 
$T(\*U, \btheta)$ is a mapping from 
${\mathcal U} \times \Theta \to {\mathcal T} \in \RR$,
and it is continuous in $\btheta$. In this case, we can use the following technique based on quantile regression to improve the computation efficiency of the repro samples method in Algorithm~\ref{alg:Ag}.  
\begin{itemize}
[leftmargin=1.5em, labelsep=0em, 
align=parleft,   itemsep=0pt,
  parsep=0pt,
  topsep=1pt] 
    \item 
Let $\pi(\btheta)$ be a proposal distribution that has no-zero probability over the range of $\Theta$. In addition to $|{\mathcal V}|$ copies of $\*u^s \sim \*U$, we  simulate $|{\mathcal V}|$ copies of $\btheta^s \sim \pi(\btheta)$. We then fit a level $1 - \alpha$ quantile regression model with $t^s = T(\*u^s, {\btheta}^s)$ as response and ${\btheta^s}$ as covariates. We output the conditional level $1 - \alpha$ quantile, say $\widehat C(1-\alpha|\theta))$,  
via regression of $t^s$ on $\btheta^s$. 
The Borel set $B_{1 -\alpha}(\btheta)$ for our repro samples method is either $(-\infty, \widehat C(1-\alpha|\theta))$ or $(\widehat C(\alpha/2|\theta), \widehat C(1 - \alpha/2|\theta))$. 
\end{itemize}
This method is similar to Algorithm 1 of \cite{dalmasso2022} in which the authors use quantile regression to get a level $1 - \alpha$ rejection region for a Neyman-Pearson test problem. The difference is that the $t^s$ in \cite{dalmasso2022} is computed from a test statistic while our $t^s$ is from a nuclear mapping function. Ours is simpler since we do not need to generate any data points $\*y^s= G(\btheta^s, \*u^s)$, especially in cases when $\*y^s$ is expensive to generate. 
\citet[][Theorems~1]{dalmasso2022} showed that the output conditional level $1 - \alpha$ quantile (say, $\widehat C(\alpha|\btheta))$ 
via regression of $t^s$ on $\btheta^s$ is a level $1 - \alpha$ critical value, if the quantile regression estimator is consistent 
\citep{koenker1994quantile,takeuchi2006nonparametric}, as $|{\mathcal V}| \to \infty$. Similarly, in our case,  if the quantile regression estimator is consistent, then $B_{1 -\alpha}(\btheta) = (- \infty, \widehat C(1-\alpha|\theta))$ or $B_{1 -\alpha}(\btheta) = (\widehat C(\alpha/2|\theta), \widehat C(1 - \alpha/2|\theta))$ satisfies (\ref{eq:B}).

Using the quantile regression method, 
we only need to calculate total $|{\mathcal V}|$ copies of $t^s = T(\*u^s, {\btheta^s})$, instead of computing $|{\mathcal V}|$ copies of $T(\*u^s, {\btheta})$ at each grid $\btheta$ value.
The method often significantly improves our computational efficiency.
Here, we require that $T(\*u, {\btheta})$ be a smooth function of ${\btheta}$ to ensure the consistency of quantile regression estimators. 
The method holds for any $\pi(\btheta)$, as long as the proposal distribution $\pi(\btheta) > 0$ for all $\btheta \in \Theta$.
Practically, we often pick a proposal distribution $\pi(\btheta)$ that is easy to simulate from and also that is dispersed over the range of $\Theta$ (so that there is not a region that $\btheta^s$ has a much smaller chance to reach than other regions in $\Theta$). For example, 
if $\Theta$ is a bounded set, $\pi(\btheta)$ can be a uniform distribution over $\Theta$. If $\Theta$ is unbounded, $\pi(\btheta)$ can be a Gaussian distribution with a large variance or a heavy-tailed distribution such as Laplace.

When $\*t^s = T(\*u^s, {\btheta}^s)$ is a vector, we instead fit a multivariate nonparametric quantile regression using  the recent work by  
\cite{delbarrio2022nonparametric}, in which the authors adopt the optimal transport concept to define multivariate center-outward quantiles and develop a quantile regression approach. 
Theoretically for the multivariate case, \citet[][Theorem 3.2]{delbarrio2022nonparametric} shows that the output level $1 - \alpha$
conditional (nested center-outward) quantile regression contour/region, say $\widehat C_{\pm}(1-\alpha|\theta)$,  is a  consistent reconstruction of their population version, when $|{\mathcal V}| \to \infty$ and under some regularity conditions. It follows that, when $|{\mathcal V}| \to \infty$ and by taking $B_n(\btheta) = \widehat C_{\pm}(1-\alpha|\theta))$, we have \eqref{eq:B}.

Finally, suppose  $\btheta = (\eta, {\*\beta}^\top)^\top$, in which $\eta$ is the parameter of interest and $\*\beta$ are nuisance parameters. If the true parameter value $\eta_0$ is inside a continuous space ${\it \Upsilon}$ and $T_p(\*u, \*\theta) = T_p(\*u, (\eta, {\*\beta}^\top)^\top)$ defined in Section 3.1 is continuous in $\eta$, we can use the same quantile regression techniques to reduce the computational burden of the profile nuclear mapping function in  \eqref{eq:T_p}. 
Specifically, suppose we have simulated $|{\mathcal V}|$ copies of $\*u^s \sim \*U$ and $|{\mathcal V}|$ copies of $\eta^s \sim \pi(\eta)$, where $\pi(\eta)$ is a proposal distribution that has non-zero probability over the range of ${\it \Upsilon}$. 
For any $\*\beta$, we can compute $t^s = T_p(\*u^s, (\eta^s, {\*\beta}^\top)^\top)$ and fit a quantile regression model with $t^s$ being the response and $\eta^s$ being the covariate. The output conditional level-$\alpha'$ quantile can give us a level-$\alpha'$ Borel set, based on which we can get $\Gamma_{\alpha'}(\*y)$. We can then define $\nu(\*u; \eta, \*\beta, \widetilde {\*\beta})$ for a given $(\*\beta, \widetilde {\*\beta})$ as in \eqref{eq:nu}
and use an optimization program to get rid of $\widetilde {\*\beta}$ and obtain a profile nuclear mapping function $T_p(\*u, \btheta)$ as defined in  \eqref{eq:T_p}. In the situation when $\eta$ is continuous, 
this method can often reduce the computing effort of the direct Monte-Carlo method in Lemma~\ref{lem:nuisance_depth}.

\subsubsection{Data-dependent candidate set}
\label{sec:data_cand} 

Another technique for reducing the computation cost 
is to reduce the search space of $\Theta$, as we discussed in Section~\ref{sec:candidate_general} for discrete or non-numerical parameters. In this subsection, we state a necessary condition for 
constructing a data-dependent candidate set (a subset of $\Theta$), say 
$\widehat \Theta(\*y_{obs})$, so that we can still maintain inference validity when searching through  $\widehat \Theta(\*y_{obs})$ instead of the entire $\Theta$. 
Here, $\widehat \Theta(\*y_{obs})$ is pre-screened using the observed data $\*y_{obs}$. 
While there can be different ways to obtain this pre-screened candidate set, a key question 
is to overcome the potential problem of ``double usage" of the observed data.
In particular, 
after constructing this candidate set $\widehat \Theta = \widehat \Theta(\*y_{obs})$, we obtain a level $1 - \alpha$ confidence set as
\begin{align}
\label{eq:G1_candidate}
& \Gamma_{1-\alpha}'(\*y_{obs})   = \widehat \Theta(\*y_{obs}) \cap \Gamma_{1-\alpha}(\*y_{obs}) \nonumber \\
& \qquad = \big\{\btheta:  \exists \, \*u^* \in {\mathcal U}, \btheta \in \widehat\Theta,  \, \hbox{s.t.} \, \*y_{obs} = 
G_z({\btheta}, \*u^*),
\, 
T(\*u^*, \btheta)  \in B_{1 - \alpha}(\btheta)  \big\}.
\end{align}
Here, the ``double usage" refers to the fact that we use the observed data in constructing both the candidate set $\widehat \Theta(\*y_{obs})$ and confidence set $\Gamma_{1-\alpha}(\*y_{obs})$.

To achieve our goal, a necessary condition that we impose on our data-dependent candidate set is
\begin{equation} \label{eq:cand-cond}
    \P\big\{ \btheta_0 \in \widehat \Theta(\*Y)\big\} = 1 - o(\delta'),
\end{equation}
where $\delta' > 0$ is a small number as described in Theorem~1. 
The corollary below 
suggests that  $\Gamma_{1-\alpha}'(\*y_{obs})$ in (\ref{eq:G1_candidate})
remains to be an (approximate) level $1 - \alpha$ confidence set. The result applies to both situations when $\theta$ is continuous or discrete (or non-numerical). This corollary (in fact a slight extension to accommodate possible nuisance parameters) has been implicitly used in the proofs of Sections 3.2, 4.2, and 4.3.

\begin{corollary}
\label{cor:cand_cs_coverage}
Let $\*Y= G(\*U, \btheta_0)$ and $\widehat\Theta = \widehat\Theta(\*Y)$ is a candidate set that satisfies (\ref{eq:cand-cond}).  
Then, we have 
$\P\big\{\btheta_0 \in \Gamma_{1-\alpha}'(\*Y)\big\} \geq 1 - \alpha - o(\delta').$
\end{corollary}

\begin{proof}[Proof of Corollary~\ref{cor:cand_cs_coverage}] 
By the fact that $\P\big\{T(\*U, \btheta_0) \in B_{1 - \alpha}(\btheta_0)\big\} \geq 1 - \alpha,$ we have 
\begin{align*}
   \P\big\{T(\*U, \btheta_0) \in B_{1 - \alpha}(\btheta_0)\big\} &  = \P\big\{T(\*U,\btheta_0) \in B_{1 - \alpha}(\btheta_0), \btheta_0 \in \widehat\Theta\big\} +  \P\big\{T(\*U,\btheta_0) \in B_{1 - \alpha}(\btheta_0), \btheta_0 \not\in \widehat\Theta\big\}  \\
   &  = \P\big\{\theta_0 \in \Gamma_{1-\alpha}'(\*Y)\big\} +  \P\big\{T(\*U,\btheta_0) \in B_{1 - \alpha}(\btheta_0), \btheta_0 \not\in \widehat\Theta\big\} \geq 1- \alpha.
\end{align*}
It then follows that
   $\P\big\{\theta_0 \in \Gamma_{1-\alpha}'(\*Y) \big\} \geq 1 - \alpha - \P\big\{T(\*U,\btheta_0) \in B_{1 - \alpha}(\btheta_0), \btheta_0 \not\in \widehat\Theta\big\} \geq 1 - \alpha - \P\big\{\theta_0 \not\in \widehat\Theta\big\}$
and thus Corollary~\ref{cor:cand_cs_coverage}.
\end{proof}

If the size of the candidate set $\widehat \Theta =  \widehat \Theta(\*y_{obs})$ is much smaller than that of $\Theta$, we can significantly lower the computing cost of Algorithm~\ref{alg:Ag} and related computations, by only considering $\btheta \in \widehat \Theta$. This reduction is particularly achievable when the parameter space is discrete. In Sections 3.2 and 4.2, we discuss an inherent ``many-to-one'' mapping in the repro samples method, based on which we can effectively obtain candidate sets that satisfy the condition (\ref{eq:cand-cond}).
These candidate sets can significantly reduce the search space for discrete parameters and also accommodate cases with additional nuisance parameters.

\section{Proofs, additional materials and examples in Section~\ref{sec:mixture}}\label{sec:appendix_mixture}
\subsection{Derivation of \eqref{eq:mix4} and proofs of theoretical results in Section~\ref{sec:mixture} }\label{sec:proof_sec4}

\subsubsection{Derivation of the confidence set in \eqref{eq:mix4}.} \label{sec:27derivation}

 Following Lemma~\ref{lemma:nuisance} and  \eqref{eq:CIeta1}, we have
\begin{align}
\Xi_{1-\alpha}(\*y_{obs})   & = \big\{\tau: \exists \*u^* \mbox{ and } (\*\mu_\tau, \*
    \sigma_\tau, \*M_\tau)   \mbox{ s.t. } \*y_{obs} =\*M_{\tau} \*\mu_\tau + \diag(\*M_\tau  {\*\sigma_\tau}) \*u^*,  T_p(\*u^*,\btheta)\leq 1-\alpha\big\} \nonumber \\
    & = \big\{\tau: \exists \*u^* \mbox{ and } (\*\mu_\tau, \*
    \sigma_\tau, \*M_\tau)   \mbox{ s.t. } \*y_{obs} =\*M_{\tau} \*\mu_\tau + \diag(\*M_\tau  {\*\sigma_\tau}) \*u^*, \widetilde T_p(\*y^*,\tau)\leq 1-\alpha, \nonumber \\
    & \hspace{2cm} \*y^* =\*M_{\tau} \*\mu_\tau + \diag(\*M_\tau  {\*\sigma_\tau}) \*u^*\big\} \nonumber \\
     &  = \resizebox{0.87\hsize}{!}{$ \big\{\tau: \exists \*u^* \mbox{ and } (\*\mu_\tau, \*
    \sigma_\tau, \*M_\tau)   \mbox{ s.t. } \*y_{obs} =\*M_{\tau} \*\mu_\tau + \diag(\*M_\tau  {\*\sigma_\tau}) \*u^*, \widetilde T_p(\*y_{obs},\tau)\leq 1-\alpha\big\}$} \nonumber  \\
    & =  \big\{\tau: \widetilde T_p(\*y_{obs},\tau)\leq 1-\alpha \big\}. \nonumber
\end{align}
The second equality follows from \eqref{eq:nuclear_final_mixed},  and the fourth equality 
holds since we can always find a $\*u^*$ that matches $\*y_{obs}$ for any $\btheta$ by setting $\*u^* = (\diag(\*M_\tau  {\*\sigma_\tau}))^{-1} (\*y_{obs} -\*M_{\tau} \*\mu_\tau).$
\hfill $\square$.

\subsubsection{Proofs of Corollary~\ref{cor:mix4_new} and Theorem~\ref{the: nuc_new_unif}}
\label{sec:B3-1}

\begin{proof}[Proof of Corollary~\ref{cor:mix4_new}]
By \eqref{eq:nuclear_final_mixed}, \eqref{eq:mix4} and that the condition  $  \P(\widetilde T(\*Y, (\tau, \*M_\tau)) \in B_{1 -\alpha}) = \P(T_p(\*U, \*\theta) \in B_{1 -\alpha}) \geq 1- \alpha$  holds for any $\btheta,$ we have
    \begin{align*}
          \P(\tau_0 \in \Xi_{1-\alpha}(\*Y))   \geq   \P(\widetilde T_p(\*Y,\tau_0)\leq 1-\alpha )  \geq   \P(\widetilde T(\*Y,(\tau_0, \*M_{0}))\leq 1-\alpha ) \geq 1- \alpha.  
    \end{align*}
\end{proof}

\begin{proof}[Proof of Theorem~\ref{the: nuc_new_unif}]
Let $\btheta = (\tau,{\*M_\tau}, \*\mu, \*\sigma)$ be a given set of parameters, $\*Y =\*M_{\tau} \*\mu + \diag(\*M_\tau  {\*\sigma}) \*U$, and $\*y$ be its realization with a realized $\*u$ of $\*U$. 
For the clarity of the proof, we slightly modified the notations here by rewriting $\A_{\*M_\tau}(\*y) = \A_\tau(\*y) $, $\B_{\*M_\tau}(\*y) = \B_\tau(\*y)$ and $\C_{\*M_\tau}(\*y) = \C_\tau(\*y)$.  

Because, when given $(\tau, {\*M}_\tau)$, the statistics  $(\A_{\*M_\tau}(\*Y), \B_{\*M_\tau}(\*Y))$ are minimal sufficient for $(\*\mu, \*\sigma)$ and $\C_{\*M_\tau}(\*Y) = \C_{\*M_\tau}(\*U) $ is ancillary, it follows from the Basu's Theorem that 
$$\{\A_{\*M_\tau}(\*Y),  \B_{\*M_\tau}(\*Y)\} \perp \C_{\*M_\tau}(\*Y) = \C_{\*M_\tau}(\*U).$$ Also, when given  $\{\A_{\*M_\tau}(\*Y),  \B_{\*M_\tau}(\*Y)\} = (\*a, \*b),$ we have 
$\*Y  
=\*M_\tau \*a+ \diag\{\*M_\tau \*b\}\C_{\*M_\tau}(\*Y) 
    = \*M_\tau \*a+ \diag\{\*M_\tau \*b\}\C_{\*M_\tau}(\*U).$
Therefore,
\begin{align}
  \left\{ \*Y|\A_{\*M_\tau}(\*Y)=\*a,\B_{\*M_\tau}(\*Y)=\*b \right\} 
   \sim \{\*M_\tau \*a+ \diag\{\*M_\tau \*b\}\C_{\*M_\tau}(\*U)\}, \nonumber
\end{align}
which does not involve $(\*\mu, \*\sigma).$ 
It then follows from the above that for any fixed $\*y$
\begin{align}
\label{eq:tilde_y_conditional}
   &  \*Y \big|\A_{\*M_\tau}(\*Y)=\A_{\*M_\tau}(\*y),\B_{\*M_\tau}(\*Y)=\B_{\*M_\tau}(\*y)  \nonumber \\
  & \sim \*M_\tau \A_{\*M_\tau}(\*y)+ \diag\{\*M_\tau \B_{\*M_\tau}(\*y)\}\C_{\*M_\tau}(\*U) \nonumber \\
    & \sim \*M_\tau \A_{\*M_\tau}(\*y)+ \diag\{\*M_\tau \B_{\*M_\tau}(\*y)\}\C_{\*M_\tau}({\*U'})= {\*Y}'.
\end{align}
That is, conditional on $\big\{\A_{\*M_\tau}(\*Y),  \B_{\*M_\tau}(\*Y)\big\} =\big\{\A_{\*M_\tau}(\*y), \B_{\*M_\tau}(\*y)\big\}$, we have 
$${\*Y}' \big| \{\A_{\tau}(\*Y'), \B_{\tau}(\*Y')\}=\{\A_{\tau}(\*y),  \B_{\tau}(\*y)\} \sim \*Y \big| \{\A_{\tau}(\*Y) \B_{\tau}(\*Y)\}=\{\A_{\tau}(\*y),  \B_{\tau}(\*y)\}.$$

Moreover, by definition
\begin{align*}
    \widetilde T(\*y, (\tau, \*M_\tau)) =  \sum_{\{\bar \tau:   P_{{\*Y}'}(\bar \tau) >   P_{{\*Y}'}(\hat \tau(\*y)) \}}  P_{{\*Y}'}(\bar \tau) \leq 1-\alpha,
\end{align*}
if and only if 
\begin{align}
\label{eq:nuclear_set_mixed}
\hat\tau(\*y) \in   \mathcal T^{[1-\alpha]},
\end{align}
where the model set $
    \mathcal T^{[1-\alpha]} = \left\{\tau^*: \sum_{\{\bar \tau:   P_{{\*Y}'}(\bar \tau) >   P_{{\*Y}'}(\tau^*) \}}  P_{{\*Y}'}(\bar \tau)\leq 1-\alpha\right\}$ only depends $(\tau, \*M_\tau)$ and the values of  
$\{\A_{\tau}(\*y),  \B_{\tau}(\*y)\}$.  We can rewrite 
 $$
 \mathcal T^{[1-\alpha]} = 
\{\bar \tau:    P_{{\*Y}'}(\bar \tau) \ge   P_{{\*Y}'}(\tau_{1-\alpha}^*) \}, $$
where
$$\tau^*_{1-\alpha} = \arg \min_{\{\tau^* \in \mathcal T^{[1-\alpha]}\}} P_{{\*Y}'}(\tau^*).$$ 
Here, $\tau^*_{1-\alpha}$ is the ``cut-off" that separates  
$\mathcal T^{[1-\alpha]} = 
\{\bar \tau:    P_{{\*Y}'}(\bar \tau) \ge   P_{{\*Y}'}(\tau_{1-\alpha}^*) \}$ and its complement set $({\mathcal T^{[1-\alpha]}})^c = 
\{\bar \tau:    P_{{\*Y}'}(\bar \tau) <   P_{{\*Y}'}(\tau_{1-\alpha}^*) \}$.

Then, it follows from  \eqref{eq:tilde_y_conditional} that
\begin{align}
\label{eq:cond_bound_nuc}
& \P_{\*Y|\A_{\*M_\tau}(\*Y), \B_{\*M_\tau}(\*Y)}\left(  \hat\tau(\*Y) \in  \mathcal T^{[1-\alpha]} \middle | \big\{\A_{\*M_\tau}(\*Y),  \B_{\*M_\tau}(\*Y)\big\} =\big\{\A_{\*M_\tau}(\*y), \B_{\*M_\tau}(\*y)\big\} \right)   \nonumber \\
& =   \P_{{\*Y}'}\left(  \hat\tau({\*Y}') \in \mathcal T^{[1-\alpha]} \right)  = \sum_{\bar \tau \in \mathcal T^{[1-\alpha]}} P_{{\*Y}'}(\bar \tau) \geq 1-\alpha.
\end{align}
The last inequality of \eqref{eq:cond_bound_nuc} holds,  
 because 
$\sum_{\tau^* \in \mathcal T^{[1-\alpha]}} P_{{\*Y}'}(\tau^*) 
=
\sum_{\{\tau^*:   P_{{\*Y}'}(\tau^*) \geq   P_{{\*Y}'}(\tau^*_{1-\alpha}) \}}  P_{{\*Y}'}(\tau^*)$ is the cumulative probability masses of all models whose probability mass $P_{{\*Y}'}(\tau^*)$ is great than  or equal to
$P_{{\*Y}'}(\tau^*_{1-\alpha})$. This cumulative mass is greater than or equal to $\alpha$.

Finally, following from \eqref{eq:nuclear_set_mixed} and since \eqref{eq:cond_bound_nuc} holds for any $\{\A_{\*M_\tau}(\*y), \B_{\*M_\tau}(\*y)\}$, we have
{
\begin{align*}
     &   \P(T_p(\*U, \*\theta) \leq 1-\alpha) = \P(\widetilde T(\*Y, (\tau, \*M_\tau)) \leq 1-\alpha)  =   \P\left( \hat \tau(\*Y) \in \mathcal T^{[1-\alpha]}\right)\\
     & = E \left[  \P_{\*Y|\A_{\*M_\tau}(\*Y), \B_{\*M_\tau}(\*Y)}\left(  \hat\tau(\*Y) \in \mathcal T^{[1-\alpha]} \middle | \A_{\*M_\tau}(\*Y),\B_{\*M_\tau}(\*Y)\right)\right] \geq 1- \alpha.
\end{align*}
}
\end{proof}

\subsubsection{Additional technical lemmas and proofs of Theorem~\ref{thm:bound_of_C_d_mixture} and Theorem~\ref{the:upper_bound}} \label{sec:B3-2}

To prove Theorem~\ref{thm:bound_of_C_d_mixture}, we first prove four additional technical lemmas. 
The proof of Theorem~\ref{thm:bound_of_C_d_mixture} is at the end of this subsection after these four Lemmas.

Specifically, Lemma~\ref{lem::angle} and \ref{lemma:u_d_sufficent} quantify the probabilities bound of the (angle) distance between   $\*U^*$ and 
$$
\*W = \diag(\*M_0\*\sigma_0) \*U,
$$
and the (angle) distance between $(\*I - \*H_\tau)\*U^* $ and $(I-\*H_{\tau})\*M_0\*\mu_0$. The corresponding realized version of $\*W$ is $\*w^{rel} = \diag(\*M_0\*\sigma_0) \*u^{rel}.$
Here,  
the (angle) distance of two vectors is measured by the square of the cosine of the angle, $$\rho(\*v_1, \*v_2) = \frac{(\*v^\top_1\*v_2)^2}{\|\*v_1\|^2\|\*v_2\|^2}, \quad \hbox{for any two vectors $\*v_1$ and $\*v_2$,}$$  
and we denote  $\*H_{\*S} = \*S (\*S^\top \*S)^{-1} \*S^\top$ as the projection matrix of $\*S$, for any matrix $\*S$ throughout this subsection.   Also, we write $\*H_\tau = \*H_{\*M_\tau}$ for simplicity
and it is the projection matrix of $\*M_\tau$. Finally, we denote 
$$
g_{\*M_\tau}(\*U)= \sqrt{1 - \rho(\*U, \*H_\tau\*U)} = \frac{\|(I-\*H_{\tau})\*U\|}{\|\*U\|}
$$ 
to be the sine distance between $\*U$ and the linear space spanned by $\*M_\tau.$ Similarly, $$
g_{\*M_\tau}(\*U^*) 
= \frac{\|(I-\*H_{\tau})\*U^*\|}{\|\*U^*\|}
\quad \hbox{and}  \quad
g_{\*M_\tau}(\*W)
= 
\frac{\|(I-\*H_{\tau})\*W\|}{\|\*W\|}.
$$  
We calculate these angle distances because
we can recover the truth $\*M_0$ with $\*M^*$ in \eqref{eq:modefied_BIC} when $\*U^*$ is in a neighborhood of $\*w^{rel}$ provided $(\*I - \*H_\tau)\*U^*$ is not too close to 
$(I-\*H_{\tau})\*M_0\*\mu_0$ in terms of the angle distance. 

Lemma \ref{lem:bound_single_ustar} concerns about multiple copies of $\*U^*$'s. It uses the calculations in  Lemma~\ref{lem::angle} and \ref{lemma:u_d_sufficent} (on a single $\*U^*$) and a bridging result in Lemma~\ref{lem: upper_coniditonal}
to derive a finite-sample probability bound that the candidate set $\widehat{\it \Upsilon}_{{\mathcal V}}(\*Y)$ of  \eqref{eq:candidate_M} covers $(\tau_0, \*M_0)$.

\begin{lemma}
\label{lem::angle}  Suppose $\tau < n$. Let $\*W= \diag(\*U)\*M_0 \*\sigma_0  = \diag(\*M_0 \*\sigma_0) \*U.$ Then, for any $0 \leq \gamma \leq 1,$ 
\begin{align}
\label{eq:angle_U_Ustar}
      \P_{(\*U, \*U^*)}\{\rho(\*W, \*U^*) > 1-\gamma^2\} > \frac{\gamma^{n-2}\arcsin (\gamma)}{n-1}. 
\end{align}
Moreover $\rho((\*I - \*H_\tau)\*W, (\*I - \*H_\tau)\*M_0\*\mu_0)$ and $\rho(\*W, \*U^*)$ are independent, $\rho(\*W, \*U^*)$ and $\*W$ are independent, and $\|\*U\|, \*W/\|\*W\|$ and $\*U^*$ are mutually independent. 
\end{lemma}

\begin{proof}[Proof of Lemma~\ref{lem::angle}]

To prove \eqref{eq:angle_U_Ustar}, we first derive the conditional distribution of $\rho(\*W, \*U^*)$,  given $\*W =w$,  
\begin{align}\label{eq:u_epsi}
    & \P_{\*U^*|\*W}\left\{{\|\*W^T \*U^*\|}\big/{(\|\*W\|\| \*U^*\|)} > \sqrt{1 - \gamma^2} \bigg|  \*W=\*w \right\}\nonumber\\
    & \qquad =  \P_{ \psi}\left\{  |\sin(\psi)|  < \gamma  \right\} \nonumber \\
    & \qquad = \frac{2}{c_1}\int_0^{\arcsin (\gamma)} \sin^{n-2}(s) ds  \nonumber \\
    & \qquad >  \frac{2}{c_1}\int_0^{\arcsin (\gamma)} (\frac{s \gamma }{\arcsin \gamma})^{n-2} ds>  \frac{\gamma^{n-2}\arcsin (\gamma)}{n-1},
\end{align}
where the first inequality follows from the fact that $\sin(s)$ is a concave function. Here,  $\psi = \arccos{\sqrt{ \rho(\*w, \*U^*)}}$ is the (positive) angle between $\*U^*$ and $\*w$, whose density function is $\sin^{n-2}(\psi)/c_1$, with a normalizing constant $c_1 = \int_0^{\pi} \sin^{n-2} (\psi) d\psi = 2 \int_0^{\frac \pi 2} \sin^{n-2} (\psi) d\psi \leq 2 \int_0^{\frac \pi 2} \sin(\psi) d\psi = 2 $. This density function is derived using a spherical transformation on $\*U^*$ in $\RR^n$ space, with $\psi$ being the first angular coordinate and a Jacobian equal to $r^{n-1} \sin^{n-2}(\psi) \prod_{ d=2}^{n-1} \sin^{n-d-1}(\psi_d)$, where $r$ is the radius and $\psi_2, \dots, \psi_{n-2}$ are the second to $(n-2)$th angular coordinates. 
Also, $sin(s) < \frac{s \gamma }{\arcsin \gamma}$ for $s \in (0, \arcsin \gamma)$ and a small $\gamma >0$.

Note that \eqref{eq:u_epsi} does not involve $\*U^*$ and $\*w$. We have
\begin{align*}
    \P_{( \*U^*, \*W)}\big\{ \rho(\*W,, \*U^*) > 1- \gamma^2 \big\} & =  \E_{\*W}\left[   \P_{\*U^*|\*W}\left\{{\|\*W^T \*U^*\|}\big/{(\|\*W\|\| \*U^*\|)} > \sqrt{1 - \gamma^2} \bigg| \*W^*\right\}\right]
  \\ & > \frac{\gamma^{n-2}\arcsin \gamma}{n-1}, 
\end{align*} 
Since $\P_{( \*U^*, \*W)}(\cdot)$ is the same as $\P_{( \*U^*, \*U)}(\cdot),$
the first result of the lemma holds.

Furthermore, from the second equation of \eqref{eq:u_epsi}, 
we see that the conditional distribution of $\rho(\*W, \*U^*)$, given $\*W  = \*w$, does not involve $\*w.$ Thus, $\rho(\*W, \*U^*)$ and $\*W$ (and thus $\*U$) are independent.
Hence, $\rho(\*U^*,\*W)$ and   $\rho((\*I - \*H_\tau)\*W, (\*I - \*H_\tau)\*M_0\*\mu_0)$ are also independent.

Finally, by the aforementioned spherical transformation, $\|\*U\|$ is independent with its direction $\*U/\|\*U\|.$ It then follows that  $\|\*U\|,$ $\*U/\|\*U\|$ and $\*U^*$ are mutually independent, since $\*U$ and $\*U^*$ are independent. Therefore because $$\frac{\*W}{\|\*W\|} = \frac{\diag(\*M_0 \*\sigma_0)\*U/\|\*U\|}{\left\|\diag(\*M_0 \*\sigma_0)\*U/\|\*U\|\right\|},$$
is a function of $\*U/\|\*U\|,$ $\|\*U\|, \*W/\|\*W\|$ and $\*U^*$ are mutually independent.   
\end{proof} 

\begin{lemma}\label{lemma:u_d_sufficent} For any $0< \gamma, \gamma_1 \leq 1,$  denote $\widetilde\gamma_1=(1-\sqrt{\gamma})\gamma_1 - \sqrt{2-2\sqrt{1-\gamma^2}}.$ 
 If $g_{\*M_\tau}(\*W) > \sqrt{\gamma}, \rho((\*I - \*H_\tau)\*W, (\*I - \*H_\tau)\*M_0\*\mu_0)< \widetilde \gamma_1^2$ and $\rho(\*W, \*U^*) > 1-\gamma^2,$ then
\begin{align*}
    \rho((\*I - \*H_\tau)\*U^*, (\*I - \*H_\tau)\*M_0\*\mu_0)<  \gamma_1^2.
\end{align*}
 
\end{lemma}

\begin{proof}[Proof of Lemma~\ref{lemma:u_d_sufficent}]
 Since $g_{\*M_\tau}(\*W) > \sqrt{\gamma}, \rho((\*I - \*H_\tau)\*W, (\*I - \*H_\tau)\*M_0\*\mu_0)< \widetilde \gamma_1^2$ and $\rho(\*W, \*U^*) > 1-\gamma^2,$ it follows that
\begin{align*}
     & \frac{1}{\|(\*I - \*H_\tau){\*U^*}\|}{\*U^*}^{T}(\*I - \*H_\tau)\*M_0\*\mu_0\\ 
     & =   \frac{1}{\|\*W\|} \*W^T(\*I - \*H_\tau)\*M_0\*\mu_0 + \left(\frac{{\*U^*}^T}{\|{\*U^*}\|}-\frac{\*W^T}{\|\*W\|}\right)(\*I - \*H_\tau)\*M_0\*\mu_0   \\& \qquad \qquad  +  \left(\frac{1}{\|(\*I - \*H_\tau){\*U^*}\|}-\frac{1}{\|{\*U^*}\|}\right){\*U^*}^T(\*I - \*H_\tau)\*M_0\*\mu_0\\
     & \leq \frac{\|(\*I - \*H_\tau)\*W\|}{\|\*W\|}\frac{1}{\|(\*I - \*H_\tau)\*W\|}\*W^T(\*I - \*H_\tau)\*M_0\*\mu_0   +  \left\|\frac{{\*U^*}^T}{\|{\*U^*}\|}-\frac{\*W^T}{\|\*W\|}\right\|\|(\*I - \*H_\tau)\*M_0\*\mu_0\| \\
     &  \hspace{0.5cm} + \frac{\|{\*U^*}\|-\|(\*I - \*H_\tau){\*U^*}\|}{\|{\*U^*}\|}\frac{1}{\|(\*I - \*H_\tau){\*U^*}\|}{\*U^*}^T(\*I - \*H_\tau)\*M_0\*\mu_0 \\
     &  \leq g_{\*M_\tau}(\*W)\widetilde\gamma_1\|(\*I - \*H_\tau)\*M_0\*\mu_0 \|  +  \sqrt{2-2\frac{{\*U^*}^\top \*W}{\|{\*U^*}^\top\| \|\*W\|}}\|(\*I - \*H_\tau)\*M_0\*\mu_0\| \\ 
     & \hspace{2cm} +(1-g_{\*M_\tau}({\*U^*}))\frac{1}{\|(\*I - \*H_\tau){\*U^*}\|}{\*U^*}^T(\*I - \*H_\tau)\*M_0\*\mu_0\\
     &   \leq g_{\*M_\tau}(\*W) \widetilde\gamma_1\|(\*I - \*H_\tau)\*M_0\*\mu_0 \|  +  \sqrt{2-2\sqrt{1-\gamma^2}}\|(\*I - \*H_\tau)\*M_0\*\mu_0\|\\
     & \hspace{2cm}+ (1-g_{\*M_\tau}({\*U^*}))\frac{1}{\|(\*I - \*H_\tau){\*U^*}\|}{\*U^*}^T(\*I - \*H_\tau)\*M_0\*\mu_0.
\end{align*}
Therefore, after a re-arrangement, we have 
\begin{align}
\label{eq:ABC}
    &  \frac{1}{\|(\*I - \*H_\tau){\*U^*}\|}{\*U^*}^T(\*I - \*H_\tau)\*M_0\*\mu_0 \nonumber
    \\     &  \qquad 
    \leq  \left\{\frac{g_{\*M_\tau}(\*W)}{g_{\*M_\tau}({\*U^*})} \widetilde\gamma_1
    + \frac{1}{g_{\*M_\tau}({\*U^*})} \sqrt{2-2\sqrt{1-\gamma^2}}\right\}\|(\*I - \*H_\tau)\*M_0\*\mu_0\|.
\end{align}
Furthermore, 
since $\|(\*I - \*H_\tau)\*W\|  \leq \|(\*I - \*H_\tau   \*H_{\*U^*}){\*W}\|$ and $\|\*H_{\*U^*}{\*W}\| \leq \|\*W\|$,
we have
\begin{align*}
    g_{\*M_\tau}({\*W})  & = \frac{\|(\*I - \*H_\tau){\*W}\|}{\|\*W\|} \leq \frac{\|(\*I - \*H_\tau   \*H_{\*U^*}){\*W}\||}{\|\*W\|}   
    \\ & \leq \frac{\|(\*I -  \*H_{\*U^*}){\*W}\|}{\|{\*W}\|} + \frac{\|(  \*H_{\*U^*}-\*H_{\tau}   \*H_{\*U^*}){\*W}\|}{\|{\*W}\|} \\
    & \leq \gamma + \frac{\|(\*I - \*H_\tau)  \*H_{\*U^*}{\*W}\|}{\|  \*H_{\*U^*}{\*W}\|} { =}  \gamma +  \frac{\|a(\*I - \*H_\tau)  {\*U^*}\|}{\|  a{\*U^*}\|} = \gamma + g_{\*M_\tau}(\*U^*),
\end{align*}
where $\*H_{\*U^*}\*W$ is a projection on the space expanded by $\*U^*$ and thus $\*H_{\*U^*}\*W = a \*U^*$ for some scalar $a$ (and $a \not = 0$ with probability $1$). 
Then 
\begin{align}
\label{eq:ABC1}
\frac{g_{\*M_\tau}({\*U^*})}{g_{\*M_\tau}(\*W)}\geq 1 - \frac{\gamma}{g_{\*M_\tau}(\*W)}.     
\end{align}
It follows from (\ref{eq:ABC}) that a sufficient condition for $$\frac{1}{\|(\*I - \*H_\tau){\*U^*}\|}{\*U^*}^T(\*I - \*H_\tau)\*M_0\*\mu_0 \leq \gamma_1\|(\*I - \*H_\tau)\*M_0\*\mu_0\|,$$ is 
$\left\{\frac{g_{\*M_\tau}(\*W)}{g_{\*M_\tau}({\*U^*})} \widetilde\gamma_1
    + \frac{1}{g_{\*M_\tau}({\*U^*})} \sqrt{2-2\sqrt{1-\gamma^2}}\right\} \leq \gamma_1$
, which can be implied by 
\begin{align*}
\widetilde\gamma_1\leq \left(1 - \frac{\gamma}{g_{\*M_\tau}(\*W)}\right) \gamma_1 - \sqrt{2-2\sqrt{1-\gamma^2}},   
\end{align*}
since we have 
(\ref{eq:ABC1}) and $g_{\*M_\tau}(\*W)  > \sqrt{\gamma}.$ Here, 
$\widetilde\gamma_1=(1-\sqrt{\gamma})\gamma_1 - \sqrt{2-2\sqrt{1-\gamma^2}}$. 

\end{proof}

\begin{lemma}
\label{lem: upper_coniditonal}
    Consider two events $\mathcal A$ and $\mathcal B$, where $\mathcal B = \bigcup_{1 \leq d \leq |{\mathcal V}|} B_d$ is a union of a sequence of $|{\mathcal V}|$ events. 
    We define the events $D_1 = B_1$ and $D_d = B_d \bigcap\big(\bigcup_{1\leq j \leq d-1} B_j\big)^c
    $, for $d = 2, \ldots, |{\mathcal V}|.$
    Suppose there exists an upper bound for conditional probability $\P(\mathcal A | D_d) \leq \bar p$ for any $1 \leq d \leq |{\mathcal V}|$. Then, the joint probability $\P(\mathcal A \cap \mathcal B) \leq \bar p.$ 
\end{lemma}

\begin{proof}
By the definition, $D_d$'s are disjoint events and
${\mathcal B} = \bigcup_{1 \leq d \leq |\mathcal V|} D_d$. Thus,  $\{D_d:  d= 1, \dots, |{\mathcal V|}\}$ is a partition of $\mathcal B$. It follows that 
\begin{align*}
\P(\mathcal A \cap \mathcal B) = \sum_{d =1}^{|{\mathcal V}|} \P(\mathcal A | D_d) \P(D_d) \leq \bar p \sum_{d = 1}^{|{\mathcal V|}}  \P(D_d)  = \bar p \, \P({\mathcal B}) \leq \bar p.    
\end{align*}

\end{proof}

\begin{lemma}
\label{lem:bound_single_ustar}
Suppose $n-\tau_0>4$. For any positive small $0 < \gamma < (\frac{1}{2.1})^4$ such that 
$\gamma^{\frac 32} < \frac{\log\{0.5C^2_{\min}\gamma + 1\}}{2 \tau_0 } $ , and $\lambda \in [\frac{\gamma^{\frac32}}{\log(n)/n},\frac{\log\{0.5C^2_{\min}\gamma + 1\}}{2 \tau_0 \log(n)/n}],$ we have 
\begin{align}
\label{eq:prob_bound_C_D_simple-1}
      &  \P_{\*U,\mathcal V  }\left\{(\tau_0, \*M_0) \not\in \widehat{\it \Upsilon}_{{\mathcal V}}(\*Y)\right\} 
       \leq \sum^{\tau_{\max}}_{\tau=\tau_0+1}
    \exp\left\{-\frac{2(\tau-\tau_0)}{3\sigma^2_{\max}\sqrt{\gamma}}    + 0.25n +  \tau\log(n) \right\} \nonumber \\
    & \quad \qquad  +  \sum_{\tau=1}^{\tau_0} \exp\left\{-\frac{1}{75\sigma^2_{max}\gamma} C_{\min}\ + 0.25n + \tau\log(n)\right\}
     + p_\gamma + \left(1-\frac{\gamma^{n-1}}{n-1}\right)^{|{\mathcal V}|},
\end{align}
where $\tau_{\max} \geq \tau_0$ is any upper bound for $\tau_0,$ and $p_\gamma \rightarrow 0$ as $\gamma \rightarrow 0.$ 
\end{lemma}

\begin{proof}[Proof of Lemma~\ref{lem:bound_single_ustar}]

We can re-express 
\eqref{eq:modefied_BIC} by 
\begin{align*}
(\tau^*, \*M^*_{\tau^*}) = \arg\min\limits_{(\tau, \*M_\tau)}   \left\{n\log \bigg(\frac{\|(\*I-\*H_{\*M_\tau, \*u^*})\*y_{obs}\|^2+1}{n}\bigg) + 2\lambda\tau\log(n)\right\}.
\end{align*}
where $\*H_{\*M_\tau, \*u^*}$ is a projection matrix to the space expanded by  $(\*M_\tau, \*u^*)$. In the following, we let $\*M^* = \*M^*_{\tau^*}$ for simplicity.  
By \eqref{eq:modefied_BIC},  if $\*M^* = \*M_{\tau},$ then
\begin{align*}
    \|(\*I - \*H_{\*M_{\tau}, \*U^*}) \*y_{obs}\|^2 + 1  - e^{2\lambda(\tau_0-\tau)\log(n)/n} \{\|(\*I - \*H_{\*M_0, \*U^*}) \*y_{obs}\|^2 +1\} \leq 0.
\end{align*} 
Correspondingly, for the random sample version with $\*Y = {\*M_0} {\*\mu_0} + \diag(\*M_0 {\*\sigma_0}) \*U = {\*M_0} {\*\mu_0} + \*W$, it  
is then
\begin{align}
\label{eq:A3_inequality1}
    \|(\*I - \*H_{\*M_{\tau}, \*U^*}) \*Y\|^2 + 1  - e^{2\lambda(\tau_0-\tau)\log(n)/n} \{\|(\*I - \*H_{\*M_0, \*U^*}) \*W\|^2 +1\} \leq 0.
\end{align} 
which implies 
\begin{align}
\label{eq:A3_inequality1a}
     1  - e^{2\lambda(\tau_0-\tau)\log(n)/n} \{\|(\*I - \*H_{\*M_0, \*U^*}) \*W\|^2+1\} \leq 0.
\end{align} 

First, let's consider the case for $\tau > \tau_0.$
Denote by a set 
\begin{align*}
 B^*  =  \bigg\{\rho(\*W, \*U^*) > 1-\gamma^2, \rho((\*I - \*H_\tau)\*U^*, (\*I - \*H_\tau)\*M_0\*\mu_0)<  \gamma_1^2\bigg\}. 
\end{align*}
Also, denote $\P_{(\*U, \*U^*|B^*)}$ as the conditional probability corresponding to the conditional distribution of $(\*U, \*U^*)$ given $B^*$ happens.
  Then, when $\lambda \geq \frac{\gamma^{\frac32}}{\log(n)/n},$ and by \eqref{eq:A3_inequality1a} and Lemma~\ref{lemma:u_d_sufficent}, the conditional probability
\begin{align*}
     & \P_{(\*U, \*U^*|B^*)}\bigg\{{\*M}^* = \*M_\tau \bigg | B^*\bigg\} \\
    & \leq \P_{(\*U, \*U^*|B^*)}\bigg\{ - e^{2\lambda(\tau_0-\tau)\log(n)/n} \|(\*I-\*H_{\*M_0, \*U^*})\*W\|^2+1 -e^{2\lambda(\tau_0-\tau)\log(n)/n} < 0 \bigg| B^*\bigg\}\\
    & \leq \P_{(\*U, \*U^*|B^*)}\bigg\{ - e^{2\lambda(\tau_0-\tau)\log(n)/n} \|(\*I-\*H_{\*U^*})\*W\|^2+1 -e^{2\lambda(\tau_0-\tau)\log(n)/n} < 0 \bigg |  B^*\bigg\} \\
    & \leq \P_{(\*U, \*U^*|B^*)}\bigg\{ - e^{2\lambda(\tau_0-\tau)\log(n)/n}\gamma^2\|\*W\|^2 +1 -e^{2\lambda(\tau_0-\tau)\log(n)/n} < 0 \bigg | B^*
    \bigg\}\\
    & = 
    \P_{(\*U, \*U^*|B^*)}\bigg\{ \|\*W\|^2 > \frac{1 -e^{2\lambda(\tau_0-\tau)\log(n)/n}}{ e^{2\lambda(\tau_0-\tau)\log(n)/n}\gamma^2} \bigg |B^*\bigg\}\\
    & \leq \P_{\*U|B^*}\bigg\{ \sigma^2_{\max}\|\*U\|^2 > \frac{1 -e^{2\lambda(\tau_0-\tau)\log(n)/n}}{ e^{2\lambda(\tau_0-\tau)\log(n)/n}\gamma^2} \bigg | B^*\bigg\}
  \\
    & \leq \P_{\*U}\bigg\{ \sigma^2_{\max}\|\*U\|^2 > \frac{1 -e^{2\lambda(\tau_0-\tau)\log(n)/n}}{ e^{2\lambda(\tau_0-\tau)\log(n)/n}\gamma^2} \bigg\}\\
    & \leq \P_{\*U}\bigg\{ \sigma^2_{\max}\|\*U\|^2 > \frac{ -2\lambda(\tau_0-\tau)\log(n)/n}{ \gamma^2} \bigg\}\\
    & \leq \P_{\*U}\bigg\{ \sigma^2_{\max}\|\*U\|^2 > \frac{ -2(\tau_0-\tau)}{ \sqrt{\gamma}} \bigg\} =  \P_{\*U}\left\{\|\*U\|^2  >  C^I_{\gamma, \tau}/\sigma^2_{\max}\right\}, \numberthis \label{eq:CI_gamma}
 \end{align*}
 where $C^I_{\gamma, \tau} = \frac{ -2(\tau_0-\tau)}{ \sqrt{\gamma}}$, and  
 the second-to-last inequality holds because $\frac{1-x}{x}\geq  -\log(x)$ for any $0 < x \leq 1.$ Also, the 
 drop of conditional distribution is because $\*U/\|\*U\|,$ $\|\*U\|, \*W/\|\*W\|$ and $\*U^*$ are mutually independent by Lemma~\ref{lem::angle}.

Because $\|\*U\|^2$ follows a  chi-square distribution, then for any $c>0$, we can bound the probability
\begin{align}
\label{eq:markov_ineuqality}
      \P_{\*U}(\|\*U\|^2 > c)  \leq E \exp\{t\|\*U\|^2 -tc\} 
    = \left(\frac{1}{1-2t}\right)^{n/2}e^{-tc}.
\end{align}
Make $t = 1/3,$ then for any $\*M_\tau, \tau > \tau_0$, the conditional probability
\begin{align}
\label{eq:bound_tau_large}
      \P_{(\*U, \*U^*|B^*)}\bigg\{{\*M}^* = \*M_\tau \bigg | B^* \bigg\} \leq \exp\left\{-\frac{1}{3\sigma^2_{\max}} C^I_{\gamma, \tau} + 0.25n \right\}.
\end{align}

Now let us consider the case of $\tau \leq \tau_0,$  
 It follows from \eqref{eq:A3_inequality1} that
\begin{align}
\label{eq:A3_inequality2}
    &   \P({\*M}^*  = \*M_\tau)
    \leq  \P\big\{ \|(\*I-\*H_{\*M_{\tau}, \*U^*})\*M_0 \*\mu_0\|^2 +  2 \*W^T (\*I-\*H_{\*M_{\tau}, \*U^*})\*M_0 \*\mu_0  \nonumber \\
    & \qquad\qquad + \|(\*I-\*H_{\*M_{\tau}, \*U^*})\*W\|^2 - e^{2\lambda(\tau_0-\tau)\log(n)/n} \|(\*I-\*H_{\*M_0, \*U^*})\*W\|^2 +1 \\
    & \qquad\qquad -e^{2\lambda(\tau_0-\tau)\log(n)/n} < 0\big\}. \nonumber
\end{align}
Then for any $\*M_\tau \neq \*M_0$ and $ \tau \leq \tau_0$,  
conditional on the event $B^*,$ when 
\begin{align}
\label{eq:lambda_upper}
    \lambda \leq \frac{\log\{C^2_{\min}(1- \gamma^2_1)/2 + 1\}}{2 \tau_0 \log(n)/n},
\end{align}
it follows from  
\eqref{eq:A3_inequality2} that,     
\begin{align*}
    &   \P_{(\*U, \*U^*|B^*)}\left({\*M}^*  = \*M_\tau \middle| B^*\right) \\
    & \leq \P_{(\*U, \*U^*|B^*)}\bigg\{\|(\*I-\*H_{\*M_\tau, \*U^*})\*M_0 \*\mu_0\|^2  + 2 \*W^T (\*I-\*H_{\*M_\tau, \*U^*})\*M_0 \*\mu_0 +1 -e^{2\lambda(\tau_0-\tau)\log(n)/n}  \\
    &   \quad  \quad \quad   - e^{2\lambda(\tau_0-\tau)\log(n)/n} \|(\*I-\*H_{\*M_0, \*U^*})\*W\|^2  < 0   \bigg |  B^*\bigg\}\\
     & \leq \P_{(\*U, \*U^*|B^*)}\bigg\{(1 - \gamma^2_1)\|(\*I-\*H_{\*M_\tau})\*M_0 \*\mu_0\|^2  - 2 \|(\*I-\*H_{\*M_\tau, \*U^*})\*W\| \|(\*I-\*H_{\*M_\tau, \*U^*})\*M_0 \*\mu_0\|  \\ 
    & \quad  \qquad  \qquad \qquad   +1 -e^{2\lambda(\tau_0-\tau)\log(n)/n} - e^{2\lambda(\tau_0-\tau)\log(n)/n} \|(\*I-\*H_{\*M_0, \*U^*})\*W\|^2  < 0 \bigg|  B^*\bigg\}\\
           & \leq \P_{(\*U, \*U^*|B^*)}\bigg\{(1 - \gamma^2_1)\|(\*I-\*H_{\*M_\tau})\*M_0 \*\mu_0\|^2  - 2 \gamma\|\*W\| \|(\*I-\*H_{\*M_\tau})\*M_0 \*\mu_0\|  +1 -e^{2\lambda(\tau_0-\tau)\log(n)/n} \\
    & \quad  \qquad  \qquad \qquad   - e^{2\lambda(\tau_0-\tau)\log(n)/n} \gamma^2\|\*W\|^2  < 0 \bigg| B^*\bigg\}\\
     & \leq \P_{(\*U, \*U^*|B^*)}\bigg\{{\|\*W\| }  > 
     \frac{ \sqrt{ C^2_{\*M_0, \*M_\tau}  +e^{\frac{2\lambda(\tau_0-\tau)\log(n)}{n}} \{(1- e^{\frac{2\lambda(\tau_0-\tau)\log(n)}{n}} )+ (1- \gamma_1^2) C^2_{\*M_0, \*M_\tau}\}   } -C_{\*M_0, \*M_\tau} }{\gamma e^{2\lambda(\tau_0-\tau)\log(n)/n} }\bigg | \\
        & \quad  \qquad  \qquad B^* \bigg\}\\
      &  \leq \P_{(\*U|B^*)}\bigg\{{\sigma_{\max} \|\*U\| }  > 
     \frac{ \sqrt{ C^2_{\*M_0, \*M_\tau}  +e^{\frac{2\lambda(\tau_0-\tau)\log(n)}{n}} \{(1- e^{\frac{2\lambda(\tau_0-\tau)\log(n)}{n}} )+ (1- \gamma_1^2) C^2_{\*M_0, \*M_\tau}\}   } -C_{\*M_0, \*M_\tau} }{\gamma e^{2\lambda(\tau_0-\tau)\log(n)/n} }\bigg | \\
        & \quad  \qquad  \qquad B^* \bigg\}\\
        & =  \P_{\*U}\bigg\{{\|\*U\|}  >  \frac{ \sqrt{ C^2_{\*M_0, \*M_\tau}  +e^{\frac{2\lambda(\tau_0-\tau)\log(n)}{n}} \{(1- e^{\frac{2\lambda(\tau_0-\tau)\log(n)}{n}} )+ (1- \gamma_1^2) C^2_{\*M_0, \*M_\tau}\}   } -C_{\*M_0, \*M_\tau} }{\sigma_{\max}\gamma e^{2\lambda(\tau_0-\tau)\log(n)/n} } \bigg\}\\
     & \leq \P_{\*U}\bigg\{\|\*U\| > C_{\min} \frac{(1-\gamma_1^2)}{5\gamma\sigma_{\
     \max}}\bigg\} = \P_{\*U}\left\{\|\*U\|  >  \sqrt{C^{II}_{\gamma,\gamma_1,\tau}}/\sigma_{\max}\right\}\\
     & \leq \exp\left\{-\frac{1}{3\sigma^2_{\max}} C^{II}_{\gamma, \gamma_1, \tau}\ + 0.25n \right\}, 
\end{align*}
where $C_{\*M_0, \*M_\tau} = \|(\*I - \*M_\tau)\*M_0\*\mu_0\|$  and $ C^{II}_{\gamma, \gamma_1 \tau} =  C^2_{\min} \frac{(1-\gamma_1^2)^2}{25\gamma^2}$. Again, the 
 drop of conditional distribution is because $\*U/\|\*U\|,$ $\|\*U\|, \*W/\|\*W\|$ and $\*U^*$ are mutually independent by Lemma~\ref{lem::angle}. 
  Also, the last inequality follows from \eqref{eq:markov_ineuqality} and making $t = 1/3.$ The second to last inequality holds because $-1+\sqrt{1+x} \geq 0.4x$ for any $0 \leq x \leq 1.$

Since $\{\*M^* \neq \*M_0\} =  \{\*M^* = \*M_\tau, \tau > \tau_0\}  \bigcup \{\*M^* \neq \*M_\tau, \*M_\tau \neq \*M_0, \tau \leq \tau_0\}$, it follows from the above and 
\eqref{eq:bound_tau_large}  that the conditional probability
\begin{align}\label{eq:aaa_1}
& \P_{(\*U, \*U^*|B^*)}\bigg\{{\*M}^* \neq \*M_0 \bigg|B^* \bigg\} \nonumber \\ 
& \leq \sum_{\{\*M_\tau: \*M_\tau \neq \*M_0, \tau \leq \tau_0\}}\P_{(\*U, \*U^*|B^*)}\bigg\{{\*M}^* = \*M_\tau \bigg| B^* \bigg\}\nonumber +  \sum_{\{\*M_\tau:  \tau > \tau_0\}}\P_{(\*U, \*U^*|B^*)}\bigg\{{\*M}^* = \*M_\tau \bigg |B^* \bigg\} \nonumber \\ 
&  \leq  \sum_{\tau=1}^{\tau_0} n^{\tau} \exp\left(-\frac{1}{3\sigma^2_{\max}}C^{II}_{\gamma, \gamma_1, \tau}\ + 0.25n \right) +  \sum_{\tau=\tau_0 + 1}^{\tau_{\max}} n^{\tau} \exp\left(-\frac{1}{3\sigma^2_{\max}}C^{I}_{\gamma,  \tau}\ + 0.25n \right) \nonumber\\
    & =  \sum_{\tau=1}^{\tau_0} \exp\left\{-\frac{1}{3\sigma^2_{\max}} C^{II}_{\gamma, \gamma_1, \tau}\ + \frac{n}4 + \tau\log(n)\right\} + \sum^{\tau_{\max}}_{\tau=\tau_0+1}
    \exp\left\{-\frac{1}{3\sigma^2_{\max}} C^I_{\gamma, \tau} + \frac{n}4 +  \tau\log(n) \right\}.
\end{align}

Suppose now we have many copies of $\*U^*$'s with
$\mathcal V   = \{\*U^s_1, \dots \*U^s_{|\mathcal V|}\},$ where $\*U^s_d \sim \*U^* \sim \*U, d= 1, \dots, |\mathcal V|$ are the random copies of $\*U^*$ in $\mathcal V,$ and let $B_d =  \bigg\{\rho(\*W, \*U^s_d) > 1-\gamma^2, \rho((\*I - \*H_\tau)\*U^s_d, (\*I - \*H_\tau)\*M_0\*\mu_0)<  \gamma_1^2\bigg\}.$ 
Now we bound the probability $\P_{\*U,\mathcal V  }\left\{(\tau_0, \*M_0) \not\in \widehat{\it \Upsilon}_{{\mathcal V}}(\*Y)\right\} $. Specifically, we  write
\begin{align} 
\label{eq:total_bound}
    &  \P_{\*U,\mathcal V  }\left\{(\tau_0, \*M_0) \not\in \widehat{\it \Upsilon}_{{\mathcal V}}(\*Y)\right\} \nonumber \\ & =  \P_{\*U,\mathcal V  }\bigg\{(\tau_0, \*M_0) \not\in \widehat{\it \Upsilon}_{{\mathcal V}}(\*Y), \bigcup_{1 \leq d \leq |\mathcal V|} B_d\bigg\}  +   \P_{\*U,\mathcal V  }\bigg\{(\tau_0, \*M_0) \not\in \widehat{\it \Upsilon}_{{\mathcal V}}(\*Y), \bigcap_{1 \leq d \leq |\mathcal V|}B_d^c \bigg\} \nonumber \\
    &  \leq \P_{\*U,\mathcal V  }\left\{(\tau_0, \*M_0) \not\in \widehat{\it \Upsilon}_{{\mathcal V}}(\*Y), \bigcup_{1 \leq d \leq |\mathcal V|}B_d\right\} + \P_{\*U}\left(\*W \in E^C_{\tau_0}(\gamma, \widetilde \gamma_1)\right) \nonumber \\ & \qquad  \qquad \qquad +  \P_{\*U,\mathcal V} \left(\bigcap_{1 \leq d \leq |\mathcal V|}\{\rho(\*W,\*U^s_d) \leq 1 - \gamma^2 \}\right),
\end{align}
where $$E_{\tau_0}(\gamma,\widetilde\gamma_1 ) = \bigcap_{\{\*M_\tau \neq \*M_0, \, \tau \leq \tau_0\}}\{\*w: g_{\*M_\tau}(\*w) > \sqrt{\gamma}, \rho((\*I - \*H_\tau)\*w, (\*I - \*H_\tau)\*M_0\*\mu_0)< \widetilde \gamma_1^2\},$$
and the last inequality follows from Lemma~\ref{lemma:u_d_sufficent}.

To bound the first term on the right-hand side of (\ref{eq:total_bound}), we use Lemma~\ref{lem: upper_coniditonal} with 
${\mathcal A} =\big\{ (\tau_0, \*M_0) \not\in \widehat{\it \Upsilon}_{{\mathcal V}}(\*Y)\big\}$, ${\mathcal B} = \bigcup_{1 \leq d \leq |\mathcal V|}B_d$. Particularly, we only need to bound the conditional probability 
$\P_{(\*U,\mathcal V) | D_d}\left\{(\tau_0, \*M_0) \not\in \widehat{\it \Upsilon}_{{\mathcal V}}(\*Y), \middle | D_d 
\right\},$ 
for $d =1, \dots, |\mathcal V|,$ where $D_d = B_d \bigcap\big(\bigcup_{1\leq j \leq d-1} B_j\big)^c$.
To this end, since 
$\big\{(\tau_0, \*M_0) \not\in \widehat{\it \Upsilon}_{{\mathcal V}}(\*Y)\big\} \subseteq \{\*M^*_{d}  \neq \*M_0  \}
$, we have 
\begin{align*}
    & \P_{(\*U,\mathcal V) | (D_d, \, \*U)}\left\{(\tau_0, \*M_0) \not\in \widehat{\it \Upsilon}_{{\mathcal V}}(\*Y) \bigg| D_d,  \*U\right\} \leq \P_{(\*U,\mathcal V) | (D_d, \, \*U)}\left\{ \*M^*_{d}  \neq \*M_0  |  D_d,  \*U\right\} \\
    & \qquad\qquad = 
    \P_{(\*U,\mathcal V) | (D_d, \, \*U)}\left\{ \*M^*_{d}  \neq \*M_0  \bigg| B_d \bigcap\big(\bigcup_{1\leq j \leq d-1} B_j\big)^c,  \*U\right\}
    \\
    & \qquad\qquad = 
    \P_{(\*U,\mathcal V) | (D_d, \, \*U)}\left\{ \*M^*_{d}  \neq \*M_0  \bigg| B_1^c, \ldots, B_{d-1}^c, B_d,  \*U\right\}  \\ & \qquad\qquad = 
    \P_{(\*U,\*U_s^d) | (B_d, \, \*U)}\left\{ \*M^*_{d}  \neq \*M_0  \bigg| B_d,  \*U\right\},
\end{align*}
 where the first equality holds because, given $\*U$, $\*U^s_j$'s remain to be iid, thus both $\*M^*_{d}$ and $B_d$ are indepent of $B_j^c$, $j =1, \ldots, d-1$.
Here, 
$\*M^*_{d}$ is obtained using (\ref{eq:modefied_BIC}) with 
 $\*U_d^s$ replacing $\*u^*$.

Since $\rho(\*W, \*U_d^s)$ is independent of $\*U_d^s$ by Lemma~\ref{lem::angle}, $B_d$ is independent of $\*U.$ It then follows immediately that 
\begin{align}
\label{eq:p_upper_bound}
    & \P_{(\*U,\mathcal V) | D_d}\left\{(\tau_0, \*M_0) \not\in \widehat{\it \Upsilon}_{{\mathcal V}}(\*Y) \middle| D_d\right\} =  E_{\*U}\left[\P_{(\*U,\mathcal V) | (D_d, \, \*U)}\left\{(\tau_0, \*M_0) \not\in \widehat{\it \Upsilon}_{{\mathcal V}}(\*Y) \bigg| D_d,  \*U\right\} \right]\nonumber\\
    & \quad \leq E_{\*U}\bigg[ \P_{(\*U,\*U_s^d) | (B_d, \, \*U)}\left\{ \*M^*_{d}  \neq \*M_0  \bigg| B_d,  \*U\right\} \bigg]  = \P_{(\*U,\*U_s^d) | B_d }\left\{ \*M^*_{d}  \neq \*M_0  \bigg| B_d\right\}\nonumber \\
   & \quad = \P_{(\*U,\*U^* | B^*)} \left\{ \*M_d^*  \neq \*M_0  \bigg| B^*\right\} \nonumber \\
   & \quad  \leq \sum_{\tau=1}^{\tau_0} \exp\left\{-\frac{C^{II}_{\gamma, \gamma_1, \tau}}{3\sigma^2_{\max}}  + \frac{n}4 + \tau\log(n)\right\} + \sum^{\tau_{\max}}_{\tau=\tau_0+1}
    \exp\left\{-\frac{C^I_{\gamma, \tau}}{3\sigma^2_{\max}} + \frac{n}4 +  \tau\log(n) \right\},
\end{align}
where the second equality holds because of (\ref{eq:aaa_1}).

 It then follows from Lemma~\ref{lem: upper_coniditonal} and \eqref{eq:p_upper_bound} that the first term  on the right-hand side of (\ref{eq:total_bound}) is bounded by  
\begin{align}
    & \P_{\*U,\mathcal V  }\left\{(\tau_0, \*M_0) \not\in \widehat{\it \Upsilon}_{{\mathcal V}}(\*Y), \bigcup_{ 1 \leq d \leq |\mathcal V|}
    B_d \right\}  \nonumber\\
     &\leq \sum_{\tau=1}^{\tau_0} \exp\left\{-\frac{C^{II}_{\gamma, \gamma_1, \tau}}{3\sigma^2_{\max}}  + \frac{n}4 + \tau\log(n)\right\} + \sum^{\tau_{\max}}_{\tau=\tau_0+1}
    \exp\left\{-\frac{C^I_{\gamma, \tau}}{3\sigma^2_{\max}} + \frac{n}4 +  \tau\log(n) \right\}.
\label{eq:inequality_single}
\end{align}

To bound the second term of on the right hand side of \eqref{eq:total_bound}, we  make $\gamma_1= \sqrt{1- \gamma^{1/2}},$ then $\widetilde\gamma_1=(1-\sqrt{\gamma})\sqrt{1- \gamma} - \sqrt{2-2\sqrt{1-\gamma^2}} \geq 1 - 2\sqrt{\gamma} \geq 0$ for $0 \leq \gamma \leq 0.25.$ Thus 
\begin{align} \label{eq:A4-1} & \P_{\*U}\left(\*W \in E^C_{\tau_0}( \gamma, \widetilde \gamma_1)\right) \nonumber
 \\
 & = \P_{\*U}\left( \*W \in \bigcup_{\{\*M_\tau \neq \*M_0, \, \tau \leq \tau_0\}}\{\*w: g_{\*M_\tau}(\*w) \leq  \sqrt{\gamma} \mbox{ or }  \rho((\*I - \*H_\tau)\*w, (\*I - \*H_\tau)\*M_0\*\mu_0) \geq  \widetilde \gamma_1^2\}\right) \nonumber
 \\
  & \leq    \sum_{\{\*M_\tau \neq \*M_0, \tau \leq \tau_0\}}\bigg[\P\left\{g_{\*M_\tau}(\*W) \leq \sqrt{\gamma}\right\} +  \P\left\{\rho((\*I - \*H_\tau)\*W, (\*I - \*H_\tau)\*M_0\*\mu_0)\geq  \widetilde \gamma_1^2\right\}\bigg] \nonumber \\
   & \leq    \sum_{\{\*M_\tau \neq \*M_0, \tau \leq \tau_0\}}\P\left\{g_{\*M_\tau}(\*W) \leq \sqrt{\gamma}\right\} \nonumber \\ & \qquad\qquad\qquad + \sum_{\{\*M_\tau \neq \*M_0, \tau \leq \tau_0\}} \P\left\{\rho((\*I - \*H_\tau)\*W, (\*I - \*H_\tau)\*M_0\*\mu_0)\geq  (1 - 2\sqrt{\gamma})^2\right\}. 
\end{align}
Since $g_{\*M_\tau}(\*W)  = \frac{\|(\*I - \*H_\tau)\*W\|}{\|\*W\|} >0$, for $|\tau| < n$,  $\P(g_{\*M_\tau}(\*W)  \leq 0) = 0$.  Furthermore, since $(\*I - \*H_\tau)\*W$ is a linear transformation of $\*W$ and $g_{\*M_\tau}(\*W)$ is a continuous random variable, it follows immediately that 
\begin{align} \label{eq:A4-2}
\P(g_{\*M_\tau}(\*W)  \leq \sqrt{\gamma}) \rightarrow 0, 
\hbox{ \,\, as $\gamma \rightarrow 0.$}  
\end{align}
Moreover, unless $(\*I-\*H_\tau)\*W = a (\*I - \*H_\tau)\*M_0\*\mu_0$ for a scalar constant $a,$
$\P\big\{\rho((\*I - \*H_\tau)\*W, (\*I - \*H_\tau)\*M_0\*\mu_0)\geq  (1 - 2\sqrt{\gamma})^2\big\} = \frac{|\*W^T(\*I - \*H_\tau) \*M_0\*\mu_0|}{\|(I - H_\tau)\*W\|\|(\*I - \*H_\tau) \*M_0\*\mu_0\|} < 1$
and $\P\big\{\rho((\*I - \*H_\tau)\*W, (\*I - \*H_\tau)\*M_0\*\mu_0)\geq  1\big\} = 0$. 
It follows that  
\begin{align} \label{eq:A4-3}
\P\left\{\rho((\*I - \*H_\tau)\*W, (\*I - \*H_\tau)\*M_0\*\mu_0)\geq  (1 - 2\sqrt{\gamma})^2\right\} \rightarrow 0, \hbox{\,\, as $\gamma \rightarrow 0.$}   
\end{align}
Together, with (\ref{eq:A4-1}) - (\ref{eq:A4-3}), we have $$\P_{\*U}\left(\*W \in E^C_{\tau_0}( \gamma, \widetilde \gamma_1)\right) \rightarrow 0, \quad \hbox{as $\gamma \rightarrow 0.$} $$

For the third term of on the right hand side of (\ref{eq:total_bound}), we follow the proof of \eqref{eq: 1-P_delta} in the proof of  Lemma~\ref{lemma:p_bound_C_D_general}. Specifically, 
\begin{align*}
   & \P_{\*U,\mathcal V} \left(\bigcap_{1 \leq d \leq |\mathcal V|}\{\rho(\*W,\*U^s_d) \leq 1 - \gamma^2 \}\right) = E_{\*U}\left[ \P_{\mathcal V|\*U} \left(\bigcap_{1 \leq d \leq |\mathcal V|}\{\rho(\*W,\*U^s_d) \leq 1 - \gamma^2 \} \middle | \*U\right)\right] \\
  &  =  E_{\*U}\left[ \left\{1 - \P_{\*U^*|\*U} \left(\rho(\*W,\*U^*) \leq 1 - \gamma^2  \middle | \*U\right)\right\}^{|\mathcal V|}\right] \\
  & \leq \left\{1 - \P_{\*U,\*U^*} \left(\rho(\*W,\*U^*) \leq 1 - \gamma^2  \right)\right\}^{|\mathcal V|} \leq \left(1-\frac{\gamma^{n-2}\arcsin (\gamma)}{n-1}\right)^{|{\mathcal V}|},
\end{align*}
by Lemma~\ref{lem::angle} and Jensesn's inequality.

With all three terms on the right-hand side of \eqref{eq:total_bound} together, we have an upper bound
\begin{align}
\label{eq:finite_p_bound_C_D}
     \P\{(\tau_0, \*M_0) \not\in \widehat{\it \Upsilon}_{{\mathcal V}}(\*Y)\} & \leq 
     \sum^{\tau_{\max}}_{\tau=\tau_0+1}
    \exp\left\{-\frac{1}{3\sigma^2_{\max}} C^I_{\gamma, \tau} + 0.25n +  \tau\log(n) \right\} \nonumber \\
    & \quad  +  \sum_{\tau=1}^{\tau_0} \exp\left\{-\frac{1}{3\sigma^2_{\max}} C^{II}_{\gamma_1,\gamma, \tau}\ + 0.25n + \tau\log(n)\right\}\nonumber\\
     & \quad 
     + p_\gamma  +  \left(1-\frac{\gamma^{n-2}\arcsin (\gamma)}{n-1}\right)^{|{\mathcal V}|},
\end{align}
where $p_\gamma = \P_{\*U}\left(\*W \in E^C_{\tau_0}( \gamma, \widetilde \gamma_1)\right) \rightarrow 0$ as $\gamma \rightarrow 0.$

 Finally, 
 $C^{II}_{\gamma_1,\gamma, \tau} = \left(\frac{1-\gamma_1^2}{5\gamma}C_{\min}\right)^2 = \frac{1}{25\gamma} C^2_{\min},$  so the bound for $\lambda$ in \eqref{eq:lambda_upper} is reduced to   $\lambda \leq \frac{\log\{C^2_{\min}(1- \gamma^2_1)/2 + 1\}}{2 \tau_0 \log(n)/n} = \frac{\log\{0.5C^2_{\min}\gamma^{1/2} + 1\}}{2 \tau_0 \log(n)/n} .$
Moreover, because $\arcsin x \geq x$, we have
\begin{align*}
    1-\frac{\gamma^{n-2}\arcsin (\gamma)}{n-1} \leq 1-\frac{\gamma^{n-1}}{n-1}.
\end{align*}
Then the result of the lemma  
follows immediately from the above and \eqref{eq:finite_p_bound_C_D}.

\end{proof}

 \begin{proof}[Proof of Theorem~\ref{thm:bound_of_C_d_mixture}]
Because the first three terms of \eqref{eq:prob_bound_C_D_simple-1} converges to 0 as $\gamma$ goes to 0, for any $\delta >0$, there exist a $\gamma_{\delta}$ that the probability bound in  \eqref{eq:prob_bound_C_D_simple-1} reduces to 
\begin{align}
\label{eq: prob_bound_candidate_2}
    \P_{\*U, {\mathcal V}}\{(\tau_0,\*M_0)\not\in \widehat{\it \Upsilon}_{{\mathcal V}}(\*Y)\} \leq \delta + \left[1- \P\left\{\mathcal \rho(\*U, \*U^*) \leq \gamma_{\delta}\right\}\right]^{|{\mathcal V}|}  \leq\delta + \left(1-\frac{\gamma_\delta^{n-1}}{n-1}\right)^{|{\mathcal V}|}.
\end{align} 
Moreover, the range for $\lambda$ in Lemma~\ref{lem:bound_single_ustar} reduces to $\Lambda_{\delta} =[\frac{\gamma_{\delta}^{1.5}}{\log(n)/n},\frac{\log\{0.5C^2_{\min}\gamma_{\delta} + 1\}}{2 \tau_0 \log(n)/n}].$  
Here the interval $\Lambda_{\delta}$ is of positive length because $x^{1.5} \leq \frac{1}{\tau_0}\log(cx+1)$ in a neighbourhood of 0 for any constant $c>0.$

 Finally, let $  P_\delta = \left(1-\frac{\gamma_\delta^{n-1}}{n-1}\right)$. By \eqref{eq: prob_bound_candidate_2} and Markov Inequality,
 \begin{align*}
      &   \P_{{\mathcal V}}\left[\P_{\*U| {\mathcal V}}\{(\tau_0,\*M_0)\not\in \widehat{\it \Upsilon}_{{\mathcal V}}(\*Y)\} -\delta \geq (1-    P_{\delta})^{|{\mathcal V}|/2}\right]\\
      & \leq  \frac{E_{{\mathcal V}}\P_{\*U|{\mathcal V}}\{(\tau_0,\*M_0)\not\in \widehat{\it \Upsilon}_{{\mathcal V}}(\*Y)\}}{(1-  P_{\delta })^{|{\mathcal V}|/2}} \\
      & =  \frac{\P_{\*U,{\mathcal V}}\{(\tau_0,\*M_0)\not\in \widehat{\it \Upsilon}_{{\mathcal V}}(\*Y)\}}{(1-  P_{\delta })^{|{\mathcal V}|/2}} 
      \leq {(1-  P_{\delta })^{|{\mathcal V}|/2}},
 \end{align*}
from which (a) follows by making $c<-\frac{1}{2}\log(1-   P_\delta).$ Then (b) follows from (a) and \eqref{eq:coverage_inequality}. 
\end{proof}

\begin{proof}[Proof of Theorem~\ref{the:upper_bound}]
With slight abuse of notation, let $\tau^*$ be defined as \eqref{eq:modefied_BIC} except we replace $\*u^*$ with a random $\*U^*,$ and we first bound $\P(\tau^* \geq \tau ).$

By \ref{eq:A3_inequality1a},   
\begin{align*}
     \P(\tau^* \geq \tau) \leq \P(1  - e^{2\lambda(\tau_0-\tau)\log(n)/n} \{\|(\*I - \*H_{\*M_0, \*U^*}) \*W\|^2+1\} \leq 0).
\end{align*} 
Make $(\tau - \tau_0)/n = c,$ then for a $\tau > \tau_0$, the above reduces to 
\begin{align*}
  \|(\*I - \*H_{\*M_0, \*U^*}) \*W\|^2 \geq  n^{2\lambda c } -1 .
\end{align*}
Then when {$0  < \sigma_{\max} < \infty,$} it follows from \cite{vu2015random} that there exists constant $a$ and $a',$ such that 
\begin{align*}
     \P\left\{ \|(\*I - \*H_{\*M_0, \*U^*}) \*W\|^2 \geq n^{2\lambda c} - 1  \right\}  \leq a \exp\left\{-a' \left(n^{2\lambda c} -1  - (n - \tau_0 -1)\right) /\sqrt{n - \tau_0 -1}\right\}\\
      = a \exp\left\{-a' \left(n^{2\lambda c} - n + \tau_0\right) /\sqrt{n - \tau_0 -1}\right\}.
\end{align*}
When $\lambda  = O(n/\log(n)),$ as in Lemma~\ref{lem:bound_single_ustar}, it then follows that the probability bound 
\begin{align*}
     \P(\tau^* \geq \tau)  \leq a \exp\{-a'(n^{O((\tau-\tau_0)/\log(n)) -1/2} - n^{1/2})\},
\end{align*} 
from which the result in Theorem~\ref{the:upper_bound} 
follows immediately.

\end{proof}

\subsection{Remark on the use of (\ref{eq:modefied_BIC}).} \label{sec:BIC}
Suppose we use a BIC-based penalized criterion, so the corresponding equation (\ref{eq:tau0}) of Section~\ref{sec:candidate_general} for 
the Gaussian mixture model in Section~\ref{sec:mixture}
is 
\begin{align}\label{eq:B3-tau0}
\resizebox{.92\hsize}{!}{$(\tau_0, \*M_0) =\arg\min\limits_{(\tau, \*M_\tau)} \min\limits_{(\*\mu_\tau, \*\sigma_\tau)}  \left\{n\log \bigg(\frac{\|\*y_{obs} - {\*M}_{\tau} {\*\mu}_{\tau} - \diag(\*M_{\tau}  {\*\sigma_\tau}) \*u^{rel}\|^2+1}{n}\bigg) + 2\lambda\tau\log(n)\right\}.$}
\end{align}
{The above recovers $\tau_0$ with a small enough $\lambda$ when $\tau_0 = O(\log(n)) < n.$}
Directly applying \eqref{eq:tau-star} then leads to the following objective function
\begin{align}\label{eq:B3-star}
\resizebox{.92\hsize}{!}{$(\tau^*, \*M^*_{\tau^*}) =\arg\min\limits_{(\tau, \*M_\tau)} \min\limits_{(\*\mu_\tau, \*\sigma_\tau)}  \left\{n\log \bigg(\frac{\|\*y_{obs} - {\*M}_{\tau} {\*\mu}_{\tau} - \diag(\*M_{\tau}  {\*\sigma_\tau}) \*u^*\|^2+1}{n}\bigg) + 2\lambda\tau\log(n)\right\}.$}
\end{align}
If we use this objective function (\ref{eq:B3-star}), 
we can also recover $(\tau_0, \*M_0)$ when $\*u^*$ approximates $\*u^{rel}$;  however, it requires us to  estimate each variance component $\*\sigma_\tau.$To avoid 
 estimating these individual variance components in  $\*\sigma_\tau$, 
we have instead used \eqref{eq:modefied_BIC}
in our development in the main paper. The use of the criterion  \eqref{eq:modefied_BIC} significantly reduces both computational and theoretical complexities, yet we can still obtain our desired developments both methodologically and theoretically.

Specifically, if we denote $\*w^{rel} = \diag({\bm M}_{0} \*\sigma_0)\*u^{rel}$, $\bar \sigma_0^2 = \|{\bm M}_{0} \*\sigma_0\|^2/n = \frac 1n \sum_{j = 1}^\tau m_{\cdot j}(\sigma_j^{(0)})^2$ and $\bar {\*u}^{rel} = \*w^{rel}/\bar \sigma_0$ with $m_{\cdot j} = \sum_{i =1}^n m_{ij}$ being the number of elements in the $j$th component, then we have $\diag({\bm M}_{0} \*\sigma_0)\*u^{rel} = \bar \sigma_0\bar {\*u}^{rel}$ and 
we 
can rewrite (\ref{eq:B3-tau0}) as
\begin{align}
\label{eq:B3-tau0-mod}
\resizebox{.86\hsize}{!}{$(\tau_0, \*M_0) =\arg\min\limits_{(\tau, \*M_\tau)} \min\limits_{(\*\mu_\tau, \*\sigma_\tau)}  \left\{n\log \bigg(\frac{\|\*y_{obs} - {\*M}_{\tau} {\*\mu}_{\tau} - \bar \sigma \bar{\*u}^{rel}\|^2+1}{n}\bigg) + 2\lambda\tau\log(n)\right\}.$}
\end{align}
Here, the fixed $\bar {\*u}^{rel} = (\bar u_1^{rel}, \ldots, \bar u_n^{rel})^\top$ and $\bar u_i^{rel}$ is a realization from $N(0, (\sigma_{j}^{(0)}/\bar \sigma_0)^2)$ for $i$ in the $j$th component. 
When we use $\*u^*$ from $N(0,I)$ to replace $\bar u_i^{rel}$ in (\ref{eq:B3-tau0-mod}), it leads to the use of the modified criterion in \eqref{eq:modefied_BIC}. 
Although $\*u^*$ is not exactly a repro copy of $\bar u_i^{rel}$, 
this modification does not affect our results on obtaining a desirable candidate set. This is because we can always find $\*u^*$ in the neighborhood (direction) of $\bar u_i^{rel}$ to retain $(\tau_0, \*M_0)$, as we show in the proofs in A-III-1.  
The numerical studies in Section~\ref{sec: numerical} also demonstrate that we can use the modified criterion \eqref{eq:modefied_BIC} to get our desired results. 

The tuning parameter $\lambda$ here plays an analogous role to the user-specified penalty parameter in the commonly used generalized information criteria \citep{konishi1996gic,konishi2008information}, that is controlling the weight of the penalization on model complexities. Under mild conditions on this penalty parameter, such GIC-type criteria are known to be consistent for selecting the true model among a fixed set of candidate models. Similarly, in our setting, we add this additional tuning parameter $\lambda$ to achieve the nominal coverage rate theoretically. {Moreover, the first part of \eqref{eq:modefied_BIC} is modified from the typical BIC in which the lowest bound of the $\log$ term is $\log(1/n)$, not $\log(0) = - \infty$.} Specifically, Theorem~\ref{thm:bound_of_C_d_mixture} explicitly characterizes the  range of \(\lambda\) for which the proposed confidence sets achieve the desired coverage for any given $C_{\min}$. And as long as $C_{\min}$ is not extremely small, the range allows \(\lambda = 1\), which corresponds exactly to the standard BIC penalty. For this reason, and to avoid introducing additional tuning complexity, we set \(\lambda = 1\) in all empirical implementations reported in the paper.

\subsection{\hspace{-3mm}Inference for cluster mean $\mu_k^{(0)}$ and standard deviation $\sigma_k^{(0)}$ }\label{sec:mu-sigma}

The repro samples method also offers a new approach to make inference for the 
location and scale parameters, without requiring the knowledge $\tau_0.$ 
To the best of our knowledge, this inference problem is not fully addressed even when $\tau_0$ is known, let alone the situation when $\tau_0$ is unknown. Here, 
we face several 
challenges. 
First, the lengths of the parameters ${\*\mu}_\tau$ and $\*\sigma_\tau$ vary with $\tau$. When $\tau \not = \tau_0$, the elements of the vectors  ${\*\mu}_\tau$ and ${\*\sigma}_\tau$ often do not necessarily lineup with those of the target parameters ${\*\mu}_0$ and $\*\sigma_0$ at all. 
Thus, to make an inference for $\*\mu_{0}$ or $\*\sigma_{0}$ when $\tau_0$ is unknown, we need to consider the task jointly with the unknown discrete parameter $\tau_0$. 
Furthermore, even if $\tau_0$ is given, the membership assignment matrix $\*M_0$ is a collection of a large number of discrete nuisance parameters, and this matrix often cannot be fully well estimated, particularly for those at the overlapping ends of neighboring components. 
Here,
we provide an inference method under the repro samples framework to address these problems. 
To simplify our presentation, our development focuses only on making inference 
for each mixture component's mean $\mu_k^{(0)}$ and standard deviation $\sigma_k^{(0)}$. Or, more specifically, we make a joint inference for each $(\tau_0, \mu_k^{(0)})$ and   $(\tau_0, \sigma_k^{(0)})$, for $k=1 \dots \tau_0$.  Other inference problems concerning the location and scale parameters $\*\mu_0$ and $\*\sigma_0$ can be handled in a  similar fashion.

We follow the three-step procedure of Section~\ref{sec:ThreeSteps} to tackle this problem. Since the first step is already done in Sections~\ref{sec:components} and \ref{sec:finding_candidates}, we directly go to~Steps~2 and~3.~Specifically, for each unique $(\tau^*, \*M^*) \in \widehat{\it \Upsilon}_{{\mathcal V}}(\*y_{obs})$, we first obtain a level $1 - \alpha$ representative interval~for $\mu_k$ and $\sigma_k$, say, $(a_k^{M^*}, b_k^{M^*})$ and $(q_k^{M^*}, r_k^{M^*})$, for $k = 1, $ $\ldots, \tau^*$, respectively. 
We then~follow the discussion in Section~\ref{sec:ThreeSteps} to construct a confidence~set~for~$(\tau_0, \mu_k^{(0)})$~and~$(\tau_0, \sigma_k^{(0)})$,~respectively: 
\begin{equation}
\label{eq:cross_mix} 
\resizebox{.929\hsize}{!}{$\Xi^{(\mu_k)}(\*y_{obs}) = \bigcup\nolimits_{\tau^* \in \widehat{\it \Upsilon}^{(\tau)}_{{\mathcal V}}(\*y_{obs})}\{\tau^*\} \times (a^*_k, b^*_k) \hbox{ and }
\Xi^{(\sigma_k)}(\*y_{obs}) = \bigcup\nolimits_{\tau^* \in \widehat{\it \Upsilon}^{(\tau)}_{{\mathcal V}}(\*y_{obs})}\{\tau^*\} \times (q^*_k, r^*_k),$}
\end{equation}
where $(a_k^{*}, b_k^{*})$=$\big( \inf_{\{\*M^*: (\tau^*, \*M^*) 
 \in \widehat{\it \Upsilon}_{{\mathcal V}}(\*y_{obs}) \}}  a_k^{M^*},\allowbreak
\sup_{\{\*M^*: (\tau^*, \*M^*) 
\in \widehat{\it \Upsilon}_{{\mathcal V}}(\*y_{obs}) \}}b_k^{M^*}\big)$, $(q_k^{*}, r_k^{*}) = $ $\big(\inf_{\{\*M^*:}$ $_{(\tau^*, \*M^*) 
 \in \widehat{\it \Upsilon}_{{\mathcal V}}(\*y_{obs}) \}}  q_k^{M^*},  \, \sup_{\{\*M^*: (\tau^*, \*M^*) 
 \in \widehat{\it \Upsilon}_{{\mathcal V}}(\*y_{obs}) \}}r_k^{M^*}\big)$,
 and $\widehat{\it \Upsilon}^{(\tau)}_{{\mathcal V}}(\*y_{obs}) $=$\{\tau^*: (\tau^*, \*M^*) \in \widehat{\it \Upsilon}_{{\mathcal V}}(\*y_{obs})$ $ \mbox{for some }  \*M^*\}$ is the collection of unique $\tau^*$'s in $\widehat{\it \Upsilon}_{{\mathcal V}}(\*y_{obs}).$ 
 By Corollary~\ref{cor:3steps} in Section~\ref{sec:ThreeSteps}, if  
 $(a_k^{*}, b_k^{*})$ is a valid level $1 - \alpha$ confidence set for $\mu_k^{(0)}$ when $\tau^*=\tau_0$, 
then $\Xi^{(\mu_k)}(\*y_{obs})$~is a level $1 - \alpha$ confidence set for $(\tau_0,\mu_k^{(0)} ).$ Similarly,  $\Xi^{(\sigma_k)}(\*y_{obs})$ is a level $1 - \alpha$ confidence set for~$(\tau^0, \sigma_k^{(0)})$. 

The remaining task is to obtain the representative sets $(a_k^{M^*}, b_k^{M^*})$ and $(q_k^{M^*}, r_k^{M^*})$, given $(\tau^*, \*M^*) \in \widehat{\it \Upsilon}_{{\mathcal V}}(\*y_{obs}).$ If $\tau^* = \tau_0$ and one set of these $(a_k^{M^*}, b_k^{M^*})$ and $(q_k^{M^*}, r_k^{M^*)})$ cover $\mu_k^{(0)}$ and $\sigma_k^{(0)}$ with probability $\alpha,$ 
then $(a^*_k, b^*_k)$ and $(q_k^{*}, r_k^{*})$ are valid confidence sets for $\mu_k^{(0)}$ and $\sigma_k^{(0)}$ respectively.
By Theorem~\ref{thm:p_bound_C_D_general}, with enough computing power allowing $|{\mathcal V}| \to \infty$, we can always capture $(\tau_0, \*M_0)$ in the candidate set $\widehat{\it \Upsilon}_{{\mathcal V}}(\*y_{obs})$ thus $(a^*_k, b^*_k)$ and $(q_k^{*}, r_k^{*})$ are valid confidence sets. In practice with limited computing resources,  we can well recover $\tau_0$ in the candidate set $\widehat{\it \Upsilon}_{{\mathcal V}}(\*y_{obs})$ (as evidenced in our empirical study in Section~\ref{sec: numerical}) but it might be hard to exactly recover $\*M_0$.  
To counter this issue, 
 we propose a robust version of $(a_k^{M^*}, b_k^{M^*})$ and $(q_k^{M^*}, r_k^{M^*})$  that only requires $\*M^*$ to be reasonably close to $\*M_0$ for $(a_k^{M^*}, b_k^{M^*})$ 
to cover $\mu_{k}^{(0)}$ and  $(q_k^{M^*}, r_k^{M^*})$ to cover $\sigma_{k}^{(0)}$ with the desirable rate. 
There are many $\*M^*$'s that are close to $\*M_0$, so in practice we can use a  far smaller $|\mathcal V|$ to achieve the desirable coverage~rate.

 Specifically, for a given $(\tau^*, \*M^*) \in \widehat{\it \Upsilon}_{{\mathcal V}}(\*y_{obs})$, let ${\mathcal D}_k^* = \{y_i: m^*_{ik}=1, i=1, \dots, n_k\}$ be the data points assigned to the $k$ component, $1 \leq k \leq \tau^*$. 
We may use the classical sign test for the population median to get a representative set for the $k$th component's 
location parameter $\mu_k$. It happens to be equivalent to a special case of the repro sample approach developed for nonparametric quantile inference in {Examples~\ref{ex:crq} and \ref{ex:crq_robust}
 of Appendix~\ref{sec:example_sec2}}:  
 $
 (a_k^{\*M^*}, b_k^{\*M^*}) =  \big\{\mu: a_L(.5) \leq \sum\nolimits_{y_i \in {\mathcal D}_k^*} I(y_i \leq \mu ) \leq a_U(.5) \big\}. $ 
 Here, $a_L(.5)$ and $a_U(.5)$ are level $1 - \alpha$ lower and upper bounds defined in Example~\ref{ex:bin} with $r$ $ = |{\mathcal D}_k^*|$ and the quantile $\zeta = \frac12$. 
See {Appendix~\ref{sec:bin_example}} for further details.

For the scale parameter $\sigma_k$, under the assumption  
$y_i \in {\mathcal D}_k^*$ are from $N(\mu_k, \sigma_k^2)$,~a~consistent robust estimator is $\hat \sigma_k =  {\it MAD}/\Psi^{-1}(.5)$, where the median absolute deviation~${\it MAD} =$ $ {\it med}\{|y_i - M_k|: y_i \in {\mathcal D}_k^*\}$, $M_k = {\it med}\{y_i: y_i \in {\mathcal D}_k^*\}$ and $\Psi(t) = P(|Z_j - M_z| < t)$~is~the cumulative distribution function of $|Z_j - M_z|$, $j = 1, \ldots, |{\mathcal D}_k^*|$. Here,  
$Z_j \overset{iid}{\sim} N(0,1)$ and $M_z =$ ${\it med}
\{Z_1,  \ldots, Z_{|{\mathcal D}_k^*|}\}$.
In the form of generalized generative equation (\ref{eq:AA}), we~have
$$ \sum\nolimits_{y_i \in {\mathcal D}_k^*} I\{|y_i - M_k| \leq {\small \Psi^{-1}(.5)} \sigma \} - 
\sum\nolimits_{j = 1}^{|{\mathcal D}_k^*|} U'_j = 0, 
$$
where $U'_j = I\{|Z_j - M_z| \leq {\small \Psi^{-1}(.5)}\}$ are dependent Bernoulli$(\frac12)$ variables. Since~the~distribution of $\sum\nolimits_{j = 1}^{|{\mathcal D}_k^*|} U'_j$ can be simulated, we can get 
a Borel interval  $B_{1 -\alpha} = [a'_L, a'_U]$,  
with
$(a_L',$ $ a_U') = 
     \arg \min_{\{(l,r): 
  P( l < \sum\nolimits_{j = 1}^{|{\mathcal D}_k^*|} U'_j < r )
     \geq 1- \alpha\}} |r - l|$. Then, given $(\tau^*, \*M^*)$, a  representative~set~of~$\sigma_k$~is: 
    $$(q_k^{\*M^*}, r_k^{\*M^*}) =  \big\{\sigma_k: a'_L \leq \sum\nolimits_{y_i \in {\mathcal D}_k^*} I(|y_i - M_k| \leq {\small \Psi^{-1}(.5))} \sigma_k ) \leq a'_U \big\}.$$ 
 Further details and discussions are provided in  
{Appendix~\ref{sec:example_sec2}.}

\section{Additional implementation details and numerical results in Section~\ref{sec: numerical}}\label{sec:numerical-details}

\subsection{Existing frequentist and Bayesian approaches: real data analysis and simulation studies} \label{sec:bayes_comp}

\subsubsection{Analysis of SLC data: existing frequentist and Bayesian approaches}

The SLC data set consists of red blood cell sodium-lithium countertransport (SLC) activity measurements from 190 individuals \citep{dudley_assessing_1991}. 
The SLC measurement is known to be correlated with blood pressure, and  is considered as an essential cause of hypertension by some researchers. 
Moreover, SLC is usually easier to study than blood pressure that could be influenced by numerous environmental and genetic factors  \citep{dudley_assessing_1991}.
The SLC data have been analyzed using the classical hypothesis testing approach and Gaussian mixture models by \citep{chen_inference_2012, roeder_graphical_1994}, both of which focus on making inference for the unknown number of components $\tau_0$. In this subsection, we summarize their conclusions, and in addition provide two point estimates using Bayesian information criterion (BIC) and  Akaike information criterion (AIC),  respectively, and also provide a set of Bayesian inference results for $\tau_0$.

{\it Point estimates of $\tau_0$:}  We use R function {\it GMM} in {\it CluterR} package to analyze the SLC data and obtained two point estimates: $\hat \tau_{BIC} = 2$ and $\hat \tau_{AIC}  = 7$, with the first one using the BIC and the second using AIC criterion.
These point estimates do not provide a quantification of the estimation uncertainty.

{\it Results from classical hypothesis testing methods:}  
\cite{roeder_graphical_1994} and \cite{chen_inference_2012} analyzed the same SLC data set 
of 190 individuals using a Gaussian mixture model 
under the classical Neyman-Pearson framework. Both of their problem setups are to test the unknown number of components $\tau_0$ with the hypothesis $H_0: \tau_0 = k$ vs $H_1: \tau_0>k$ for a small integer $k$; and both used modified penalized likelihood ratio testing methods that are justified by large-sample theories. 
\cite{roeder_graphical_1994} concluded that $\tau_0=3$ is the smallest value of $\tau_0$ not rejected by the data with a further assumption that the variance of each component is 
the same.  
\cite{chen_inference_2012} dropped the assumption of equal variances and 
concluded that the hypothesis of $\tau_0=2$ could not be rejected, 
and thus $\tau_0=2$ a good fit for the data. By inverting the two tests in these two papers, we get two one-sided confidence sets $[3, \infty)$ and  $[2, \infty)$. Without further developments, the existing approaches of \cite{roeder_graphical_1994} and \cite{chen_inference_2012} cannot provide two-sided confidence intervals for the smallest possible $\tau_0$.  

{\it Bayesian analysis results:} Bayesian procedures can also be used to analyze the unknown number of components $\tau_0$ in a Gaussian mixture model; see, e.g.,   \cite{richardson_bayesian_1997}. 
We conduct a Bayesian analysis of the SLC data to obtain credible sets for $\tau_0$, using the method proposed in \cite{richardson_bayesian_1997} and based on R package {\it mixAK} \citep{komarek2014capabilities}. 
In our analysis, we have tried four different priors on $\tau$: uniform, Poisson(0.4),  Poisson(1) and Poisson(5), all of which are truncated outside the set $\{1, $ $\dots, 10\}$. The priors used on $(\*\mu_\tau,\*\sigma_\tau)$, given $\tau$, are the default priors (i.e., Gaussian and inverse Gamma priors) in {\it mixAK}. 
The posterior distributions for the four priors are plotted in Figure~\ref{fig:SLCposteriors}, corresponding to which the 95\% credible sets are $\{2,3\}$, $\{2\}$, $\{2,3\}$, $\{2, $ $3,4\}$, respectively. 
It turns out 
the four posterior distributions and the four 95\% credible sets of $\tau$ are quite different for different priors of $\tau$. Even in the two cases (i.e., the cases of the uniform and Poisson(1) priors) that produce the same $95\%$ credible set, 
their posterior distributions are quite different. 
It is clear that the
Bayesian inference for $\tau$ is very sensitive to the choice of its prior. A simulation study next in Appendix~\ref{sec:App_Simu} affirms this observation. 
The simulation study further suggests that the Bayesian procedure is not suited for the task of recovering the true $\tau_0$ (under repeated experiments), and its outcomes are affected greatly by the default priors on $\*\mu$~and~$\*\sigma$ as~well.  
\begin{figure}[ht]
    \centering
\includegraphics[width=\textwidth, height= 8cm]{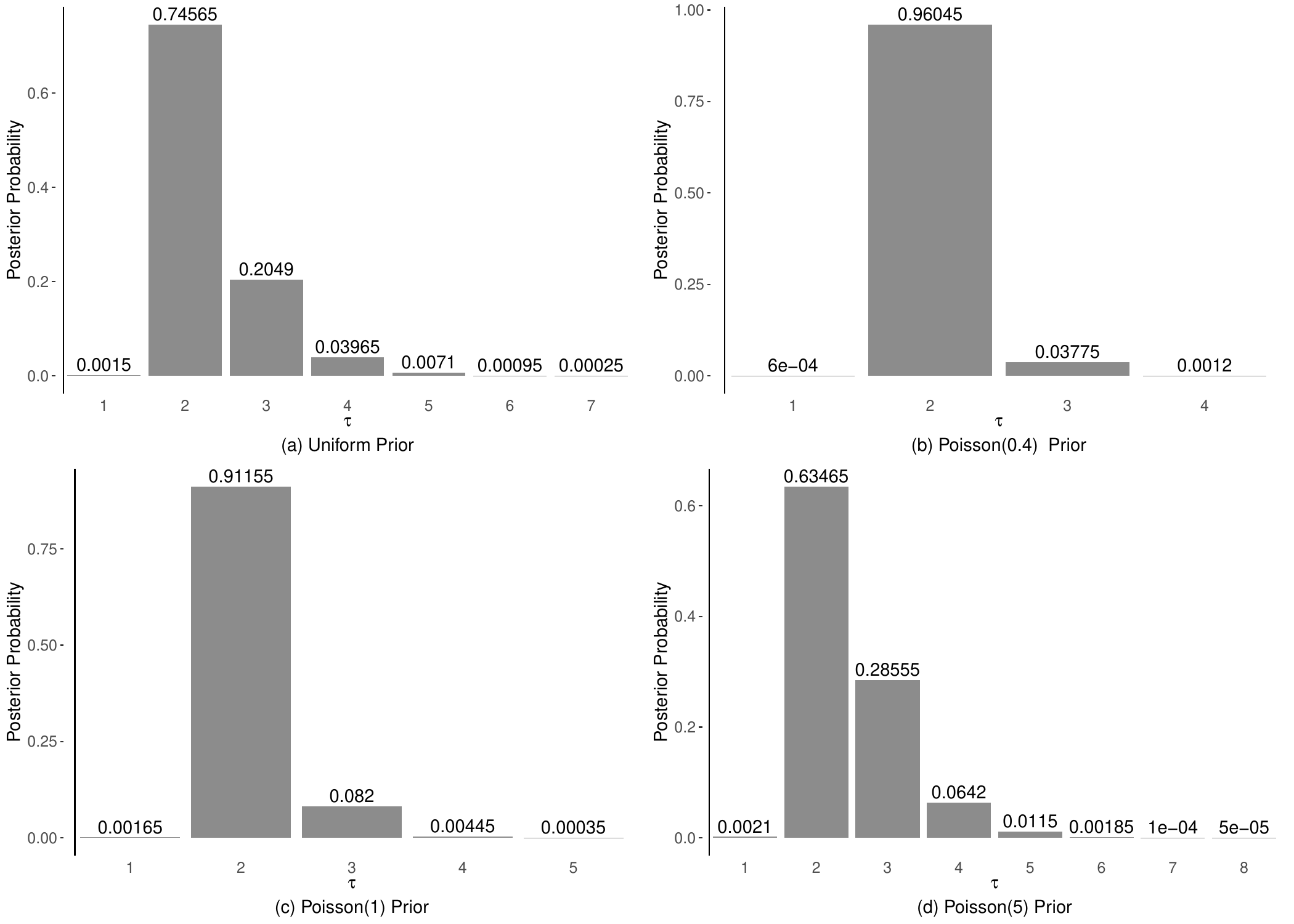}
    \caption{\small Posterior distributions of $\tau$ under four different priors of $\tau$: (a) uniform; (b) Poisson($\mu = 0.4$); (c) Poisson($\mu = 1$); (d) Poisson($\mu = 5$); All priors are truncated outside the set $\{\tau: \tau = 1, \dots, 10\}$.
    }
\label{fig:SLCposteriors}
\end{figure}

\subsubsection{Simulation studies: performance of existing frequentist and Bayesian approaches}\label{sec:App_Simu}

In this simulation study, we 
use the same $2 \times 200$ simulated data sets used in 
Section~\ref{sec: numerical}.

{\it BIC point estimate of $\tau_0$:} The numerical evidence in our simulation studies indicates the inferior performance of using a point estimate to estimate $\tau_0$:
The standard BIC point estimator recovers the true value of $\tau_0 = 3$ only 13 times (6.5\%) out of the 200 repetitions.  
The number of times 
drops to $0$ (out of 200 repetitions) when the true $\tau_0=4.$ The results are not surprising, since with a strong penalty term the BIC point estimator is often biased towards smaller $\tau$ values in practice, even though the BIC estimator is known to be a consistent estimator. 
These two extremely low rates of correctly estimating the true $\tau_0$ highlight the challenge of the traditional point estimation method in our simulation examples, and more generally the risk of using just a point estimation in a Gaussian mixture model in general. Indeed, the result shows that it is critically important to quantify the estimation uncertainty through a confidence set.

{\it Classical hypothesis testing approaches:}
Our simulation studies can also illustrate the advantage of our proposed confidence set over the existing one-sided hypothesis testing approaches in \cite{roeder_graphical_1994} and \cite{chen_inference_2012},  both of which 
do not provide an upper bound for $\tau_0$ and 
advocate for adopting 
the use of the smallest $\tau_0$ that can not be rejected. 
However, we observe in our study that such a practice could very well mislead us to an overly small model as well, resulting in possibly incorrect scientific interpretation. In particular, we find that, when the true $\tau_0=3$ and we perform the PLR test of \cite{chen_inference_2012} on $H_0: \tau_0=2$ versus 
$H_1: \tau_0 = 3$,
only 74\% times we fail to reject $H_0: \tau_0=2$. When the true $\tau_0=4$ and we perform the PLR test on $H_0: \tau_0=3$ versus $H_1: \tau_0 = 4$,
we fail 93\% times to reject $H_0: \tau_0=3$. These results suggest that most of the time, the PLR test prefers a wrong $\tau$ over the true one. Furthermore, inverting either of the one-sided tests does not produce an as useful confidence set, since the upper bound of the set is $\infty$. Comparing with the results in Section~\ref{sec:tau-numerical}, 
we can see that the repro samples method is a more intact approach to conducting inference on $\tau_0$ compared to the one that uses only a testing method under the classical Neyman-Pearson hypothesis test framework.

{\it Bayesian method:}
We have also studied 
the empirical performance of the Bayesian procedure described in \cite{richardson_bayesian_1997} 
when the number of components is unknown. Using the same 200 data sets that have produced 
Figure~\ref{fig:simulation_mixture} with $\tau_0=4$,   
Figure~\ref{fig:bayes_simulation} summarizes the results from the Bayesian procedure by R package {\it mixAK}. Here, we display the credible sets produced using (a) uniform prior on $\tau = 1, \dots, 10$; (b) Poisson(1) prior,  which prefers smaller $\tau$; and (c) Poisson(5) prior, whose mode is $4$ and matches with the true $\tau_0$.
We see that the Bayesian level-$95\%$ credible sets consistently underestimate and fail to cover $\tau_0$, with the coverage rate equal to 21.5\%, $0\%$ and $50.5\%$, respectively.
In addition, the outcomes of the Bayesian procedure are sensitive to the specification of the prior distribution of $\tau,$ as evidenced by the drastic differences among the three panels of Figure~\ref{fig:bayes_simulation}.
At first glance, the underestimation of $\tau_0$ observed in Figure~\ref{fig:bayes_simulation}(a) with the uniform prior is somewhat surprising. However, further investigation reveals that the priors on $\*\mu$ and~$\*\sigma$ also have a significant impact on estimating $\tau$, and 
the underestimation is attributable to the shrinkage effect of implementing a multilevel hierarchical model  \citep{richardson_bayesian_1997}. Indeed, it is well known that shrinkage is ubiquitous when parameters are modeled hierarchically. 
Here both the means and variances of the components would share a common prior, thus over-promoting smaller $\tau$ as observed across all three plots in Figure~\ref{fig:bayes_simulation}.

\begin{figure}[ht]
    \centering    \includegraphics[width=\textwidth, height= 5cm]{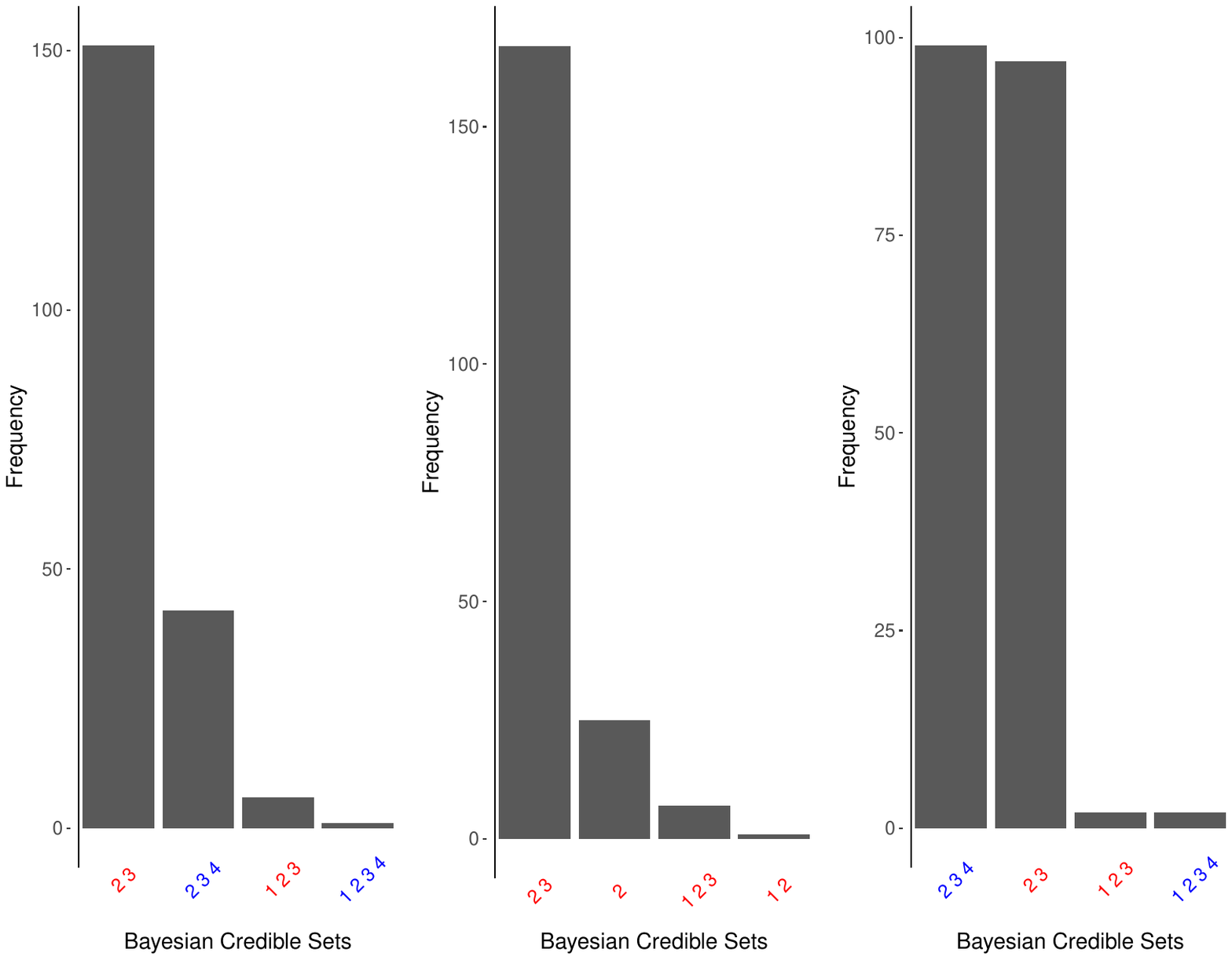}
    (a)\hspace{4.8cm}  (b)  \hspace{4.8cm} (c)
    \caption{Bar plots of the 200  95\% Bayesian credible sets obtained in the simulation study (200 repetitions) with $\tau_0 = 4$. Prior choices are (a) uniform; (b) Poisson($\mu = 1$); (c) Poisson($\mu = 5$); All priors truncated over $\{\tau=1, \dots, 10\}$.}
    \label{fig:bayes_simulation}
\end{figure}

\subsection{Numerical studies on inference for $\mu_j^{(0)}$ and $\sigma_j^{(0)}$ using the repro samples approach}
\label{sec:simu_location_scale}

\subsubsection{Analysis of SLC data}
We further provide  inference for the location and scale parameters of the mixture model, 
$\mu_k^{(0)}$ and $\sigma_k^{(0)}$, for $k = 1, \ldots, \tau_0$, using the approach proposed in Section~\ref{sec:mu-sigma}.
Table~\ref{tab:SLC}(a) reports the corresponding estimated mean, standard deviation and weight of each component.  
Table~\ref{tab:SLC}(b) reports the level-$95\%$ representative intervals $(a_k^*, b_k^*)$ and $(q_k^*, r_k^*)$ of 
$\mu_k$ and $\sigma_k$, 
for $1 \leq k \leq \tau$ and $\tau = 2, 3, 4$, respectively; See Sections~\ref{sec:ThreeSteps} and~\ref{sec:mu-sigma} for the definition of representative intervals. 
If $\tau_0$ is indeed one of the $\tau = 2, 3$ or $4$, then
the representative intervals $(a_k^*, b_k^*)$ and $(q_k^*, r_k^*)$ of the corresponding row are level-$95\%$ confidence intervals for
$\mu_k^{(0)}$ and $\sigma_k^{(0)}$, 
for $1 \leq k \leq \tau_0$, respectively.
By (\ref{eq:cross_mix}),  
our theoretically guaranteed $90\%$ confidence sets of  $(\tau_0,\mu_k^{(0)})$ and $(\tau_0,\sigma_k^{(0)})$, for a given  $1 \leq k \leq \tau_0$, are,  respectively, 
$$
\bigcup\nolimits_{\tau \in \{2,3,4\}} \{\tau\} \times (a_k^*, b_k^*) \quad \hbox{and} \quad \bigcup\nolimits_{\tau \in \{2,3,4\}} \{\tau\} \times (q_k^*, r_k^*),
$$ where $(a_k^*, b_k^*)$ and $(q_k^*, r_k^*)$ are the corresponding $\mu_k^{(0)}$ and $\sigma_k^{(0)}$ intervals reported in Table~\ref{tab:SLC}(b). 
This set of location and scale confidence sets 
have accommodated the uncertainty associated with the unknown $\tau_0$. 
Together, our results provide a full picture of the uncertainty in estimating parameters, across both the discrete parameter of the unknown number of components and the continuous location and scale parameters.   

 \begin{table}[!t]
 \centering
 \resizebox{\textwidth}{!}{\begin{tabular}{llll}
 \hline\hline
 {\bf (a)} & Estimate of location $\mu_k^{(0)}$, $k = 1, \ldots, \tau$ & Estimate of scale $\sigma_k^{(0)}$, $k = 1, \ldots, \tau$ & Estimate of weight $p_k^{(0)}$, $k = 1, \ldots, \tau$ \\ 
   \hline
 $\tau =2$  & 0.2206, 0.3654 & 0.0571, 0.1012 & 0.7057, 0.2943 \\ 
   $\tau =3$ & 0.1887, 0.2809, 0.4199 & 0.0414, 0.0474, 0.0886 & 0.4453, 0.3866, 0.168 \\ 
   $\tau = 4$  & 0.1804, 0.2556, 0.3351, 0.4403 & 0.0362, 0.0268, 0.0359, 0.086 & 0.4018, 0.2941, 0.1742, 0.1299 \\ 
    \hline
\end{tabular}}
 \centering
 \resizebox{\textwidth}{!}{\begin{tabular}{lll}
 \hline
  {\bf (b) }& Confidence/Representative Intervals of $\mu_k^{(0)}$, $k = 1, \ldots, \tau$ & Confidence/Representative Intervals of $\sigma_k^{(0)}$, $k = 1, \ldots, \tau$  \\ 
   \hline
1 & (0.233, 0.264) & (0.077, 0.104) \\ 
   $\tau =2$ & (0.212, 0.258),  (0.379, 0.495) & (0.055, 0.087), (0.019, 0.114) \\ 
   $\tau = 3$ & (0.193, 0.245), (0.219, 0.439), (0.398, 0.463) & (0.039, 0.075), (0.013, 0.321), (0.025, 0.078)\\
  $\tau =4$ & (0.168, 0.188), (0.212, 0.242), (0.321, 0.349), (0.404, 0.508) & (0.021, 0.062), (0.035, 0.053), (0.024, 0.063), (0.017, 0.161)\\
      \hline\hline
 \end{tabular}}
 \captionsetup{font= footnotesize}
 \caption{ Analysis of SLC data: (a) Estimated mean, standard deviation and weight of each component; (b) Corresponding level-$95\%$ representative intervals for $\mu_k$ and $\sigma_k$, $k = 1, \ldots, \tau$. 
 If
 $\tau_0$ is one of those in the level-$95\%$  confidence set of $\tau_0$ $\{2, 3, 4\}$, then the representative intervals in 
 the corresponding row are level-$95\%$ confidence intervals for
 $\mu_k^{(0)}$ and $\sigma_k^{(0)}$, 
 for~$1 \leq k \leq \tau_0$,~respectively.
 }
 \label{tab:SLC}
 \end{table}

\subsubsection{Simulation Results}
We
also evaluate the performance of the confidence sets proposed  for location and scale parameters $\mu_j^{(0)}$ and $\sigma_j^{(0)}$, $j = 1, \ldots, \tau_0$, as described 
in Section~\ref{sec:mu-sigma}. 
Table~\ref{tab:Inference_mu_sigma} reports the coverage rates and average lengths of the confidence sets, from $200$ repetitions, 
where the values in the brackets are the corresponding standard errors.   
For $\tau_0 =3,$ the proposed confidence sets have high coverage rates across the board for both $\mu_j^{(0)}$'s and $\sigma_j^{(0)}$'s. For $\tau_0 = 4$, the coverage rates for the scale parameters $\sigma_j^{(0)}$'s still meet or exceed the 95\% confidence level 
for three of all four clusters, with the coverage rate for the one remaining slightly low at 92.5\%.  As for location parameters $\mu_j^{(0)}$'s, two  clusters are around $95\%$ and the other two slightly undercover at 91.0\% and 92.5\%, respectively. The under coverage might be because more overlapping between the clusters has complicated our computations. 
Another contributing factor might be the smaller sample size of each cluster for $\tau_0 = 4$
(compared to $\tau_0 =3$), which leads to greater variability in the resulting confidence sets.  
Nevertheless, this set of results indicate that the proposed robust confidence sets generally achieve desired coverage rates for both $\mu^{(0)}$'s and $\sigma^{(0)}$'s.
To the best of our knowledge (and also from conversations with experts in the field), this set of empirical coverage rates on means and variances (even if $\tau_0$ is given) are among the best in simulation studies mimicking~the~setup of the SLC~data.

\begin{table}[!t]
\centering
\resizebox{0.8\textwidth}{!}{\begin{tabular}{lr|c|c}\hline \hline
& & Coverage & Average Width (SD) \\ \hline
\multirow{2}{*}{$\tau_0 =3$} & $\mu_j^{(0)}$ & .985, .965, .975 & .136(.004), .202(.006), .330(.004) \\\cline{2-4}
& $\sigma_j^{(0)}$ &  .995, .985, .995  & 
.029(.002), .010(.006), .042(.003) 
\\
 \hline
\multirow{2}{*}{$\tau_0 = 4$} 
& $\mu_j^{(0)}$ &  .910, .926, .940, .955
&  .100(.004), .192(.008), .256(.011), .215(.006) \\\cline{2-4}
& $\sigma_j^{(0)}$ &  .950, .925, .965, .960  & .015(.002), .060(.005), .076(.008), .048(.004)\\ \hline
\hline
\end{tabular}}
 \captionsetup{font= small}
\caption{ Simulation Results: Coverage and average width of $\mu_j^{(0)}$ and  $\sigma_j^{(0)}$, $j = 1, \ldots, \tau_0$, $\tau_0 = 3$ and $4$; SD is for standard deviation (reported in round brackets);  Repetitions is $200$. 
}
\label{tab:Inference_mu_sigma}
\end{table}

\section{Brief review and comparison with existing and relevant inference procedures}\label{sec:BFFcomprison}

The repro samples method is a frequentist approach, but
its development inherits several key ideas from several 
existing inferential procedures
across Bayesian, fiducial, and frequentist (BFF) paradigms \citep{Berger2020}.  
For readers' convenience, we first provide in Appendix~\ref{sec:comp_review} a brief review of several  inferential approaches across BFF paradigms that are relevant to the repro samples method . We then provide in Appendix~\ref{sec:comp_smry} a comparative discussion on  the differences, in which we  also highlight some possible improvements of the proposed repro samples method over the existing approaches. 
Appendix~\ref{sec:comp_additional_example} includes a couple of additional  illustrative examples for some further and specific comparisons.

\subsection{A brief review of existing and 
relevant BFF procedures}\label{sec:comp_review}

The repro samples method  utilizes artificial sample sets for inference. 
We review below several relevant artificial sample-based inferential approaches across the BFF paradigms. Many of these approaches have helped shape the various developments of the repro samples~method. 

\begin{itemize}
[leftmargin=1.5em, labelsep=0em, 
align=parleft,   itemsep=0pt,
  parsep=0pt,
  topsep=1pt]    \item 
{\bf  Approximate Bayesian Computation (ABC)}
refers to a family of techniques used to approximate posterior densities of $\btheta$ by bypassing direct likelihood evaluations \citep[cf.][]{Rubin1984,Tavare1997, csillery2010approximate, 
Peters2012}.  
A basic version of the so-called {\it rejection ABC algorithm} has the following [a] - [c] steps:
\begin{itemize}
[leftmargin=1.5em, labelsep=0em, 
align=parleft,   itemsep=0pt,
  parsep=0pt,
  topsep=1pt] 
\item[] {\it 
   {\rm [Step a]} Simulate an artificial model parameter $\btheta^*$ from a prior distribution $\pi(\btheta)$ and, given the simulated $\btheta^*$, simulate an artificial data set ${\bf y}^*$ 
(e.g., set ${\bf y}^* = G(\btheta^*, {\bf u}^*)$ for a simulated ${\bf u}^*$); 
{\rm [Step b]} If 
$ {\bf y}^* \approx {\bf y}_{obs}$,
 we collect the artificial parameter $\btheta^*$; 
 {\rm [Step c]} Repeat Steps [a] and [b] to obtain a large set of $\btheta^*$.}
 \end{itemize}
ABC is in fact 
a Bayesian inversion method,
since keeping $\btheta^*$ that has the match ${\bf y}^* \approx {\bf y}_{obs}$ is equivalent to keep the $\btheta^*$ that solves $G(\btheta, {\bf u}^*) \approx {\bf y}_{obs}$.

Operationally,
because it is difficult to match ${\bf y}^* \approx {\bf y}_{obs}$, the ABC method instead suggests to match
$|S({\bf y}_{obs}) - S({\bf y}^*)|\leq \epsilon$ for a pre-chosen summary statistic $S(\cdot)$ and tolerance $\epsilon >0$. The kept $\btheta^*$ form  
{\it an ABC posterior},  $p_{\epsilon}(\btheta|\*y_{obs})$.  
If $S(\cdot)$
is a {\it sufficient statistic},  then the ABC posterior $p_{\epsilon}(\btheta | \*y_{obs}) \to p(\btheta | \*y_{obs})$,  
the target posterior distribution,  as $\epsilon \to 0$ at a fast rate \citep{Barber2015,Li2016}.
To improve its computing efficiency, this basic ABC algorithm has been extended to, for instance, more complex ABC-Markov Chain Monte-Carlo (ABC-MCMC) type of algorithms, but the key matching and Bayesian inversion ideas remain. 
Note that, the ABC method seeks to match each single artificial sample copy $\*y^*$ with the observed data $\*y_{obs}$ through a summary statistic and a preset tolerance level $\epsilon$ to form a rejection~rule. 

There are two remaining issues unsolved in the ABC method~\citep[cf.,][]{Li2016,Thornton2018}. First, when $S(\cdot)$
is not sufficient, $p_{\epsilon}(\btheta | \*y_{obs}) \not\to p(\btheta | \*y_{obs})$.  
In this case,
there is no guarantee that an ABC posterior is a Bayesian posterior, thus 
the interpretation of the inferential result by the ABC posterior is unclear.  
However, it is not possible to
derive a sufficient statistic when the likelihood is intractable. Thus, the requirement that $S(\cdot)$ is sufficient places a limitation on the practical use of the ABC method (as a likelihood-free inference approach). 
Second, the ABC method requires the preset threshold value $\epsilon \to 0$ at a fast rate, leading to a degeneracy of computing efficiency. This is because the acceptance probability in Step [b] goes to $0$ for a very small $\epsilon \to 0$. 
In practice, there is no clear-cut choice for an appropriate $\epsilon$ to balance the procedure' computational efficiency and inference validity. It 
remains to be an unsolved question in the literature \cite[e.g.,][]{Li2016}.

\item {\bf  Generalized fiducial inference (GFI)} is a generalization of Fisher's fiducial method, which is understood in contemporary statistics as {an inversion method} to solve a pivot equation for model parameter $\btheta$ \citep{Zabell1992, Hannig2016, Thornton2022}. 
The GFI method extends the inversion approach from a pivot to actual data; i.e., it solves for $\btheta^*$ from
${\bf y}_{obs} = G(\btheta, {\bf u}^*)$ for artificially generated ${\bf u}^*$'s. However, since a solution is not always possible, the GFI method proposes to 
consider an optimization under an $\epsilon$-approximation: 
\begin{equation}
 \btheta_\epsilon^* =  {\rm argmin}_{\theta \in \{\theta: \,\, || {\bf y}_{obs} - G(\theta, {\bf u}^*)||^2 \leq \epsilon \}}  || {\bf y}_{obs} - G({\theta}, {\bf u}^*)||^2
\label{eq:GFI}
\end{equation}
and let $\epsilon \to 0$ at a fast rate~\citep{Hannig2016}.
A GFI method can be
viewed as a constrained optimization method that limits artificial data ${\bf y}^* = G({\theta}, {\bf u}^*)$ within an $\epsilon$-neighborhood of ${\bf y}_{obs}$. 
Eq. (\ref{eq:GFI}) can also be viewed as an inversion operation on the model equation ${\bf y}^* = G({\theta}, {\bf u}^*)$, which we refer to as a {\it Fisher Inversion} in this paper. 
The GFI development  relies on large-sample and the so-called {\it fiducial Berstein-von-Mises theorem} to justify its inference validity \citep{Hannig2016}. So often times it has a guaranteed statistical (frequentist) performance, provided that the sample size $n \to \infty$ and $G$ is smooth/differentiable in $\btheta$~\citep{Hannig2009, Hannig2016}. Different than the ABC method that uses the rejection sampling approach, the GFI uses an optimization to approximately solve the equation $\*y_{obs} = G(\btheta, \*u^*)$ within a $\epsilon$-constrained set in $\Theta$.
Operationally, the 
GFI method has the same computational issue as in the ABC method on the choice of the pre-specified $\epsilon$, since, 
as $\epsilon \to \infty$ at a fast rate, equation (\ref{eq:GFI}) may not have a solution. 

\item[(iii)] {\bf Efron's bootstrap and other related artificial sampling methods}, in which many copies of artificial data are generated, are a popular way to help quantify uncertainty in complicated estimation problems in statistics. 
They are the most successful and broadly-used artificial-sample-based inference approach 
 to date. 
Let $\hat {\*\theta}({\*y}_{obs})$ 
be a point estimator of the true parameter $\btheta_0$ and $\btheta^* = \hat{\btheta}(\*y^*)$ be the corresponding bootstrap estimator calculated using a bootstrap sample $\*y^*$. 
The key validity justification of bootstrap methods relies on the so-called {\it bootstrap central limit theorem} (CLT) \citep[e.g.,][]{Singh1981,Bickel1981}: 
$$\hat{\btheta}(\*Y^*) - \hat {\*\theta}({\*y}_{obs}) | \hat {\*\theta}({\*y}_{obs}) \sim \hat {\*\theta}(\*Y) -\btheta_0 | \btheta_0,$$  
in which the (multinomial distributed) uncertainty in resampled artificial data $\*Y^*$ asymptotically matches the  (often non-multinomial distributed) uncertainty in $\hat \btheta(\*Y)$ inherited from the target random sample $\*Y$. 
When the CLT does not apply (e.g., the parameter space $\Theta$ is a discrete set), the bootstrap methods are not supported by theory and often perform poorly. Through the lenses of confidence distribution, the bootstrap is closely connected with other BFF inference approaches \cite[cf., e.g.,][]{Xie2013, Thornton2022}.

\item[(iv)] {\bf Inferential model (IM)}  is an attempt to develop a general framework for ``prior-free exact  probabilistic inference'' \citep{Martin2015}. As described in \cite{Martin2015} and under the model assumption (\ref{eq:1}), an IM procedure include three steps: 
\begin{itemize}
[leftmargin=1.5em, labelsep=0em, 
align=parleft,   itemsep=0pt,
  parsep=0pt,
  topsep=1pt] 
    \item[] 
{\it {\rm [A-step]} Associate $(\*Y, \btheta)$ with the unobserved auxiliary statistic $\*U$, i.e., $\*Y = G({\btheta}, {\*U})$;
{\rm [P-Step]} Predict the unobserved auxiliary statistic $\*U$ with a random set ${\mathcal S}$ with distribution~${\mathbb P}_{\mathcal S}$;
{\rm [C-step]} Combine observed $\*y_{obs}$ with random set ${\mathcal S}$  into a new data-dependent random set 
$\Theta_{\*y_{obs}}({\mathcal S}) = \bigcup_{\*u \in {\mathcal S}} \left\{\btheta: \*y_{obs} = G(\btheta, \*u)\right\}$.} 
\end{itemize}
A novel aspect of the IM development that is different than a typical fiducial procedure is that the uncertainty quantification of $\*Y$ is through the parameter-free $\*U$ using a random set in [P-step]. This separation removes the impact of parameter $\*\theta$ and makes the task of uncertainty quantification easier. Since ${\mathcal S}$ is a random set, the operation in the third [C-step] is essentially an inversion operation of on the random set ${\mathcal S}$; this operation is referred to as 
{\it Fisher-Dempster}
inversion in this paper. 
In order to fully express the outcomes in a probabilistic form, the IM development needs to use a system of imprecise probability known as {\it Dempster-Shafer Calculus}. 
The outcomes of an IM algorithm are the so-called {\it plausible} and {\it belief functions}  (i.e., lower and upper probability functions) for $\btheta_0$, which can be used for frequentist inference. The plausible and belief functions are akin to the so-called {\it upper or lower confidence distributions} \citep[e.g.,][]{Thornton2022}, although the IM method is not considered by most (including Martin and Liu) as a frequentist development.  
\end{itemize}
In addition to the four reviewed methods, the repro samples also has some connections to the {\bf universal inference} \citep{wasserman2020universal} and {\bf conformal prediction} \citep{shafer2008tutorial} approaches, both of which will also be compared in the next subsection. 

\subsection{Repro samples method and connections to BFF inferential approaches}\label{sec:comp_smry}

\quad The development of a repro samples method borrows two important ideas from the aforementioned BFF procedures. First, the repro samples method utilizes artificial data samples to help quantify the uncertainty of statistical inference. This idea has been utilized in the approaches of bootstrap, ABC, GFI, and Monte-Carlo hypothesis testing, as well as more recent {\it simulation-based inference} (SBI) methods \citep{dalmasso2022,tomaselli2025robustsimulationbasedinference}, among others.  In these approaches, a key aspect of using artificial data samples for inference 
is the common attempt 
to match (in a certain way) the artificial data with the observed sample. This matching is done by directly comparing either the value or distribution of the artificial $\*y^*$ (or its summary statistic) with that of $\*y_{obs}$ (or its summary statistic). See \cite{Thornton2022} for a more elaborated discussion. The proposed repro samples method fully utilizes artificial samples and the matching idea to both develop inference approaches (e.g., Sections~\ref{sec:general} and \ref{sec: nuisance_control}) and solve computational questions (e.g., Section \ref{sec:finding_candidates}).
In addition, the repro samples method also borrows an idea from the IM development to quantify the inference uncertainty by first assessing the uncertainty of the unobserved $\*u^{rel}$, although the repro samples method goes further to emphasize the use of a nuclear mapping function, which makes the approach more flexible and effective for many complex inference problems. 
Here, by using the nuclear mapping function, our comparison is a little more flexible too, not limiting to only the direct comparison of the value or distribution of the artificial $\*y^*$ (or its summary statistic) with that of $\*y_{obs}$ (or its summary statistic).
The use of a nuclear mapping function is also closely connected to a test statistic used in a classical Neyman-Pearson test, although it is even more flexible than the notion of test statistic (see more details in Section~\ref{sec:NeymanInversion}).

We consider that the repro samples method is a further development (generalization and refinements) of these aforementioned existing procedures across the BFF paradigms. 
The two important ideas described in the previous paragraph have guided our developments of both finite and asymptotic inference procedures under the frequentist framework without the need to use any likelihood or decision theoretical criteria.
In the remaining of this subsection, we provide direct discussions, comparing the repro samples method with: (i) the ABC and GFI methods; (ii) the IM method; (iii) the bootstrap method
(iv) the universal inference method; (v) the conformal prediction approach; and (vi) the Monte-Carlo hypothesis testing procedure under the classical Neyman-Pearson framework. 
To conclude the subsection, we further describe the role that the repro samples method can play in bridging across BFF inference approaches.

\subsubsection{Repro samples method versus ABC and GFI methods}
Both the GFI and repro samples approach highlight the matching of equation $\*y_{obs} = G(\btheta, \*u^*)$ and the fact that when $\*u^*$ matches or is close to $\*u^{rel}$, the solution of $\btheta$ from the equation is equal or close to the true $\btheta_0$. The difference lies in that the GFI approach compares each single copy of artificial sample $\*y^* = G(\btheta, \*u^*)$ with the single copy of the observed $\*y_{obs}$, so does the ABC method. In the Monte-Carlo implementation of the repro samples method (e.g., using  Algorithm~\ref{alg:Ag} in Section~\ref{sec:guide}), however, we compare the single copy of the observed $\*y_{obs}$ with multiple copies of $\*y^s = G(\btheta, \*u^s)$
(or, more accurately, the realized $T(\btheta, \*u^{rel})$ with multiple copies of $T(\btheta, \*u^s)$'s), for $\*u^s \in {\mathcal V}$, where ${\mathcal V}$ is the collection of many copies of $\*u^s$. By comparing observed sample $\*y_{obs}$ (or realized $\*u^{rel}$) with many copies of artificially generated $\*y^s = G(\btheta, \*u^s)$ (or $\*u^s \in {\mathcal V}$), it allows us to use a confidence level $1 - \alpha$ to calibrate the uncertainty thus side-step the difficult issue in both GFI and ABC methods on how to appropriately preset the
threshold $\epsilon$. Furthermore, the use of the nuclear mapping function adds much flexibility to the repro samples method.

Comparing with the ABC method, the repro samples method sidesteps two issues (on validity and computing) 
encountered in 
an ABC method --- requiring  a) the summary statistic used in an ABC algorithm be a sufficient statistic 
and b) a pre-specified approximation tolerance 
value $\epsilon \to 0$ at a fast rate as sample size $n \to \infty$ \citep{Li2016}. 
Furthermore, the repro samples method is fully developed as
a  frequentist 
method, and it does not need  to assume 
 a prior distribution. 
 
 Comparing with the GFI method, the repro samples method avoids GFI's issue of requiring a pre-specified $\epsilon \to 0$ at a fast rate, as sample size $n \to \infty$, 
 similar to that of the ABC method. 
Furthermore, 
the repro samples method 
has a guaranteed finite-sample frequentist performance and can effectively handle irregular inference problems in which large-sample theorems do not apply, but a GFI method typically cannot. We also avoid the potential philosophical issues that are often associated with a fiducial inference on its interpretations. 
\subsubsection{Repro samples method versus IM method.}
The repro samples and IM  methods both promote first quantifying the uncertainty of $\*U$ and use it to help address the overall uncertainty inherited in the sampling data. To produce a level $1 - \alpha$ confidence set, the repro sample method uses a single fixed Borel set and a potential parameter value $\btheta$-dependent nuclear mapping function to help quantify the uncertainty in $\*U$. The single Borel set is fixed and not a random set, unlike the IM.   Moreover, the repro samples method is a fully frequentist approach
developed using only the standard probability tool. The IM method, on the other hand, attempts to achieve a higher level goal of producing a prior-free probabilistic inference for the unknown parameter $\btheta$. To achieve this goal, it requires to use random sets and a complex imprecise probability system known as {\it Dempster-Shafer calculus}. 

Furthermore, the IM method focuses on finite-sample inference, while the repro samples can be used for both finite and large-sample inferences. Also, the introduction and promotion of the so-called nuclear mapping function in the repro samples method allows it to greatly extend the scope of the inferential framework, 
making it more flexible and effective for many complex inference problems in both statistics and modern data science practice.

Finally, the repro samples method uses the standard Fisher inversion technique (not just the Fisher-Dempster inversion technique) that actually matches of single copies of simulated artificial $\*u^*$ with the (unknown) realized $\*u^{rel}$. The development of candidate set for discrete parameter $\eta$ in Section~\ref{sec:candidate_general} is rooted in the fact that when $\*u^*$ is equal or in a neighborhood of $\*u^{rel}$ can give us the same $\eta_0$ in the many-to-one mapping function (\ref{eq:tau-star}). This matching scheme may not be easily adapted in an IM approach.

\subsubsection{ Repro samples method versus Efron's Bootstrap and related methods.} 
 The key justification for inference methods based on bootstrap is the so-called {\it bootstrap central limit theorem}
\citep{Singh1981, Bickel1981} 
 in which the randomness of resampling (multinomial-distributed) is matched with the sampling randomness of the 
 data (often not-multinomial distributed) using the large-sample 
CLT. This matching is through the asymptotic variances (one conditional and the other unconditional) of the two point estimators (i.e., the bootstrap point estimator computed using  the artificial bootstrap sample and the conventional point estimator computed using  the actual sample data). Instead of resampling, the repro samples method directly generates artificial samples using $\*u^*$. The matching in a repro samples method is  more direct and flexible, not limiting to point estimators. More importantly, a repro samples method does not need to rely on any large-sample theorems. It is more broadly applicable and can effectively address many inference problems 
 that the bootstrap methods cannot, 
 including those in which the target parameters are discrete or non-numerical.  

\subsubsection{ Repro samples method versus universal inference approach.} Universal inference \citep{wasserman2020universal} is a  framework recently developed to provide performance-guaranteed  finite-sample 
    inferences
without regularity conditions. 
The universal inference method typically requires a tractable likelihood function and data splitting. It relies on a Markov inequality to justify its validity (testing size), at some expanse of power \citep{Dunn2022, TseDavison2023, xie2023discussion}. 
The universal inference framework can also handle (designed to handle)  irregular inference problems and is particularly effective for sequential testing problems \citep{ramdas2020admissible,xie2023discussion}. 

The case study example of Gaussian mixture is also studied in \cite{wasserman2020universal} but, unlike the repro samples method, it cannot provide a two-sided confidence set for the unknown number of components $\tau_0$.
The repro samples method
is a likelihood-free approach that is also effective for irregular inference problems, and it does not systemically suffer a loss of power while maintaining inference validity. 
See also Appendix~\ref{sec:comp_additional_example} for a concrete,  illustrative example that dissects the sources of power loss in the universal inference method while the repro samples method does not suffer any power loss.

Finally, the universal inference method has been extended to situations where the model is completely unknown and misspecified \cite{park2023robust}. A similar extension, but under a specific binary classification situation , has also been done within the repro samples framework \cite{hou2025repro}.

\subsubsection{ Repro samples method versus conformal prediction method}\label{sec:conformal}
 Conformal prediction is a popular nonparametric prediction method used to construct a level $1 - \alpha$ confidence set for a future  (unobserved) observation, say $y_{new}$, given observed 
$n$ exchangeable data points $\{y_{1}, \ldots, y_n\}$  \citep{vovk2005algorithmic, shafer2008tutorial}. Particularly, for a potential value $y$ of $y_{new}$, if it is within a $\alpha$-level central region of $\{y_{1}, \ldots, y_n\}$, then the value $y$ is {\it conformal} with $\{y_{1}, \ldots, y_n\}$
 at the $\alpha$-level. The collection of such $y$ values forms a level $1 - \alpha$ predictive set of $y_{new}$. In our case, we are interested in 
 a fixed unknown parameter
$\btheta_0$, instead of the random 
 $y_{new}$. Nevertheless, we can use the same conformal concept to intuitively explain the Monte-Carlo repro samples approach (cf., Algorithm~\ref{alg:Ag} in Section~\ref{sec:guide}): 
 For a potential value $\btheta$ of $\btheta_0$, we 
 generate $|{\mathcal V}|$ sets of artificial data $\*y^s = G(\btheta,\*u^s)$ (or, more accurately,  $|{\mathcal V}|$ copies of $T(\btheta, \*u^s)$'s), for $ \*u^s \in {\mathcal V}$; If the  
 observed data set $\*y_{obs}$ is conformal with the $|{\mathcal V}|$ artificial data sets $\*y_s$'s (or, more accurately, if there exists an $\*u^* \in {\mathcal U}$ such that $\*y_{obs} = G(\btheta, \*u^*)$ and the value $T(\*u^*, \btheta)$
is conformal with multiple copies of $T(\btheta, \*u^s)$'s), for $ \*u^s \in {\mathcal V}$, 
 we collect the $\btheta$ value. The collection of such $\btheta$ values forms a level $1 - \alpha$ confidence set of $\btheta_0$. 
 
 \subsubsection{ Repro samples method versus the classical Monte-Carlo hypothesis testing approach}  As stated in Appendix~\ref{sec:conformal} above, a Monte-Carlo implementation  of the repro samples method in Algorithm~\ref{alg:Ag} of Section~\ref{sec:guide} can be described as follows: for any potential value $\btheta \in \Theta$, we generate multiple copies of $\*u^s$ (denote the set of their collection ${\mathcal V}$) and compute corresponding $T(\*u^s, \btheta)$'s. If there exists an $\*u^* \in {\mathcal U}$ such that $\*y_{obs} = G(\btheta, \*u^*)$ and the value $T(\*u^*, \btheta)$
{\it is  conformal with} the multiple copies of $T(\*u^s, \btheta)$, $\*u^s \in {\mathcal V}$, at the level $1- \alpha$, then we keep the $\btheta$ in  $\Gamma_{1-\alpha}(\*y_{obs})$. Here, the matching of a single copy of $T(\*u^*, \btheta)$ with the multiple copies of $T(\*u^s, \btheta)$'s is evaluated by a conformal measure at level $1 - \alpha$, where the concept of  ``conformal" is 
the same as that in  conformal prediction \citep{vovk2005algorithmic}. 

In the special case when the nuclear mapping is defined through a test statistic $T(\*u, \btheta) = \widetilde T(\*y, \btheta)$ where $\*y = G(\btheta, \*u)$ (as discussed in Section~\ref{sec:NeymanInversion}) and also that we can always find a $\*u^* \in {\mathcal U}$ such that $\*y_{obs} = G(\btheta, \*u^*)$ for any $\btheta$, the Monte-Carlo Algorithm~\ref{alg:Ag} is simplified to: for any potential value $\btheta \in \Theta$, we generate $|{\mathcal V}|$ copies of artificial data 
$\*y^s = G(\btheta, \*u^s)$ and compute the corresponding $\widetilde T(\*y^s, \btheta)$, for $\*u^s \in {\mathcal V}$. If the test statistic $\widetilde T(\*y_{obs}, \btheta)$ is conformal with the many copies of $\widetilde T(\*y^s, \btheta)$'s at level $\alpha$, then we keep this $\btheta$ to form the level $1 - \alpha$ confidence set $\Gamma_{1-\alpha}(\*y_{obs})$. 
It is easy to see that this simplified algorithm is equivalent to a Monte-Carlo implementation of the classical hypothesis test of $H_0: \btheta_0 = \btheta$, where the conformal statement above corresponds to the statement of not rejecting $H_0$. Again, we see that the classical testing approach may be viewed a special case of our repro samples method. In the Monte-Carlo version, we utilize multiple copies of artificial samples $\*y^s$'s and compare them to a single copy of the observed data $\*y_{obs}$.

\subsubsection
{Artificial-sample-based inference and a bridge of BFF paradigms} 

Finally, as an artificial-sample-based inference procedure, the repro samples method plays a role in bridging across BFF inference approaches. 

There have been several recent developments on the foundation of statistical inference across Bayesian, fiducial, and frequentist (BFF) paradigms. \cite{Reid2022} and \cite{Thornton2022} provide comprehensive overviews on several recent BFF research developments at the foundation level. Based on the development of confidence distribution,  \cite{Thornton2022} also explore the aspect of making inferences through matching artificial samples with observed data across BFF procedures. It argues that the matching of simulated artificial randomness with the sampling randomness inherited from a statistical model provides a unified
bridge to connect BFF inference procedures.
Here, the artificial randomness includes, for example, bootstrap randomness in bootstrap, MCMC randomness in a Bayesian analysis, the randomness of $\*U^*$ in fiducial procedures such as the 
GFI and IM methods. 
The repro samples method is another development that seeks to align simulated artificial randomness with the sampling randomness. By doing so we can effectively measure and quantify the uncertainty in our statistical inferential statements. 
The repro samples method can be used to address many difficult inference problems, especially those involving a discrete parameter space and also those where large-sample CLT does not apply. It provides a further advancement on the foundation of statistical inference to meet the growing need for ever-emerging data science.

\subsection{Additional illustrations Examples for Comparison}\label{sec:comp_additional_example}

\subsubsection{Example of systemic power loss ---  universal inference versus repro samples method} 
\label{sec:universal-power-loss}

\cite{TseDavison2023} provides an excellent note dissecting the performance of the universal inference method. The simple normal example they used to create Figure 1 in their paper offers a clear insight that shows the universal inference method systematically suffers a loss of power due to the use of the Markov inequality and sample splitting. We revisit this example below and demonstrate that the repro samples method does not exhibit the same systemic power loss issue.

\begin{example}
\label{ex:universal}
\citep[][Example setting for Figure 1]{TseDavison2023}
  Suppose we have $n = 1{,}000$ samples from $N(\theta_0, 1)$ where the unknown true parameter value $\theta_0 =0$. The goal is to make inference for the unknown $\theta_0$.

 If we use the repro samples method, we have $\bar y^{obs} = \theta_0 + \bar u^{rel}$, where $\bar y^{obs}$ is the sample mean and $\bar u^{rel}$ is the mean of the $n = 1,000$ realized $N(0,1)$ noises. So if in Eq (\ref{eq:B}) we use a nuclear mapping function $T(\mathbf{u}) = \bar u$ with a corresponding Borel set $B_{1-\alpha} = 
 (\Phi^{-1}(\frac \alpha2)/\sqrt{n}, \Phi^{-1}(1 - \frac \alpha2)/\sqrt{n}$), 
 then the interval by the repro samples method $$\Gamma_{1 - \alpha}(\mathbf{y}_{obs}) = \left(\bar y^{obs} + \frac1{\sqrt{n}}\Phi^{-1}(\frac \alpha2), \,\,\, \bar y^{obs} + \frac1{\sqrt{n}} \Phi^{-1}(1 - \frac \alpha2)\right), $$ which is the shortest level $1 - \alpha$ confidence interval. To prepare for comparison next, we re-write $\Gamma_{1 - \alpha}(\mathbf{y}_{obs})$ as 
 \begin{equation}
 \label{eq:universal-repro-example}
 \Gamma_{1 - \alpha}(\mathbf{y}_{obs}) = \left\{ \theta: 
 n(\theta - \bar y^{obs})^2 \leq  \chi_1^2(1 - \alpha) \right\}, 
 \end{equation}
 where $\chi_1^2(1 - \alpha)$ is the $1 - \alpha$ quantile of the $\chi_1^2$ distribution.  In this example, the inequalities ``$\ge$" in both Eq (\ref{eq:B}) and Eq (\ref{eq:G1}) can all be replaced by ``$=$".  There is no loss of any power.

Now if we use a universal inference method, which needs to split the data in half, the level $1 - \alpha$ confidence interval is 
\begin{equation}
  \label{eq:universal-example}
C_{1 - \alpha}(\mathbf{y}_{obs}) = \left\{\theta: \frac{n}2(\theta - \bar y_0)^2 - \frac{n}2(\bar y_1 - \bar y_0)^2 \leq  - 2\log(\alpha)\right\},
 \end{equation}
 where $\bar y_1$ and $\bar y_0$ are the sample mean of $500$ training and testing data, respectively.  By comparing (\ref{eq:universal-repro-example}) and (\ref{eq:universal-example}) and following the breakdown equations (5)-(7)
of \cite{TseDavison2023},  
we can clearly see two major sources of systemic power loss of the universal confidence interval $C_{1 - \alpha}(\mathbf{y}_{obs})$:
\begin{itemize}
[leftmargin=1.5em, labelsep=1.5em, 
align=parleft,   itemsep=0pt,
  parsep=0pt,
  topsep=0pt] 
\item[(a)] by sample splitting, noting that $\frac{n}2(\theta - \bar y_0)^2 - \frac{n}2(\bar y_1 - \bar y_0)^2 < n(\theta - \bar y^{obs})^2$,
\item[(b)] by the use of the
Markov inequality to control type 1 error in all cases, noting that $-2\log(\alpha) >  \chi_1^2(1 - \alpha)$). 
\end{itemize}

\end{example}

\subsubsection{Inference under different model expressions of the sample inference problem -- Repro vs GFI \& IM} \label{sec:App-G-unique-example}

In this discussion, we define the notion of ``{\it the same inference problem}'' as having the same probabilistic model with the same sample likelihood function. \cite{Hannig2009} noted that there may be different versions of a generative model (\ref{eq:1}) for the same inference problem. For instance, \citealp[][Remark~6]{Hannig2009} pointed out that a generative model based on the full sample 
\(\mathbf{Y} = (Y_1, \ldots, Y_n)\) and one based on a sufficient statistic \(S(\mathbf{Y})\) may differ, even though they correspond to the same underlying probabilistic model. This raises a potential issue of \emph{model identifiability} if inference results depend on the form of the generative model used and differ across model formulations. Indeed, both the outcomes (i.e., the generalized fiducial distribution and the plausible function) 
of GFI and IM depend on the formulation, especially on the distribution of \(\mathbf{U}\) used in the model. Their outcomes (distributions) for the same inference problem differ by a Jacobian matrix for difference model forms. However, with the use of a nuclear function, the repro samples method side-steps this issue, and as long as we use the same nuclear mapping function in our inference, we can produce the same outcome for these alternative forms of the generative model. Below is a concrete example that illustrates this discussion.

\begin{example}
\label{ex:fiducial}
\citep[][Example~4]{Hannig2016} Suppose $Y_i \sim U(\theta, \theta^2)$ for $i = 1, \ldots, n$ with $\theta > 1$. Then the generative model based on the full sample is
\begin{equation}  \label{eq:ModelFullY}
Y_i = \theta + \theta(\theta - 1) U_i, \quad U_i \sim U(0,1),
\end{equation}
while the model based on the minimal sufficient statistic $(Y_{(1)}, Y_{(n)})$, the smallest and largest $Y_i$,  is
\begin{equation}  \label{eq:ModelSufficientY} 
Y_{(1)} = \theta + \theta(\theta - 1) U_{(1)} \mbox{ and }
Y_{(n)} = \theta + \theta(\theta - 1) U_{(n)},
\end{equation}
where $U_{(1)}$ and $U_{(n)}$ are the smallest and largest $U_i$, respectively.

In the repro samples framework, such alternative formulations in (\ref{eq:ModelFullY}) and (\ref{eq:ModelSufficientY}) still yield the same inference on $\btheta$, provided the same nuclear mapping function is applied. Note that $\{\theta: 
y_i^{obs} = \theta + \theta(\theta - 1) u^*_i, \exists u^*_i \in (0,1), 
i = 1, \ldots, n, 
\} = \{\theta: y_{(1)}^{obs} = \theta + \theta(\theta - 1) u^*_{(1)}, \,\, 
y_{(n)}^{obs} = \theta + \theta(\theta - 1) u^*_{(n)}, \exists (u^*_{(1)},  u^*_{(n)}), 0<u^*_{(1)}< u^*_{(n)} < 1 \}$. If we adopt the same nuclear mapping function $(U_{(1)}, U_{(n)})$ with the Borel set 
\[
B_{1-\alpha} = \left\{\mathbf{U}^*: U_{(1)}^* > \left(\tfrac{\alpha}{2}\right)^{1/n}, \, U_{(n)}^* < 1 - \left(\tfrac{\alpha}{2}\right)^{1/n} \right\},
\]
 then  the resulting level-$(1 - \alpha)$ confidence set for $\theta$ defined in (\ref{eq:G1}) remains the same whether the generative model is (\ref{eq:ModelFullY}) or (\ref{eq:ModelSufficientY}).
This contrasts with the Generalized Fiducial and Inferential Model (IM) approaches, where results depend on the distributional form of $\mathbf{U}$ in (\ref{eq:ModelFullY}) and $(U_{(1)}, U_{(n)})$ in (\ref{eq:ModelSufficientY}). Particularly, $\mathbf{U}=(U_1,\ldots,U_n)$ and $(U_{(1)}, U_{(n)})$ differ by a Jacobian factor, so the fiducial distributions obtained from the two forms of generative models are not the same \citep[see][Remark~6]{Hannig2016}.    
\end{example}

\bibliographystyle{agsm_nourl}
\bibliography{ref4all}

@book{kallenberg1997foundations,
  title={Foundations of modern probability},
  author={Kallenberg, Olav},
  year={2021},
  publisher={Springer},
edition = {3rd}
}

@article{park2023robust,
  title={Robust universal inference},
  author={Park, Beomjo and Balakishan, Sivaraman and Wasserman, Larry},
  journal={arXiv2307.04034},
  year={2023}
}

@article{tomaselli2025robustsimulationbasedinference,
      title={Robust Simulation Based Inference}, 
      author={Lorenzo Tomaselli and Valérie Ventura and Larry Wasserman},
      year={2025},
journal={arXiv:2508.02404}
}

@article{LiuShao2003,
author = {Xin Liu and Yongzhao Shao},
title = {{Asymptotics for likelihood ratio tests under loss of identifiability}},
volume = {31},
journal = {Ann. Stat.},
number = {3},
publisher = {Institute of Mathematical Statistics},
pages = {807 -- 832},
year = {2003},
doi = {10.1214/aos/1056562463},
URL = {https://doi.org/10.1214/aos/1056562463}
}

@article{dudley_assessing_1991,
  title = {Assessing the Role of {{APNH}}, a Gene Encoding for a Human Amiloride-Sensitive {{Na}}+/{{H}}+ Antiporter, on the Interindividual Variation in Red Cell {{Na}}+/{{Li}}+ Countertransport.},
  author = {Dudley, C R and Giuffra, L A and Raine, A E and Reeders, S T},
  year = {1991},
  month = oct,
  journal = {JASN},
  volume = {2},
  number = {4},
  pages = {937--943},
  issn = {1046-6673, 1533-3450},
  doi = {10.1681/ASN.V24937},
  language = {en}
}

@article{powell_censored_1986,
  title = {Censored Regression Quantiles},
  author = {Powell, James L.},
  year = {1986},
  month = jun,
  journal = {Journal of Econometrics},
  volume = {32},
  number = {1},
  pages = {143--155},
  issn = {0304-4076},
  doi = {10.1016/0304-4076(86)90016-3},
  langid = {english}
}

@article{roeder_graphical_1994,
  title = {A {{Graphical Technique}} for {{Determining}} the {{Number}} of {{Components}} in a {{Mixture}} of {{Normals}}},
  author = {Roeder, Kathryn},
  year = {1994},
  journal = {J. Am. Stat. Assoc.},
  volume = {89},
  number = {426},
  pages = {487--495},
  publisher = {{[American Statistical Association, Taylor \& Francis, Ltd.]}},
  issn = {0162-1459},
  doi = {10.2307/2290850}
}

@article{chen_inference_2012,
  title = {Inference on the {{Order}} of a {{Normal Mixture}}},
  author = {Chen, Jiahua and Li, Pengfei and Fu, Yuejiao},
  year = {2012},
  month = sep,
  volume = {107},
  pages = {1096--1105},
  issn = {0162-1459, 1537-274X},
  doi = {10.1080/01621459.2012.695668},
  journal = {J. Am. Stat. Assoc.},
  language = {en},
  number = {499}
}

@article{shafer2008tutorial,
	title={A tutorial on conformal prediction},
	author={Shafer, Glenn and Vovk, Vladimir},
	journal={J. Mach. Learn. Res.},
	volume={9},
	number={Mar},
	pages={371--421},
	year={2008}
}

@book{vovk2005algorithmic,
	title={Algorithmic learning in a random world},
	author={Vovk, Vladimir and Gammerman, Alex and Shafer, Glenn},
	year={2005},
	publisher={Springer Science \& Business Media}
}

@book{Martin2015,
	Author = {Martin, Ryan and  Liu, Chuanhai},
	Publisher = {Chapman \& Hall},
	Title = {Inferential Models: Reasoning with Uncertainty},
	Year = {2015}}

@article{hastie2012,
  title={Model choice using reversible jump {Markov} chain {Monte} {Carlo}},
  author={Hastie, David I and Green, Peter J},
  journal={Statistica Neerlandica},
  volume={66},
  number={3},
  pages={309--338},
  year={2012},
  publisher={Wiley Online Library}
}

@article{Kass1995,
author = { Robert   Kass  and  Adrian   Raftery },
title = {Bayes Factors},
journal = {J. Am. Stat. Assoc.},
volume = {90},
number = {430},
pages = {773-795},
year  = {1995},
publisher = {Taylor & Francis},
doi = {10.1080/01621459.1995.10476572}
}

@article{Rubin1984,
	Author = {Rubin, D. B.},
	Journal = {Ann. Stat.},
	Pages = {1151--1172},
	Title = {Bayesianly justifiable and relevant frequency calculations for the applied statistician},
	Volume = {12},
	Year = {1984}}

@article{Tavare1997,
  title={Inferring coalescence times from
DNA sequence data},
  author={Tavar{\'e}, Simon and Balding, D. J.  and Griffiths, R. C. and Donnelly, P.},
  journal={Genetics},
  volume={145},
  pages={505--518},
  year={1997}
}

@ARTICLE{Peters2012,
	author={Peters, G.W. and Fan, Y. and Sisson, S.A.},
	title={On sequential Monte Carlo partial rejection control approximate Bayesian computation},
	journal={Statistical Computing},
	year={2012},
	volume={22},
	pages={1209-1222}
}

@ARTICLE{Barber2015,
	title={The rate of convergence for approximate Bayesian computation},
	author={Barber, Stuart and Voss, Jochen and Webster, Mark},
	journal={Electronic Journal of Statistics},
	volume={9},
	number={1},
	pages={80--105},
	year={2015},
	publisher={The Institute of Mathematical Statistics and the Bernoulli Society}
}

@ARTICLE{Li2016,
	author={Li, Wentao and Fearnhead, Paul},
	title={On the Asymptotic Efficiency of approximate Bayesian computation estimators},
	journal={Biometrika},
	year={2018},
	volume={105},
	number={2},
	pages={286-299}
}

@article{Thornton2018,
    author = {Suzanne Thornton and Wentao Li and Minge Xie},
    title = {Approximate Confidence Distribution Computing},
    journal = {The New England Journal of Statistics in Data Science},
    volume = {1},
    number = {2},
    year = {2023},
    pages = {270--282},
    doi = {10.51387/23-NEJSDS38},
    issn = {2693-7166},
    publisher = {New England Statistical Society}
}

@incollection{Thornton2022,
	Author = {Thornton, S. and Xie, M.},
  title={Bridging Bayesian, frequentist and fiducial (BFF) inferences using confidence distribution},
  booktitle={Handbook on Bayesian, Fiducial and Frequentist (BFF) Inferences},
  editor = {Berger, J.O. and Meng, X.-L. and Reid, N. and Xie, M.},
  year={2024},
  publisher={Chapman \& Hall}
}

@incollection{Reid2022,
	Author = {Reid, Nancy},
  title={`Distributions for Parameters'},
  booktitle={Handbook on Bayesian, Fiducial and Frequentist (BFF) Inferences}, 
  editor = {Berger, J.O. and Meng, X.-L. and Reid, N. and Xie, M.},
  year={2024},
  publisher={Chapman \& Hall}
}

@article{Hannig2016,
	author = {Hannig, Jan and Iyer, Hari and Lai, Randy C. S. and Lee, Thomas C. M.},
title = {Generalized Fiducial Inference: A Review and New Results},
journal = {J. Am. Stat. Assoc.},
volume = {111},
number = {515},
pages = {1346-1361},
year = {2016},
publisher = {Taylor & Francis},
doi = {10.1080/01621459.2016.1165102},


URL = { 
    
        https://doi.org/10.1080/01621459.2016.1165102
    
    

},
eprint = { 
    
        https://doi.org/10.1080/01621459.2016.1165102
    
    

}

}

@article{Hannig2009,
	Author = {Hannig, J.},
	Journal = {Stat. Sin.},
	Pages = {491-544},
	Title = {On generalized fiducial inference.},
	Volume = {19},
	Year = {2009}}

@article{Zabell1992,
	Author = {Zabell, S. L.},
	Journal = {Stat. Sci.},
	Pages = {369--387},
	Title = {R. A. Fisher and the fiducial argument},
	Volume = {7},
	Year = {1992}}

@article{Xie2013,
	Author = {Xie, M. and Singh, K.},
	Journal = {Int. Stat. Rev.},
	Pages = {3--39},
	Title = {Confidence distribution, the frequentist distribution estimator of a parameter: a review},
	Volume = {81},
	Year = {2013}}

@article{martin2013inferential,
	Author = {Martin, Ryan and Liu, Chuanhai},
	Journal = {J. Am. Stat. Assoc.},
	Number = {501},
	Pages = {301--313},
	Publisher = {Taylor \& Francis Group},
	Title = {Inferential models: A framework for prior-free posterior probabilistic inference},
	Volume = {108},
	Year = {2013}}

@article{Chuang2000,
	Author = {Chuang, Chin-Shan and Lai, Tze Leung},
	Journal = {Stat. Sin.},
	Pages = {1-50},
	Title = {HYBRID RESAMPLING METHODS FOR CONFIDENCE INTERVALS},
	Volume = {10},
	Year = {2000}}

@article{Michael2019,
author = {Michael, Haben and Thornton, Suzanne and Xie, Minge and Tian, Lu},
title = {Exact inference on the random-effects model for meta-analyses with few studies},
journal = {Biometrics},

volume = {75},
number = {2},
pages = {485-493},
doi = {10.1111/biom.12998},
url = {https://onlinelibrary.wiley.com/doi/abs/10.1111/biom.12998},
eprint = {https://onlinelibrary.wiley.com/doi/pdf/10.1111/biom.12998},
year = {2019}
}

@ARTICLE {Singh1981,
	author  = "Singh, Kesar",
	title   = "On the asymptotic accuracy of {E}fron's bootstrap",
	journal = "Ann. Stat.",
	year    = "1981",
	volume  = "9",
	number  = "6",
	pages   = "1187-1195"
}

@article{Bickel1981,
	Author = {Bickel, Peter and Freedman, D. },
	Journal = {Ann.  Stat.},
	Pages = {1196-1217},
	Title = {Some asymptotic theory for the bootstrap},
	Volume = {9},
	Year = {1981}}

@article{pokotylo_depth_2019,
  title = {Depth and {{Depth-Based Classification}} with {{R Package}} Ddalpha},
  author = {Pokotylo, Oleksii and Mozharovskyi, Pavlo and Dyckerhoff, Rainer},
  year = {2019},
  month = oct,
  journal = {J Stat Softw},
  volume = {91},
  pages = {1--46},
  issn = {1548-7660},
  doi = {10.18637/jss.v091.i05},
  copyright = {Copyright (c) 2019 Oleksii Pokotylo, Pavlo Mozharovskyi, Rainer Dyckerhoff},
  langid = {english}
}

@article{komarek2014capabilities,
  title={Capabilities of {R} package {mixAK} for clustering based on multivariate continuous and discrete longitudinal data},
  author={Kom{\'a}rek, Arno{\v{s}}t and Kom{\'a}rkov{\'a}, Lenka},
  journal={J Stat Softw},
  volume={59},
  pages={1--38},
  year={2014}
}

@article{richardson_bayesian_1997,
  title = {On {{Bayesian Analysis}} of {{Mixtures}} with an {{Unknown Number}} of {{Components}} (with Discussion)},
  author = {Richardson, Sylvia. and Green, Peter J.},
  year = {1997},
  journal = {J. R. Stat. Soc. Ser. B},
  volume = {59},
  number = {4},
  pages = {731--792},
  issn = {1467-9868},
  doi = {10.1111/1467-9868.00095},
  langid = {english}
}

@article{Efron1998,
	Author = {Efron, B.},
	Journal = {Stat. Sci.},
	Pages = {95--122},
	Title = {R. A. Fisher in the 21st century},
	Volume = {13},
	Year = {1998}}

@book{Berger2020,
	editor = {Berger, J.O. and Meng, X.-L. and Reid, N. and Xie, M.},
	Publisher = {Chapman \& Hall: New York},
	Title = {Handbook on Bayesian,
Fiducial and Frequentist (BFF) Inferences.},
	Year = {2020}}

@Book{EfroTibs93,
  Title  = {An Introduction to the Bootstrap},
  Author  = {Bradley Efron and Robert J. Tibshirani},
  Publisher = {Chapman \& Hall/CRC},
  Year   = {1993},
  Address  = {Florida, USA},
  Series  = {Monographs on Statistics and Applied Probability}
}

@book{fishman1996,
  title={Monte Carlo: concepts, algorithms, and applications},
  author={Fishman, G},
  year={1996},
  publisher={Springer Science \& Business Media},
  address={New York, USA}
}

@incollection{Robert2016,
	author  = "Robert, C. P",
	title   = "Approximate {B}ayesian Computation: A Survey on Recent Results",
	editor = "Cools R., Nuyens D.",
	booktitle = "Monte Carlo and Quasi-Monte Carlo Methods",
	publisher = "Springer Proceedings in Mathematics \and Statistics",
	address = "Switzerland",
	year = 2016,
	volume = 163, 
	pages = "185-205"
}

@Book{tCAS90a,
  author =	 {G. Casella and R. L. Berger},
  title = 	 {Statistical Inference},
  publisher = 	 {Wadsworth and Brooks/Cole},
  year = 	 {1990},
  address =	 {Pacific Grove, CA},
  OPTnote = 	 {},
  OPTannote = 	 {}
}

@article{liu_multivariate_1999-1,
  title = {Multivariate Analysis by Data Depth: Descriptive Statistics, Graphics and Inference},
  shorttitle = {Multivariate Analysis by Data Depth},
  author = {Liu, Regina and Parelius, Jesse M. and Singh, Kesar},
  year = {1999},
  month = jun,
  journal = {Ann. Stat.},
  volume = {27},
  number = {3},
  pages = {783--858},
  publisher = {{Institute of Mathematical Statistics}},
  issn = {0090-5364, 2168-8966},
  doi = {10.1214/aos/1018031260}
}

@article{Liu2021,
author = {Dungang Liu and Regina Y. Liu and Minge Xie},
title = {Nonparametric Fusion Learning for Multiparameters: Synthesize Inferences From Diverse Sources Using Data Depth and Confidence Distribution},
journal = {J. Am. Stat. Assoc.},
volume = {117},
number = {540},
pages = {2086-2104},
year = {2022},
publisher = {Taylor & Francis},
doi = {10.1080/01621459.2021.1902817},


URL = { 
    
        https://doi.org/10.1080/01621459.2021.1902817
    
    

},
eprint = { 
    
        https://doi.org/10.1080/01621459.2021.1902817
    
    

}

}

@article{wang2022highdimenional,
  title={Finite and large sample inference for model and coefficients in high dimensional linear regression with repro samples},
  author={Wang, Peng and Xie, Minge and Zhang, Linjun},
  year={2025},
  journal={Ann. Stat. (in press)}
}

@article{wasserman2020universal,
  title={Universal inference},
  author={Wasserman, Larry and Ramdas, Aaditya and Balakrishnan, Sivaraman},
  journal={PNAS},
  volume={117},
  number={29},
  pages={16880--16890},
  year={2020},
  publisher={National Acad Sciences}
}

@article{grun2007fitting,
  title={Fitting finite mixtures of generalized linear regressions in R},
  author={Gr{\"u}n, Bettina and Leisch, Friedrich},
  journal={Computational Statistics \& Data Analysis},
  volume={51},
  number={11},
  pages={5247--5252},
  year={2007},
  publisher={Elsevier}
}

@article{Dunn2022,
    author = {Dunn, Robin and Ramdas, Aaditya and Balakrishnan, Sivaraman and Wasserman, Larry},
    title = "{Gaussian universal likelihood ratio testing}",
    journal = {Biometrika},
    volume = {110},
    number = {2},
    pages = {319-337},
    year = {2022},
    month = {11},
    issn = {1464-3510},
    doi = {10.1093/biomet/asac064},
    url = {https://doi.org/10.1093/biomet/asac064},
    eprint = {https://academic.oup.com/biomet/article-pdf/110/2/319/50311367/asac064.pdf},
}

@article{dalmasso2022,
author = {Niccol{\`o} Dalmasso and Luca Masserano and David Zhao and Rafael Izbicki and Ann B. Lee},
title = {{Likelihood-free frequentist inference: bridging classical statistics and machine learning for reliable simulator-based inference}},
volume = {18},
journal = {Electron. J. Stat.},
number = {2},
publisher = {Institute of Mathematical Statistics and Bernoulli Society},
pages = {5045 -- 5090},
year = {2024},
doi = {10.1214/24-EJS2307},
URL = {https://doi.org/10.1214/24-EJS2307}
}

@article{xie2023discussion,
  title={Discussion on “A note on universal inference” by {D}rs. {T}immy {T}se and {A}nthony {D}avison},
  author={Xie, Minge},
  journal={Stat},
  volume={12},
  number={1},
  pages={e4571},
  year={2023},
  publisher={Wiley Online Library}
}

@article{TseDavison2023,
	author = {Tse, Timmy and Davison, Anthony C.},
title = {A note on universal inference},
journal = {Stat},
volume = {11},
number = {1},
pages = {e501},
doi = {https://doi.org/10.1002/sta4.501},
url = {https://onlinelibrary.wiley.com/doi/abs/10.1002/sta4.501},
eprint = {https://onlinelibrary.wiley.com/doi/pdf/10.1002/sta4.501},
year = {2022}
}

@article{ramdas2020admissible,
  title={Admissible anytime-valid sequential inference must rely on nonnegative martingales},
  author={Ramdas, Aaditya and Ruf, Johannes and Larsson, Martin and Koolen, Wouter},
  journal={Preprint arXiv:2009.03167},
  year={2020}
}

@misc{delbarrio2022nonparametric,
      title={Nonparametric multiple-output center-outward quantile regression},
  author={del Barrio, Eustasio and Sanz, Alberto Gonzalez and Hallin, Marc},
  journal={Preprint arXiv:2204.11756},
  year={2022}
}

@article{hallin2021distribution,
  author = {Marc Hallin and Eustasio del Barrio and Juan Cuesta-Albertos and Carlos Matr{\'a}n},
title = {{Distribution and quantile functions, ranks and signs in dimension $d$: a  transportation approach}},
volume = {49},
journal = {Ann. Stat.},
publisher = {Institute of Mathematical Statistics},
pages = {1139--1165},
year = {2021},
doi = {10.1214/20-AOS1996},
URL = {https://doi.org/10.1214/20-AOS1996}
}

@article{koenker1994quantile,
  title={Quantile smoothing splines},
  author={Koenker, Roger and Ng, Pin and Portnoy, Stephen},
  journal={Biometrika},
  volume={81},
  pages={673--680},
  year={1994},
  publisher={Oxford University Press}
}

@article{takeuchi2006nonparametric,
  title={Nonparametric quantile estimation},
  author={Takeuchi, Ichiro and Le, Quoc and Sears, Timothy and Smola, Alexander and others},
  year={2006},
  publisher={MIT Press}
}

@article{csillery2010approximate,
  title={Approximate Bayesian computation (ABC) in practice},
  author={Csill{\'e}ry, Katalin and Blum, Michael GB and Gaggiotti, Oscar E and Fran{\c{c}}ois, Olivier},
  journal={Trends in ecology \& evolution},
  volume={25},
  number={7},
  pages={410--418},
  year={2010},
  publisher={Elsevier}
}

@book{smith2007calibration,
  title={Calibration of Computer Models},
  author={Smith, Jim Q.},
  year={2007},
  publisher={John Wiley \& Sons},
  address={Hoboken, NJ},
  isbn={978-0-470-01567-8},
}

@article{wang2019variational,
  title={Variational Bayes under model misspecification},
  author={Wang, Yixin and Blei, David},
  journal={Advances in Neural Information Processing Systems},
  volume={32},
  year={2019}
}

@article{chen2023exact,
  title={Exact Inference for Common Odds Ratio in Meta-Analysis with Zero-Total-Event Studies},
  author={Xiaolin Chen and Jerry Q Cheng and Lu Tian and Minge Xie},
  journal={Statistics in Biosciences},
  year={2023},
  url={https://api.semanticscholar.org/CorpusID:264405845}
}

@article{XieWang2022,
  title = {Repro {{Samples Method}} for {{Finite-}} and {{Large-Sample Inferences}}},
  author = {Xie, Minge and Wang, Peng},
  year = {2022},
  month = jun,
  journal = {arXiv e-prints},
  eprint = {2206.06421},
  eprinttype = {arxiv},
  primaryclass = {math, stat},
  pages = {arXiv.2206.06421},
  doi = {10.48550/arXiv.2206.06421},
  archiveprefix = {arXiv}
}

@article{vu2015random,
  title={Random weighted projections, random quadratic forms and random eigenvectors},
  author={Vu, Van and Wang, Ke},
  journal={Random Structures \& Algorithms},
  volume={47},
  number={4},
  pages={792--821},
  year={2015},
  publisher={Wiley Online Library}
}

@article{liang2025extended,
  title={Extended fiducial inference: toward an automated process of statistical inference},
  author={Liang, Faming and Kim, Sehwan and Sun, Yan},
  journal={J. R. Stat. Soc. Ser. B (Stat. Methodol.)},
  volume={87},
  number={1},
  pages={98--131},
  year={2025},
  publisher={Oxford University Press UK}
}

@article{Chen2017,
  author    = {Chen, Jiahua},
  title     = {Consistency of the MLE under Mixture Models},
  journal   = {Statistical Science},
  year      = {2017},
  volume    = {32},
  number    = {1},
  pages     = {47--63}
}

@article{FraleyRaftery2002,
  author    = {Fraley, Chris and Raftery, Adrian E.},
  title     = {Model-Based Clustering, Discriminant Analysis, and Density Estimation},
  journal   = {Journal of the American Statistical Association},
  year      = {2002},
  volume    = {97},
  number    = {458},
  pages     = {611--631}
}

@article{konishi1996gic,
  author  = {Konishi, Sadanori and Kitagawa, Genshiro},
  title   = {Generalised information criteria in model selection},
  journal = {Biometrika},
  year    = {1996},
  volume  = {83},
  number  = {4},
  pages   = {875--890}
}

@book{konishi2008information,
  author    = {Konishi, Sadanori and Kitagawa, Genshiro},
  title     = {Information Criteria and Statistical Modeling},
  publisher = {Springer},
  address   = {New York},
  year      = {2008},
  series    = {Springer Series in Statistics}
}

@article{FraleyRaftery1998,
  author  = {Fraley, Chris and Raftery, Adrian E.},
  title   = {How Many Clusters? Which Clustering Method? Answers via Model-Based Cluster Analysis},
  journal = {The Computer Journal},
  year    = {1998},
  volume  = {41},
  number  = {8},
  pages   = {578--588},
  doi     = {10.1093/comjnl/41.8.578}
}

@article{wu2025nonparametric,
  title        = {Nonparametric MLE for Gaussian Location Mixtures: Certified Computation and Generic Behavior},
  author       = {Wu, Yihong and Polyanskiy, Yury},
  year         = {2025},
  journal      = {arXiv:2503.20193},
  url          = {https://arxiv.org/abs/2503.20193},
}

@article{ghosal2001entropies,
  title   = {Entropies and Rates of Convergence for Maximum Likelihood and Bayes Estimation for Mixtures of Normal Densities},
  author  = {Ghosal, Subhashis and van der Vaart, Aad W.},
  journal = {Ann. Stat.},
  year    = {2001},
  volume  = {29},
  number  = {5},
  pages   = {1233--1263},
  doi     = {10.1214/aos/1013699999}
}

@article{hou2025repro,
    title={Repro Samples Method for Model-Free Inference in High-Dimensional Binary Classification},
    author={Xiaotian Hou and Peng Wang and Minge Xie and Linjun Zhang},
    year={2025},
    journal={arXiv:2510.01468}
}

@article{Awan03072025,
author = {Jordan Awan and Zhanyu Wang},
title = {Simulation-Based, Finite-Sample Inference for Privatized Data},
journal = {Journal of the American Statistical Association},
volume = {120},
number = {551},
pages = {1669--1682},
year = {2025},
publisher = {Taylor \& Francis},
doi = {10.1080/01621459.2024.2427436},


URL = { 
    
        https://doi.org/10.1080/01621459.2024.2427436
    
    

},
eprint = { 
    
        https://doi.org/10.1080/01621459.2024.2427436
    
    

}

}



\end{document}